\documentclass[lettersize,journal]{IEEEtran}
\usepackage{amsmath,amsfonts}
\usepackage{array}
\usepackage[caption=false,font=normalsize,labelfont=sf,textfont=sf]{subfig}
\usepackage{textcomp}
\usepackage{stfloats}
\usepackage{url}
\usepackage{verbatim}
\usepackage{graphicx}
\usepackage{cite}
\usepackage{multirow}
\usepackage[printonlyused]{acronym}
\newacro{ris} [RIS] {reconfigurable intelligent surface}
\newacro{mimo} [MIMO] {multiple-input multiple-output}
\newacro{mc} [MC] {mutual coupling}
\newacro{tx} [Tx] {transmitter}
\newacro{rx} [Rx] {receiver}
\newacro{bs} [BS] {base station}
\newacro{ue} [UE] {user equipment}
\newacro{em} [EM] {electromagnetic}
\newacro{aod} [AoD] {angle-of-departure}
\newacro{aoa} [AoA] {angle-of-arrival}
\newacro{los} [LoS] {line-of-sight}
\newacro{nlos} [NLoS] {non-line-of-sight}
\newacro{rfc} [RFC] {radio frequency chain}
\newacro{awgn} [AWGN] {additive white Gaussian noise}
\newacro{arv} [ARV] {array response vector}
\newacro{tdd} [TDD] {time division duplexing}
\newacro{upa} [UPA] {uniform planar array}
\newacro{ula} [ULA] {uniform linear array}
\newacro{simo} [SIMO] {single-input multiple-output}
\newacro{cs}[CS]{compressed sensing}
\newacro{mm}[MM]{Majorization-Minimization}
\newacro{sca}[SCA]{successive convex approximation}
\newacro{kkt}[KKT]{Karush-Kuhn-Tucker}
\newacro{snr}[SNR]{signal-to-noise ratio}
\newacro{omp}[OMP]{orthogonal matching pursuit}
\newacro{dr}[DR]{dictionary reduction}
\newacro{nmse}[NMSE]{normalized mean squared error}
\newacro{ghz}[GHz]{gigahertz}
\newacro{mmwave}[mmWave]{millimeter-wave}
\newacro{thz}[THz]{terahertz}
\newacro{gd}[GD]{gradient descent}
\newacro{svd}[SVD]{singular value decomposition}
\newacro{csi}[CSI]{channel state information}
\usepackage{color} 
\newcommand{\red}[1]{{\color{red}{#1}}}

\usepackage[table]{xcolor}
\usepackage{makecell}
\usepackage{hhline}

\newcommand{\hthickline}{\noalign{\hrule height 0.9pt}}
\usepackage{algorithm}
\usepackage{algpseudocode}% http://ctan.org/pkg/algorithmicx
\algtext*{EndWhile}% Remove "end while" text
\algtext*{EndIf}% Remove "end if" text
\algtext*{EndFor}
\usepackage{tikz} 
\usepackage[utf8]{inputenc}
\usepackage{pgfplots} 
\usetikzlibrary {arrows.meta}
\makeatletter
\newcommand{\gettikzxy}[3]{%
  \tikz@scan@one@point\pgfutil@firstofone#1\relax
  \edef#2{\the\pgf@x}%
  \edef#3{\the\pgf@y}%
}
\usepackage{cases} % for case environment with individual labels
\usepackage{units} % use \unit command
\usepackage{bm} % for bold symbol
\usepackage{amssymb} % for AMS symbols

\newtheorem{proposition}{\textbf{Proposition}}

\newtheorem{remark}{\textbf{Remark}}

\newtheorem{condition}{\textbf{Condition}}
\newenvironment{proof}{\textit{\textbf{Proof:}}}{\hfill$\square$}
\makeatletter
\algnewcommand{\LineComment}[1]{\Statex \hskip\ALG@thistlm \(\triangleright\) #1}
\makeatother

\begin{document}
% transpose and hermitian
\newcommand{\TT}{\mathsf{T}}
\newcommand{\HH}{\mathsf{H}}

\newcommand{\nth}[1]{{#1}{\text{th}}}
\newcommand{\norm}[1]{\left\|{#1}\right\|}
\newcommand{\diagopr}[1]{\mathrm{diag}\left(#1\right)}
\newcommand{\invdiag}[1]{\mathrm{diag}^{-1}\left(#1\right)}
\newcommand{\vect}[1]{\mathrm{vec}\left(#1\right)}
\newcommand{\invvec}[1]{\mathrm{vec}^{-1}\left(#1\right)}
\newcommand{\abs}[1]{\left|{#1}\right|}

\newcommand{\err}{\mathrm{E}}
\newcommand{\strmidx}{\mathrm{S}}
\newcommand{\RFidx}{\mathrm{RF}}
\newcommand{\BBidx}{\mathrm{BB}}
\newcommand{\quantidx}{\mathrm{quant}}
\newcommand{\aeidx}{\mathrm{AE}}
\newcommand{\azidx}{\mathrm{az}}
\newcommand{\elidx}{\mathrm{el}}

\newcommand{\ula}{\mathrm{ULA}}
\newcommand{\upa}{\mathrm{UPA}}

\newcommand{\casidx}{\mathrm{cas}}
\newcommand{\mcidx}{\mathrm{mc}}
\newcommand{\convidx}{\mathrm{cv}}

\newcommand{\ovs}{\mathrm{S}}
\newcommand{\DicRed}{\mathrm{DR}}

\newfont{\bb}{msbm10 scaled 1100}
\newcommand{\PP}{\mbox{\bb P}}
\newcommand{\EE}{\mbox{\bb E}}
% Vectors
\newcommand{\av}{{\bf a}}
\newcommand{\bv}{{\bf b}}
\newcommand{\cv}{{\bf c}}
\newcommand{\dv}{{\bf d}}
\newcommand{\ev}{{\bf e}}
\newcommand{\fv}{{\bf f}}
\newcommand{\gv}{{\bf g}}
\newcommand{\hv}{{\bf h}}
\newcommand{\iv}{{\bf i}}
\newcommand{\jv}{{\bf j}}
\newcommand{\kv}{{\bf k}}
\newcommand{\lv}{{\bf l}}
\newcommand{\mv}{{\bf m}}
\newcommand{\nv}{{\bf n}}
\newcommand{\ov}{{\bf o}}
\newcommand{\pv}{{\bf p}}
\newcommand{\qv}{{\bf q}}
\newcommand{\rv}{{\bf r}}
\newcommand{\sv}{{\bf s}}
\newcommand{\tv}{{\bf t}}
\newcommand{\uv}{{\bf u}}
\newcommand{\wv}{{\bf w}}
\newcommand{\vv}{{\bf v}}
\newcommand{\xv}{{\bf x}}
\newcommand{\yv}{{\bf y}}
\newcommand{\zv}{{\bf z}}
\newcommand{\zerov}{{\bf 0}}
\newcommand{\onev}{{\bf 1}}
\newcommand{\avr}{\av_\text{R}}
% Matrices
\newcommand{\Am}{{\bf A}}
\newcommand{\Bm}{{\bf B}}
\newcommand{\Cm}{{\bf C}}
\newcommand{\Dm}{{\bf D}}
\newcommand{\Em}{{\bf E}}
\newcommand{\Fm}{{\bf F}}
\newcommand{\Gm}{{\bf G}}
\newcommand{\Hm}{{\bf H}}
\newcommand{\Id}{{\bf I}}
\newcommand{\Jm}{{\bf J}}
\newcommand{\Km}{{\bf K}}
\newcommand{\Lm}{{\bf L}}
\newcommand{\Mm}{{\bf M}}
\newcommand{\Nm}{{\bf N}}
\newcommand{\Om}{{\bf O}}
\newcommand{\Pm}{{\bf P}}
\newcommand{\Qm}{{\bf Q}}
\newcommand{\Rm}{{\bf R}}
\newcommand{\Sm}{{\bf S}}
\newcommand{\Tm}{{\bf T}}
\newcommand{\Um}{{\bf U}}
\newcommand{\Wm}{{\bf W}}
\newcommand{\Vm}{{\bf V}}
\newcommand{\Xm}{{\bf X}}
\newcommand{\Ym}{{\bf Y}}
\newcommand{\Zm}{{\bf Z}}
\newcommand{\Onem}{{\bf 1}}
\newcommand{\Zerom}{{\bf 0}}
% text uppercase
\newcommand{\At}{{\rm A}}
\newcommand{\Ct}{{\rm C}}
\newcommand{\Dt}{{\rm D}}
\newcommand{\Et}{{\rm E}}
\newcommand{\Ft}{{\rm F}}
\newcommand{\Gt}{{\rm G}}
\newcommand{\Ht}{{\rm H}}
\newcommand{\It}{{\rm I}}
\newcommand{\Jt}{{\rm J}}
\newcommand{\Kt}{{\rm K}}
\newcommand{\Lt}{{\rm L}}
\newcommand{\Mt}{{\rm M}}
\newcommand{\Nt}{{\rm N}}
\newcommand{\Ot}{{\rm O}}
\newcommand{\Pt}{{\rm P}}
\newcommand{\Qt}{{\rm Q}}
\newcommand{\St}{{\rm S}}
\newcommand{\Wt}{{\rm W}}
\newcommand{\Vt}{{\rm V}}
\newcommand{\Xt}{{\rm X}}
\newcommand{\Yt}{{\rm Y}}
\newcommand{\Zt}{{\rm Z}}

% SIMON modifications
\newcommand{\Rt}{{\rm B}}
\newcommand{\Tt}{{\rm U}}
\newcommand{\Bt}{{\rm R}}
\newcommand{\Ut}{{\rm T}}

% Bold greek letters
\newcommand{\alphav}{\hbox{\boldmath$\alpha$}}
\newcommand{\betav}{\hbox{\boldmath$\beta$}}
\newcommand{\gammav}{\hbox{\boldmath$\gamma$}}
\newcommand{\deltav}{\hbox{\boldmath$\delta$}}
\newcommand{\etav}{\hbox{\boldmath$\eta$}}
\newcommand{\lambdav}{\hbox{\boldmath$\lambda$}}
\newcommand{\kappav}{\hbox{\boldmath$\kappa$}}
\newcommand{\epsilonv}{\hbox{\boldmath$\epsilon$}}
\newcommand{\nuv}{\hbox{\boldmath$\nu$}}
\newcommand{\muv}{\hbox{\boldmath$\mu$}}
\newcommand{\zetav}{\hbox{\boldmath$\zeta$}}
\newcommand{\phiv}{\hbox{\boldmath$\phi$}}
\newcommand{\varphiv}{\hbox{\boldmath$\varphi$}}
\newcommand{\psiv}{\hbox{\boldmath$\psi$}}
\newcommand{\thetav}{\hbox{$\boldsymbol\theta$}}
\newcommand{\varthetav}{\hbox{$\boldsymbol\vartheta$}}
\newcommand{\tauv}{\hbox{\boldmath$\tau$}}
\newcommand{\omegav}{\hbox{\boldmath$\omega$}}
\newcommand{\xiv}{\hbox{\boldmath$\xi$}}
\newcommand{\sigmav}{\hbox{\boldmath$\sigma$}}
\newcommand{\piv}{\hbox{\boldmath$\pi$}}
\newcommand{\rhov}{\hbox{\boldmath$\rho$}}

\newcommand{\Gammam}{\hbox{\boldmath$\Gamma$}}
\newcommand{\Lambdam}{\hbox{\boldmath$\Lambda$}}
\newcommand{\Deltam}{\hbox{\boldmath$\Delta$}}
\newcommand{\Sigmam}{\hbox{\boldmath$\Sigma$}}
\newcommand{\Phim}{\hbox{\boldmath$\Phi$}}
\newcommand{\Pim}{\hbox{\boldmath$\Pi$}}
\newcommand{\Psim}{\hbox{\boldmath$\Psi$}}
\newcommand{\psim}{\hbox{\boldmath$\psi$}}
\newcommand{\chim}{\hbox{\boldmath$\chi$}}
\newcommand{\omegam}{\hbox{\boldmath$\omega$}}
\newcommand{\Thetam}{\hbox{\boldmath$\Theta$}}
\newcommand{\Omegam}{\hbox{\boldmath$\Omega$}}
\newcommand{\Xim}{\hbox{\boldmath$\Xi$}}
 % include macro definition for mathematical notations
\bstctlcite{IEEEexample:BSTcontrol}

% ------------------------------------------------------------------------------
\title{Mutual Coupling-Aware Channel Estimation and Beamforming for RIS-Assisted Communications}

\author{Pinjun~Zheng, Simon~Tarboush, Hadi~Sarieddeen, and Tareq~Y.~Al-Naffouri \vspace{-2em}

\thanks{P. Zheng and S. Tarboush are \emph{co-first authors}; they contributed equally to this paper. P. Zheng is with the University of British Columbia, Canada (e-mail: pinjun.zheng@ubc.ca). S.~Tarboush and  T.~Y.~Al-Naffouri are with the King Abdullah University of Science and Technology (KAUST), Saudi Arabia (Email: simon.w.tarboush@gmail.com; tareq.alnaffouri@kaust.edu.sa). H. Sarieddeen is with the American University of Beirut (AUB), Lebanon (hadi.sarieddeen@aub.edu.lb). The majority of P. Zheng's contributions to this work were made during his Ph.D. studies at KAUST. This work was supported by the KAUST Office of Sponsored Research (OSR) under Award No. ORFS-CRG12-2024-6478 and the AUB University Research Board.}}

\markboth{This work has been submitted to the IEEE for possible publication.  Copyright may be transferred without notice.}{draft}
%\markboth{draft}{draft}
\maketitle
% ------------------------------------------------------------------------------

\begin{abstract} 
This work studies the problems of channel estimation and beamforming for active reconfigurable intelligent surface~(RIS)-assisted multiple-input multiple-output (MIMO) communication, incorporating the mutual coupling~(MC) effect through an electromagnetically consistent model. We first demonstrate that MC can be incorporated into a compressed sensing~(CS) formulation, albeit with an increase in the dimensionality of the sensing matrix. To overcome this increased complexity, we propose a two-stage strategy. Initially, a low-complexity MC-unaware CS estimation is performed to obtain a coarse channel estimate, which is then used to implement a dictionary reduction (DR) for the MC-aware estimation, effectively reducing the dimensionality of the sensing matrices.  {This method achieves estimation accuracy close to the direct MC-aware CS method with less overall computational complexity.} Furthermore, we consider the joint optimization of RIS configuration, base station precoding, and user combining in a single-user MIMO system. We employ an alternating optimization strategy to optimize these three beamformers. The primary challenge lies in optimizing the RIS configuration, as the MC effect renders the problem non-convex and intractable. To address this, we propose a novel algorithm based on the successive convex approximation (SCA) and the Neumann series expansion. Within the SCA framework, we propose a surrogate function that rigorously satisfies both convexity and equal-gradient conditions to update the iteration direction. Numerical results validate our proposal, demonstrating that the proposed channel estimation and beamforming methods effectively manage the MC in RIS, achieving higher spectral efficiency compared to state-of-the-art approaches.
\end{abstract}


\section{Introduction}

Next-generation wireless networks aim to enhance spectral and energy efficiency, necessitating the integration of game-changing enablers. To access abundant bandwidth, these networks are shifting to higher frequency bands, including \ac{mmwave}~\cite{Wang2018Millimeter} and \ac{thz}~\cite{Tarboush2021Teramimo}. However, higher frequencies present challenges, such as a higher likelihood of blockage, which reduces link availability. To mitigate these effects, the integration of \ac{ris}~\cite{Di2020Smart} offers a viable solution. An \ac{ris} is a planar metamaterial-based surface composed of nearly passive reflecting elements. Furthermore, active~\acp{ris}, which incorporate reflection-type amplifiers~\cite{Rao2023An}, have been proposed to alleviate the multiplicative fading effect and overcome the capacity limitations of passive~\acp{ris} \cite{Zhang2022Active}. The dense and holographic~\acp{ris} have gained attention due to their flexibility and precision in shaping \ac{em} waves with high energy efficiency~\cite{Wan2021Terahertz,Bjornson2024Towards}. Although promising, these novel \ac{ris} designs introduce new challenges, particularly the \ac{mc} effect. Existing studies suggest that densifying the \ac{ris} layout strengthens the \ac{mc} among adjacent unit cells~\cite{Gradoni2021End,Zheng2024On}. Moreover, increasing the \ac{ris} amplification coefficients can further exacerbate \ac{mc} effects, degrading channel estimation and localization accuracy~\cite{Zheng2023JrCUP}. Therefore, accurate modeling and sophisticated signal processing are crucial to optimize \ac{ris} communication performance and fully leverage the potential benefits of its \ac{em} interactions.

\subsection{Related Work}
\subsubsection{Conventional RIS Modeling} 
The commonly adopted model for \ac{ris}-assisted communication describes \ac{ris} reflections as a linear cascaded channel, where the overall channel is expressed as the product of the \ac{tx}-\ac{ris} subchannel, the \ac{ris} reflection response, and the \ac{ris}-\ac{rx} subchannel~\cite{Pan2022An,Zheng2023JrCUP}. While mathematically tractable, this linear model tacitly assumes that each \ac{ris} unit cell radiates \ac{em} waves independently, neglecting the nonlinear coupling effects among adjacent unit cells. Recent findings in~\cite{Rabault2024Tacit} cast doubt on this widely used linear cascaded model, suggesting it poorly describes the physical reality. As analyzed in~\cite{Rabault2024Tacit}, the \ac{mc} can result from the proximity-induced coupling due to adjacent radiators and the reverberation-induced long-range coupling due to environmental scattering multipath. Both mechanisms are nonlinear. The neglect of proximity-induced coupling becomes untenable when considering densely integrated \ac{ris}. Meanwhile, overlooking reverberation-induced long-range coupling presents challenges in rich-scattering environments.

\subsubsection{\ac{mc}-Aware RIS Modeling}
Notable efforts have been dedicated to accurately accounting for \ac{mc} in RIS-assisted communications. For instance, based on coupled-dipole formalism, a physics-based end-to-end model called PhysFad~\cite{Faqiri2023PhysFad} has been proposed. This model provides a tuning mechanism for jointly modeling transceivers, RIS elements, and the scattering environment. Another suitable theory for \ac{mc} modeling is the microwave multiport network theory~\cite{Pozar2011Microwave,Ivrlavc2014Multiport}, which extends basic circuit and network concepts to handle complex microwave analyses. Typically, three equivalent representations can be used to analyze microwave networks: impedance, admittance, and scattering parameters~\cite{Nerini2024Universal}.
An impedance matrix (Z-parameters) model of~\ac{mc} was first adopted in~\cite{Gradoni2021End}. Subsequently, equivalent models based on the scattering matrix (S-parameters) have been developed~\cite{Li2024Beyond}. These models are essentially equivalent and have been unified in, e.g.,~\cite{Abrardo2024Design,Nerini2024Universal}. 
From a practical perspective, the scattering matrix model has the advantage of being more directly related to the radiation pattern~\cite{Wijekoon2024Phase} and is experimentally verifiable, as demonstrated in our previous~\ac{mc} measurement work~\cite{Zheng2024Mutual}.  {While this paper focuses on the scattering matrix model, numerous other electromagnetic models exist. These include, but are not limited to, the local reflection coefficient model, homogenized and inhomogenized impedance models, periodic discrete models, and full-wave electromagnetic
models~\cite{Bidabadi2025Physically}.}

\subsubsection{Channel Estimation}
Considering the sparse channel nature in the high-frequency bands~\cite{Wang2018Millimeter,Tarboush2021Teramimo}, earlier investigations in~\cite{Wang2020Compressed}, based on the conventional cascaded model, formulated the channel estimation as a sparse signal recovery problem, which can be solved using~\ac{cs} techniques. In~\cite{Wei2021Channel} and~\cite{Chen2023Channel},~\ac{cs}-based channel estimation was addressed in multi-user scenarios based on the conventional linear model by leveraging double-structured sparsity. Other works, such as~\cite{He2021Channel, Ardah2021Trice}, proposed two-stage channel estimation methods that decouple the cascaded channel into two separate subproblems. However, none of these previous studies account for \ac{mc}. Prior information can be incorporated with~\ac{cs} tools to significantly enhance the estimation performance, such as weighted~\ac{cs} algorithms, structured codebook design, and dictionary design, as shown in~\cite{Tarboush2023Compressive,Tarboush2024Cross}, which has not been investigated for RIS channel estimation problems.

Few works have addressed the \ac{ris}-assisted channel estimation problem with \ac{mc}. An evaluation of channel estimation bounds in~\cite{Zheng2024On} indicates that denser \ac{ris} unit cell integration substantially degrades estimation accuracy. Therefore, \ac{mc} awareness is emerging as a critical concern in formulating channel estimators, particularly for systems with densely integrated RIS designs. Furthermore,~\cite{Bayraktar2024RIS} suggests that dictionary learning is an effective tool for handling \ac{ris} \ac{mc} in channel estimation when a linear \ac{mc} model is considered.

\subsubsection{Beamforming}
The \ac{ris} beamforming based on the conventional linear model has been extensively studied in the literature. Various approaches have been applied to optimize the~\ac{ris} phase shifts, such as semidefinite relaxation~\cite{Wu2019Intelligent},~\ac{gd}~\cite{Huang2019Reconfigurable}, alternating direction method of multipliers~\cite{Liu2022Joint}, and majorization-minimization~\cite{Huang2019Reconfigurable,Liu2022Joint}, among others. More practical designs have also been considered, including dual reflection phase and amplitude variations~\cite{Li2021Intelligent} and designs that account for imperfect \ac{csi}~\cite{Yu2020Robust}. Additionally, active~\ac{ris} optimization based on the conventional model has been addressed to improve the sum-rate~\cite{Chen2024Enhancing} and to account for hardware impairments~\cite{Peng2024Beamforming}.

The study of \ac{mc} awareness in \ac{ris} beamforming is also in its early stages. 
By properly adopting the impedance matrix model accounting for the~\ac{mc} and optimizing the~\ac{ris} phase shift accordingly,~\cite{Qian2021Mutual} demonstrated an enhanced end-to-end~\ac{snr}, while~\cite{Abrardo2021MIMO} obtained an improved sum-rate for multi-user interference channels. A joint active and passive beamforming problem was further studied in~\cite{Wijekoon2024Phase} for \ac{ris}-assisted downlink and uplink transmissions based on the scattering matrix model. Nonetheless, these works do not consider active~\acp{ris}, which possess a higher amplification factor and will accentuate the~\ac{mc} effect~\cite{Zheng2023JrCUP}. Additionally, more effective beamforming algorithms are needed for the scattering matrix-based model. For example, some up-to-date works utilize the~\ac{gd} algorithm to optimize the RIS coefficients based on the scattering matrix model~\cite{Wijekoon2024Phase}. Although decent performance is achieved, the convergence of~\ac{gd} to the global minimum is not always guaranteed in such non-convex problems, which deserves further investigation.  

\subsection{Contributions}
To the best of the authors’ knowledge, this is the first work that comprehensively considers both channel estimation and beamforming for an active~\ac{ris}-assisted communication in the presence of~\ac{mc} among RIS unit cells. We consider a \ac{mimo} communication system assisted by an active \ac{ris}, which works in an \emph{uplink} channel estimation and \emph{downlink} communication paradigm. We first formulate and solve the uplink channel estimation problem using the electromagnetic-consistent channel model based on the scattering parameter representation, which takes the RIS~\ac{mc} into account. Based on the estimated uplink channel, we first derive the downlink channel and then jointly optimize RIS configuration, \ac{bs} precoder, and \ac{ue} combiner for the downlink communication, again considering the exact \ac{mc}-aware model.  
The main contributions of this paper are summarized as follows:
\begin{itemize}
 \item \textit{We formulate and solve the RIS-assisted channel estimation problem in the presence of \ac{mc}.} We first demonstrate that the nonlinearity introduced by RIS \ac{mc} undermines the conventional \ac{cs} formulation for channel estimation, rendering existing \ac{cs}-based estimators infeasible. By carefully manipulating the observation format based on the \ac{mc} mechanism, we propose a new \ac{cs} formulation that accounts for \ac{mc} without any loss of accuracy. Subsequently, we apply the \ac{omp} algorithm~\cite{Lee2016Channel} to solve the new \ac{cs} problem. At this stage, similar to~\cite{Wang2020Compressed}, we can estimate an equivalent cascaded channel, which is a coupled Kronecker product of the UE-RIS and RIS-BS subchannels. While the two individual subchannels cannot be separated, we later demonstrate that this equivalent cascaded channel is sufficient for beamforming design.
 \item \textit{We propose a \ac{dr} strategy to further reduce the complexity of the proposed \ac{mc}-aware channel estimator by utilizing prior estimation information.} While the constructed \ac{cs} problem accurately accounts for \ac{mc}, it incurs a dimension lift in the basis space. Specifically, by incorporating \ac{mc}, the dimension of atoms in the \ac{cs} dictionary increases from~$N_\mathrm{B} N_\mathrm{I}$ in the conventional solution to~$N_\mathrm{B} N_\mathrm{I}^2$ (here $N_\mathrm{B} $ and $N_\mathrm{I}$ respectively denote the sizes of the \ac{bs} and \ac{ris} arrays), hindering the feasibility of our method in large-scale RISs. To address this issue, we propose a two-stage strategy to achieve low-complexity yet accurate estimation. This approach first performs the conventional \ac{mc}-unaware \ac{cs} estimation to obtain a coarse estimate. Based on this initial estimate, we then design a \ac{dr} procedure to reduce the dictionary size using the exact \ac{mc}-aware model. This \ac{dr} approach can effectively counteract the dimension lift due to the \ac{mc} effect.
 \item \textit{We formulate and solve the joint \ac{ris} configuration, \ac{bs} precoder, and \ac{ue} combiner optimization problem based on the estimated equivalent cascaded channel.} To tackle this complex joint optimization problem, we adopt an alternating optimization strategy, optimizing the three beamformers alternately. For the \ac{bs} and \ac{ue} design, a simple closed-form optimal solution is derived. For the \ac{ris} configuration optimization, the nonlinear \ac{mc} mechanism introduces challenges of intractability and non-convexity. We first apply the Neumann series expansion to approximate the objective function and reformulate it as a tractable optimization problem. Since the resultant problem remains non-convex, we adopt the \ac{sca} framework to solve it. 
 \item \textit{We validate our methods through extensive numerical simulations.} In terms of channel estimation accuracy, the proposed estimation method, even with high~\ac{mc} levels, outperforms the \textit{conventional model}-based strategy by several dBs, and has a similar accuracy to the \textit{exact model}-based procedure but with less complexity. Moreover, the proposed~\ac{mc}-aware joint beamforming algorithm surpasses the state-of-the-art benchmarks, achieving enhanced spectral efficiency. 
\end{itemize} 

\subsection{Organization and Notation}
The remainder of this paper is organized as follows. Section~\ref{sec:Ch_Model} introduces the uplink channel model, while Section~\ref{sec:rx_sig_mc} formulates the \ac{mc}-aware uplink channel estimation problem. Section~\ref{sec:CE_MC} details the proposed low-complexity channel estimation solution. Subsequently, we define the downlink beamforming problem and outline the alternating optimization strategy in Section~\ref{sec:JBF}. Given that the BS precoder and UE combiner are addressed with a closed-form solution immediately following the problem formulation, we delve into the RIS configuration optimization in Section~\ref{sec:RISCO}. The simulation results are presented in Section~\ref{sec:Sim_Res_Disc}, and the conclusion is given in Section~\ref{sec:Conc}. 

We use the following notation throughout the paper. Non-bold lower and upper case letters (e.g., $a, A$) denote scalars, bold lower case letters (e.g., $\av$) denote vectors, and bold upper case letters (e.g., $\Am$) denote matrices. We use~$A_{i,j}$ to denote the entry at the~$i$$^\text{th}$ row and~$j$$^\text{th}$ column of the matrix~$\Am$. In addition,~$[\Am]_{n,:}$ and~$[\Am]_{:,m}$ denote $n$$^\text{th}$ row and $m$$^\text{th}$ column of~$\Am$, respectively. The superscripts ${(\cdot)}^\TT$, ${(\cdot)}^*$, ${(\cdot)}^\HH$, ${(\cdot)}^{-1}$, and~${(\cdot)}^\dagger$ represent the transpose, conjugate, Hermitian (conjugate transpose), inverse, and pseudo-inverse operators, respectively. For two $M\times N$ matrices $\Am$ and $\Bm$, $\Am\otimes\Bm$ denotes the $M^2\times N^2$ Kronecker product matrix. Furthermore,  {$\Am\bullet\Bm$ denotes the $M\times N^2$ row-wise Khatri-Rao product matrix};~$\odot$ denotes the Hadamard product.

\section{Uplink Channel Model}
\label{sec:Ch_Model}

\begin{figure}[t]
  \centering
  \includegraphics[width=\linewidth]{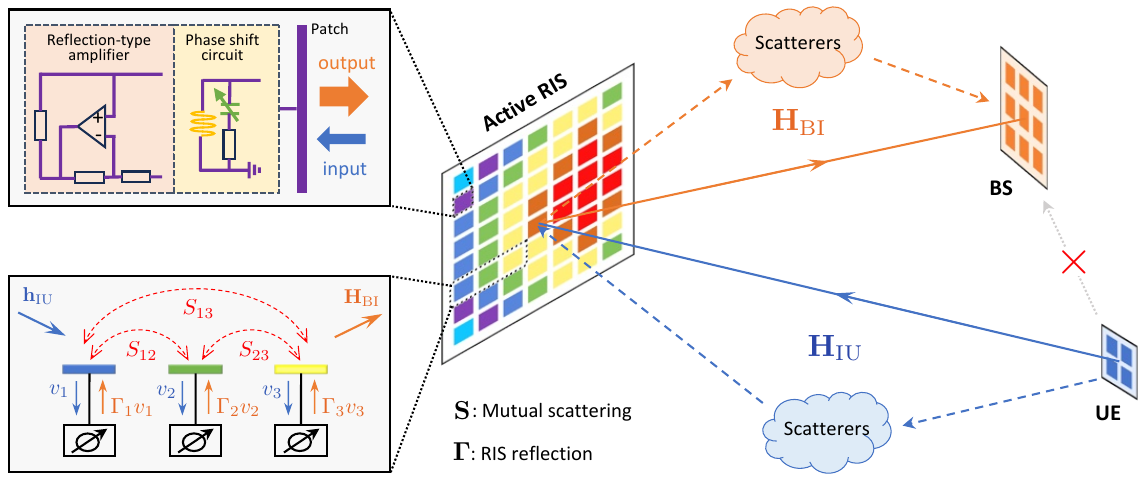}
  \caption{ 
      Illustration of a \ac{ris}-assisted uplink \ac{mimo} communication system in the presence of \ac{mc} between RIS unit cells, where the UE-BS direct channel is assumed to be absent.
    }
  \label{fig:system}
  \vspace{-2.5mm}
\end{figure}

Consider an \ac{ris}-assisted time division duplex \ac{mimo} communication system, consisting of an~$N_\Tt$-antenna \ac{ue}, an~$N_\Rt$-antenna \ac{bs}, and an~$N_\It$-unit cell \ac{ris}, as shown in Fig.~\ref{fig:system}. To focus on RIS reflection, we assume the \ac{ue}-\ac{bs} direct channel does not exist. In this work, the \ac{ris} is composed of a \ac{upa} consisting of $N_\It=N_\It^\mathrm{h}\times N_\It^\mathrm{v}$ tightly-packed unit cells, where the superscripts `$\mathrm{h}$' and `$\mathrm{v}$' denote the horizontal and vertical dimensions, respectively. The same assumption is applied for the \ac{bs} and \ac{ue}, i.e., $N_\Rt=N_\Rt^\mathrm{h}\times N_\Rt^\mathrm{v}$ and $N_\Tt = N_\Tt^\mathrm{h}\times N_\Tt^\mathrm{v}$. 

We consider a comprehensive uplink channel estimation and downlink beamforming problem. Specifically, channel estimation is performed for the uplink, and the downlink channels are inferred via channel reciprocity~\cite{Chen2023Channel}, which enables downlink beamforming. To facilitate the subsequent algorithm development, this section models the \ac{ue}-\ac{ris}-\ac{bs} uplink channel. In the following subsections, we first detail the individual \ac{ue}-\ac{ris} and \ac{ris}-\ac{bs} subchannels, and then elaborate on how the \ac{mc} between RIS unit cells reshapes the overall cascaded channel.

\subsection{The Channel Between UE and RIS}
Assuming a far-field propagation scenario, the high-frequency narrow-band frequency-domain channel from the \ac{ue} to the \ac{ris}, $\Hm_{\mathrm{IU}}\in\mathbb{C}^{N_\It\times N_\Tt}$, is defined as~\cite{Wei2021Channel}
\begin{equation}
	\label{eq:ch_TxRIS_spatial_domain}
 \Hm_{\It\Tt} = \sqrt{\frac{N_\It N_\Tt}{L_\Tt}}\sum_{\ell=1}^{L_{\Tt}} \alpha_\ell\av_\It(\phiv_\ell)\av^\TT_\Tt(\varphiv_\ell),
\end{equation}
where~$\alpha_\ell\in\mathbb{C}$ denotes the complex channel gain,~$\varphiv_\ell\in\mathbb{R}^2$ the \ac{aod} at the \ac{ue},~$\phiv_\ell\in\mathbb{R}^2$ the \ac{aoa} at the \ac{ris}, and $\av_\Tt(\varphiv_\ell)\in\mathbb{C}^{N_\Tt}$ and $\av_\It(\phiv_\ell)\in\mathbb{C}^{N_\It}$ are the \ac{arv} at the \ac{ue} and \ac{ris} corresponding to $\varphiv_\ell$ and $\phiv_\ell$, respectively. Note that each \ac{aoa} and \ac{aod} comprises an azimuth and an elevation component, e.g., $\phiv_\ell=[\phi_\ell^\azidx,\phi_\ell^\elidx]^\TT$.\footnote{We define the elevation angle as the angle between the direction of interest and the positive~$\mathrm{Z}$-axis, which is also called the inclination angle.} Under a \ac{upa} setup, the \ac{arv} at the \ac{ris} can be expressed as~\cite{Tarboush2023Compressive}
\begin{equation}
    \label{eq:UPA_ARV}
    \av_\It(\phiv_\ell)=\frac{1}{\sqrt{N_\It}}e^{-j2\pi\phi_\ell^\mathrm{h}\nv(N_\It^\mathrm{h})}\otimes e^{-j2\pi\phi_\ell^\mathrm{v}\nv(N_\It^\mathrm{v})},
\end{equation}
where~$\nv(N)=[0,1,\dots,N-1]^\TT$, and~$\phi_\ell^\mathrm{h}$ and~$\phi_\ell^\mathrm{v}$ are the spatial angles corresponding to the horizontal and vertical dimensions, respectively. Assuming the \ac{upa} is deployed on the YZ-plane of the \ac{ris}'s body coordinate system, we obtain $\phi_\ell^\mathrm{h}\triangleq{d_\It}\sin(\phi_\ell^\azidx)\sin(\phi_\ell^\elidx)/{\lambda}$ and $\phi_\ell^\mathrm{v}\triangleq{d_\It}\cos(\phi_\ell^\elidx)/{\lambda}$, where~$\lambda$ is the wavelength of the operating frequency, $f_c$, and $d_\It$ is the inter-element spacing of the \ac{ris}. In addition,~$L_{\Tt}$ in~\eqref{eq:ch_TxRIS_spatial_domain} is the number of paths between the \ac{ue} and \ac{ris}, where~$\ell=1$ stands for the \ac{los} path and~$\ell=2,\dots,L_\mathrm{U}$ correspond to $L_\mathrm{U}-1$ \ac{nlos} paths. 

By concatenating all \acp{arv},~\eqref{eq:ch_TxRIS_spatial_domain} can be written as~\cite{Wang2020Compressed,He2021Channel}
\begin{equation}
    \label{eq:ch_TxRIS_angle_domain}
   \Hm_{\It\Tt} =\Am_\It(\phiv)\Sigmam_{\It\Tt}\Am^\TT_\Tt(\varphiv),
\end{equation}
where we define $\phiv=[\phiv_1^\TT,\dots,\phiv_{L_\mathrm{U}}^\TT]^\TT\in\mathbb{R}^{2L_\mathrm{U}},\ \varphiv=[\varphiv_1^\TT,\dots,\varphiv_{L_\Tt}^\TT]^\TT\in\mathbb{R}^{2L_\Tt},\ \Am_\It(\phiv)=[\av_{\It}(\phiv_1),\cdots,\av_\It(\phiv_{L_\Tt})]\in\mathbb{C}^{N_\It\times L_\Tt},\ \Am_\Tt(\varphiv)=[\av_{\Tt}(\varphiv_1),\cdots,\av_\Tt(\varphiv_{L_\Tt})]\in\mathbb{C}^{N_\Tt\times L_\Tt}$, and $\Sigmam_{\It\Tt}=\sqrt{N_\It N_\Tt/{L_\Tt}} \diagopr{\alpha_1,\cdots,\alpha_{L_\Tt}}\in\mathbb{C}^{L_\Tt \times L_\Tt}$ is the beamspace channel matrix. 

\subsection{The Channel Between RIS and BS}
Similarly, the frequency-domain channel from the \ac{ris} to the \ac{bs},~$\Hm_{\Rt\It}\in\mathbb{C}^{N_\Rt \times N_\It}$, can be expressed as~\cite{Wei2021Channel}
\begin{equation}
    \label{eq:ch_RISRx_spatial_domain}
    \Hm_{\Rt\It} = \sqrt{\frac{N_\It N_\Rt}{L_{\Rt}}}\sum_{\ell=1}^{L_{\Rt}} \rho_\ell\av_\Rt(\varthetav_\ell)\av_\It^\TT(\thetav_\ell),
\end{equation}
where $L_{\Rt}$ is the number of paths between the \ac{ris} and \ac{bs},~$\rho_\ell\in\mathbb{C}$ denotes the complex channel gain,~$\thetav_\ell\in\mathbb{R}^2$ denotes the \ac{aod} at the \ac{ris},~$\varthetav_\ell\in\mathbb{R}^2$ denotes the \ac{aoa} at the \ac{bs},  and~$\av_\Rt\in\mathbb{C}^{N_\Rt}$ and~$\av_\It\in\mathbb{C}^{N_\It}$ denote the \acp{arv} of the antenna arrays at the \ac{bs} and \ac{ris}, respectively. We denote~$d_\Rt$ as the inter-element spacing of the \ac{bs}. Analogously, by concatenating all the \ac{bs}/\ac{ris} \ac{arv}s and the channel coefficients in a matrix form, the channel~$\mathbf{H}_\mathrm{BI}$ can be written as~\cite{Wang2020Compressed,He2021Channel}
\begin{equation}
\label{eq:ch_RISRx_angle_domain}    \Hm_{\Rt\It}=\Am_\Rt(\varthetav)\Sigmam_{\Rt\It}\Am_\It^\TT(\thetav),
\end{equation}
where we define~$\varthetav=[\varthetav_1^\TT,\dots,\varthetav_{L_\mathrm{B}}^\TT]^\TT\in\mathbb{R}^{2L_\mathrm{B}},\ \thetav=[\thetav_1^\TT,\dots,\thetav_{L_\mathrm{B}}^\TT]^\TT\in\mathbb{R}^{2L_\mathrm{B}},\ \Am_\Rt(\varthetav)=[\av_{\Rt}(\varthetav_1),\cdots,\av_\Rt(\varthetav_{L_\Rt})]\!\in\!\mathbb{C}^{N_\Rt\times L_\Rt},\ \Am_\It(\thetav)=[\av_{\It}(\thetav_1),\cdots,\av_\It(\thetav_{L_\Rt})]\in\mathbb{C}^{N_\It\times L_\Rt}$, and the beamspace channel matrix~$\Sigmam_{\Rt\It}=\sqrt{N_\Rt N_\It/{L_\Rt}}\diagopr{\rho_1,\cdots,\rho_{L_\Rt}}\in\mathbb{C}^{L_\Rt\times L_\Rt}$.

\subsection{Conventional Cascaded Channel Model}

We elaborate on the channel model with the uplink training stage as the background. During training, the active \ac{ris} employs a set of different amplitude and phase configurations to reflect signals from the \ac{ue} to the \ac{bs}, which leads to a different cascaded channel matrix for different transmissions (more details in Section~\ref{sec:rx_sig}).\footnote{In the downlink data transmission stage, the RIS response is set to a fixed optimized configuration determined through the beamforming process, as will be presented in Section~\ref{sec:JBF} and Section~\ref{sec:RISCO}.} Since we have assumed the \ac{ue}-\ac{bs} direct link to be absent, the conventional \ac{ris}-assisted cascaded channel model corresponding to the ${m_\It}^\text{th}$ training configuration---without accounting for \ac{mc} between \ac{ris} unit cells---can be written as~\cite{Pan2022An,Wang2020Compressed,He2021Channel,Zheng2023JrCUP,Chen2023Channel,Cheng2024Degree,Zheng2024LEO}
\begin{equation}
\label{eq:convCHM}
    \Hm^{m_\It}_\convidx = \Hm_{\Rt\It} \Gammam_{m_\It} \Hm_{\It\Tt},
\end{equation}
where $\Gammam_{m_\It}=\diagopr{\bm{\gamma}_{m_\It}}$, and~$\bm{\gamma}_{m_\It}\in\mathbb{C}^{N_\It}$ is the vector of \ac{ris} reflection coefficients at the ${m_\It}^\text{th}$ training configuration, $m_\It=1,2,\dots,M_\It$, and $M_\It$ denotes the number of \ac{ris} configurations during the training. The reflection coefficient of the~$i$$^\text{th}$ \ac{ris} unit cell can be expressed as~$\gamma_{m_\It,i}=a_ie^{j\vartheta_i}$, where~$a_i\in\mathbb{R}^+$ represents the amplification factor and~$\vartheta_i\in(0,2\pi]$ denotes the phase shift. Both~$a_i$ and~$\vartheta_i$ are reconfigurable; however, they are usually dependent on each other~\cite{Wang2024Wideband}. For nearly-passive \acp{ris}, $a_i\in(0,1]$, as \acp{ris} passively reflect signals without power amplification. When reflection-type amplifiers are incorporated, which is known as active \acp{ris}, signal amplification is enabled and~$a_i>1$ is available~\cite{Long2021Active,Zhang2022Active,Rao2023An}. It has been shown in~\cite{Zheng2023JrCUP} that a higher amplification factor of \ac{ris} accentuates the impact of mutual coupling between \ac{ris} unit cells. Hence, this work employs the active \ac{ris} setup to comprehensively evaluate the impact of \ac{mc} and propose the corresponding solutions. 

Examining~\eqref{eq:convCHM}, it is evident that the received signal at the \ac{bs} is a linear combination of signals reflected from each \ac{ris} unit cell. This model assumes that each RIS unit cell reflects incident \ac{em} waves independently without any interaction, thereby neglecting the \ac{mc} effect. In scenarios where strong MC exists, this assumption can lead to significant model mismatch and pose challenges for channel estimation~\cite{Zheng2024On}, degrading overall communication performance.

\subsection{Mutual Coupling-Aware Channel Model}
\label{sec:MCACM}

Although the \ac{mc} effect is sufficiently weak and can be reasonably omitted in many cases, recent studies have shown that the \ac{mc} effect between RIS unit cells is pronounced when they are tightly integrated (e.g., holographic RIS/MIMO~\cite{Wan2021Terahertz,Bjornson2024Towards}) or the amplification factor increases (e.g., active \acp{ris}~\cite{Long2021Active,Zheng2023JrCUP}). Based on the S-parameter multiport network theory, a mutual coupling-aware communication model has been recently derived and validated~\cite{Li2024Beyond,Abrardo2024Design,Zheng2024Mutual}, which will be adopted in this paper. To focus on the \ac{ris}, this work neglects the mutual coupling at the \ac{bs} and \ac{ue}. 

By incorporating the mutual coupling effect within the \ac{ris} response, the channel model~\eqref{eq:convCHM} is reformulated as~\cite{Wijekoon2024Phase,Li2024Beyond,Abrardo2024Design,Zheng2024Mutual}
\begin{equation}
    \label{eq:MC_ChM}
    \Hm_\mcidx^{m_\It} = \Hm_{\Rt\It} (\Gammam_{m_\It}^{-1} - \Sm)^{-1} \Hm_{\It\Tt},
\end{equation}
where~$\Sm\in\mathbb{C}^{N_\It\times N_\It}$ is the scattering matrix between the \ac{ris} unit cells, characterizing the \ac{mc} effect. Specifically,~$S_{i,j}$ denotes the scattering parameter between the~$i$$^\text{th}$ and~$j$$^\text{th}$ unit cells of the \ac{ris}, which is defined as
$
	S_{i,j} = \frac{V_i^\mathrm{out}}{V_j^\mathrm{in}}\big|_{V_k^\mathrm{in}=0,\ \forall k\neq j}
$~\cite{Pozar2011Microwave},
where~$V_j^\mathrm{in}$ denotes the voltage wave driving at the~$j$$^\text{th}$ unit cell and~$V_i^\mathrm{out}$ denotes the voltage wave coming out from the~$i$$^\text{th}$ unit cell.
 {The initial derivation of model~\eqref{eq:MC_ChM} can be found in~\cite{Shen2022Modeling}, with the complete derivation and analysis available in, e.g.,~\cite{Nerini2024Universal, Abrardo2024Design}. Additionally, an experimental validation  of~\eqref{eq:MC_ChM} on a real \ac{ris} prototype has been reported in~\cite{Zheng2024Mutual}. Note that although this model is commonly used to analyze passive \ac{ris}, it also remains valid for the active \ac{ris} case; see, e.g.,~\cite{Cao2025Active}.} For clarity, we refer to~\eqref{eq:convCHM} as the \emph{conventional model} and~\eqref{eq:MC_ChM} as the \emph{exact model}. Note that by setting~$\Sm = \mathbf{0}$, the exact model reduces to the conventional model.

\begin{remark}\label{rmk:RISamp}
    It can be inferred from~\eqref{eq:MC_ChM} that with the same level of \ac{mc}, higher amplification of the \ac{ris} can exacerbate the impact of \ac{mc}. This is because the scattering matrix~$\Sm$ is superposed on the inversion of~$\Gammam$. Hence, the greater the values in~$\Gammam$, the more significantly its inversion can be impacted by~$\Sm$.
\end{remark}

% pdf/png/jpg figures
\begin{figure}[t]
  \centering
  \includegraphics[width=0.85\linewidth]{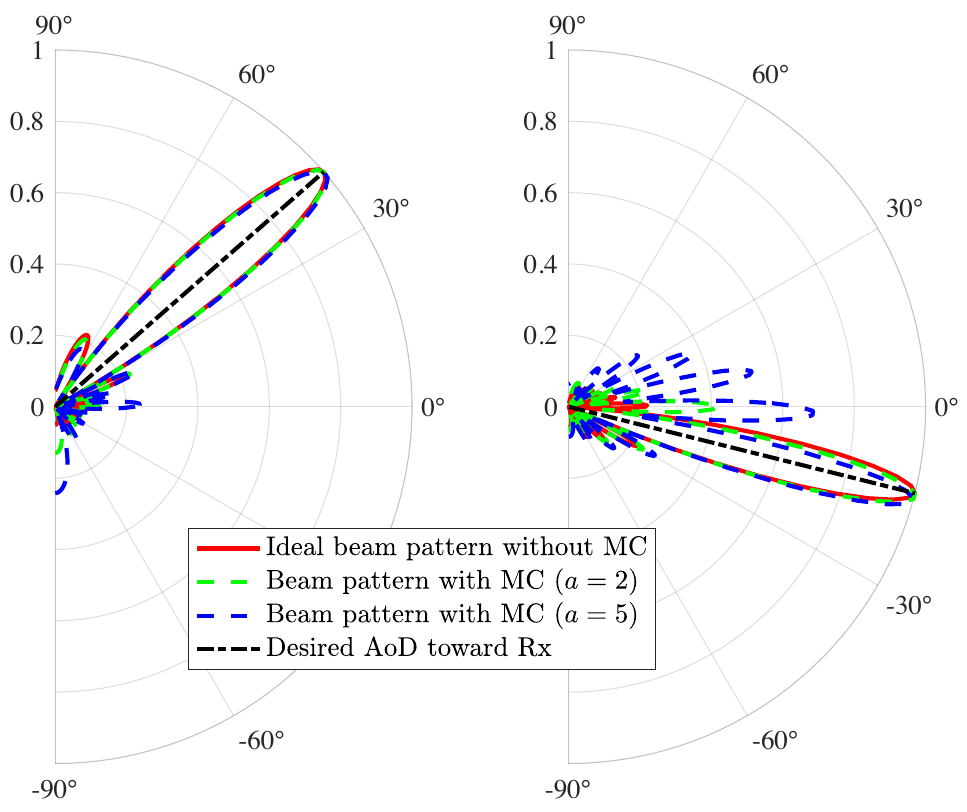}
  \vspace{-1em}
  \caption{ 
      Normalized beam pattern of the \ac{ris} reflection using the directional beam~$\bm{\gamma}=a\av_\It^*(\theta)\odot\av_\It^*(\phi)$ in a 16-element \ac{ula} layout \ac{ris}. Here, we test over different values of the RIS amplification factor~$a=\{2,5\}$, assuming all unit cells maintain the same amplification factor. The scattering matrix~$\Sm$ is assigned based on the measured data in~\cite{Zheng2024Mutual}.
    }
  \label{fig:BeamPattern}
  \vspace{-3mm}
\end{figure}

To verify the inference in Remark~\ref{rmk:RISamp}, Fig.~\ref{fig:BeamPattern} depicts the normalized beam pattern reflected from~\ac{ris} when a directional beam~$\bm{\gamma}=a\av_\It^*(\theta)\odot\av_\It^*(\phi)$ is used in a 16-element \ac{ula} layout \ac{ris}. This evaluation is based on the scattering parameters measured in~\cite{Zheng2024Mutual}. As illustrated, when the \ac{ris} amplification is low ($a=2$), the impact of \ac{mc} on the \ac{ris} radiation pattern is weak and negligible. However, when the \ac{ris} amplification is enlarged ($a=5$), we observe that the same \ac{mc} can more significantly distort the beam pattern.  

 {
\begin{remark}
This paper focuses on the \ac{mc}-aware setting, addressing channel estimation and beamforming problems given a scattering matrix~$\Sm$. However, in practical scenarios, this information may be unavailable or inaccurate. This challenge can be addressed via either offline measurement or online estimation. For offline measurement, a practical framework to measure the scattering matrix of a given RIS hardware is presented in~\cite{Zheng2024Mutual}. For online estimation, insightful examples can be found in, e.g.,~\cite{Fadakar2025Mutual}. Since these topics lie beyond the scope of this paper, we assume~$\Sm$ to be known and do not delve into its acquisition.
\end{remark}} 

\section{Uplink Signal Model and \ac{cs} Formulation}
\label{sec:rx_sig_mc}
This section details the signal model in the uplink. Subsequently, we present some useful mathematical transformations based on the natures of the received signals to reveal a \ac{cs}-like structure, thereby facilitating the later development of the channel estimation method. 
\subsection{The Received Signal}
\label{sec:rx_sig}
For illustrative convenience, this work assumes that both the \ac{bs} and the \ac{ue} are equipped with a single \ac{rfc}, adopting an analog beamforming scheme. We begin the uplink training procedure by sending pilots, training beams, and adjusting the amplitudes and phase shifts of the \ac{ris} elements within a time coherence block divided into multiple subframes. On the \ac{ris} side, this is achieved using~$M_\It$ configurations. Each subframe corresponds to one random~\ac{ris} setup. During each subframe, while the~\ac{ris} holds the configuration constant, both the \ac{ue} and \ac{bs} utilize various training beams. We assume that the number of transmitted pilots and precoding training beams equals the total number of combining training beams, which is $M_\Rt$ per subframe. Specifically, for the ${m_\It}^\text{th}$ \ac{ris} configuration, the \ac{ue} transmits~$M_\Rt$ pilots using random precoding training beams $\fv_{m_\Rt}\in\mathbb{C}^{N_\Tt}$, while the \ac{bs} records~$M_\mathrm{B}$ received pilots through random combining beams $\wv_{m_\Rt}\in\mathbb{C}^{N_\Rt}$, $\ m_\Rt=1,2,\dots,M_\Rt$. This results in a total of $M_\Rt M_\It$ training measurements. Note that we assume, during the uplink training, entries in $\fv_{m_\Rt}$ and $\wv_{m_\Rt}$ are constrained by $|f_{i,m_\Rt}|^2=\frac{P_\Tt}{N_\Tt}$, $\forall i=1,2,\dots,N_\Tt$, $|w_{j,m_\Rt}|^2=\frac{1}{{N_\Rt}}$, $\forall j=1,2,\dots,N_\Rt$, where $P_\Tt$ is the transmit power at the \ac{ue}. Moreover, the $M_\mathrm{B}$ precoding and combining beams remain fixed for different \ac{ris} configurations.

Based on the exact model~\eqref{eq:MC_ChM}, the received signal for the $(m_\Rt,m_\It)$ measurement is given by
\begin{align}
    \notag
    y_{m_\Rt,m_\It}&{=}\wv_{m_\Rt}^\HH\Hm_\mcidx^{m_\It}\fv_{m_\Rt}x_{m_\Rt}+\omega_{m_\Rt,m_\It},\\ %\label{eq:rxsig_mc_onemeas_2}
    \notag
    &{=}x_{m_\Rt}\big(\fv^\TT_{m_\Rt}\otimes\wv_{m_\Rt}^\HH\big)\vect{\Hm_\mcidx^{m_\It}}+\omega_{m_\Rt,m_\It},\\ \label{eq:rxsig_mc_onemeas_3}
    &{=} x_{m_\Rt}\pv_{m_\Rt}^\TT\vect{\Hm_\mcidx^{m_\It}}+\omega_{m_\Rt,m_\It},
\end{align}
where~$x_{m_\Rt}$ is the transmitted training pilot and~$\omega_{m_\Rt,m_\It}$ denotes the total additive noise. Here, we assume $\EE[x_{m_\Rt}x_{m_\Rt}^*]=1$. We denote $\pv_{m_\Rt}\triangleq\big(\fv^\TT_{m_\Rt}\otimes\wv_{m_\Rt}^\HH\big)^\TT\in\mathbb{C}^{N_\Rt N_\Tt}$ in~\eqref{eq:rxsig_mc_onemeas_3}. Based on the propagation model of active \acp{ris}, the total noise~$\omega_{m_\mathrm{B},m_\mathrm{I}}$ can be expressed as~\cite{Long2021Active,Zhang2022Active}
\begin{equation}\label{eq:totalNoise}
    \omega_{m_\Rt,m_\It} = \wv_{m_\Rt}^\HH\big(\Hm_{\Rt\It}(\Gammam_{m_\It}^{-1} - \Sm)^{-1}\omegav^{\mathrm{I}}_{m_\Rt,m_\It} + \omegav^{\mathrm{B}}_{m_\Rt,m_\It}\big),
\end{equation}
where~$\omegav^{\mathrm{I}}_{m_\Rt,m_\It}\!\sim\!\mathcal{CN}(\mathbf{0},\sigma_\mathrm{I}^2\mathbf{I}_{N_{\It}})$ and $\omegav^{\mathrm{B}}_{m_\Rt,m_\It}\!\sim\!\mathcal{CN}(\mathbf{0},\sigma_\mathrm{B}^2\mathbf{I}_{N_{\Rt}})$ are the thermal noise at the active \ac{ris} and the \ac{bs}, respectively.

We collect all the received signals over~$m_\mathrm{B}$ and $m_\mathrm{I}$, and denote the matrix of the resultant received signals after taking away the impact of $x_{m_\It}$ using a matched filter~\cite{Tarboush2023Compressive,Tarboush2024Cross} as~$\Ym\in\mathbb{C}^{M_\Rt\times M_\It}$, whose each column is given by
\begin{equation}
\label{eq:YmI}
    [\Ym]_{:,m_\It} = \Pm\vect{\Hm_\mcidx^{m_\It}} + \bar{\omegav}_{m_\It},
\end{equation}
where~$\Pm=[\pv_1,\pv_2,\dots,\pv_{M_\Rt}]^\TT\in\mathbb{C}^{M_\Rt\times N_\Rt N_\Tt}$ and~$\bar{\omegav}_{m_\mathrm{I}}=[\omega_{1,m_\It},\omega_{2,m_\It},\dots,\omega_{M_\Rt,m_\It}]^\TT\in\mathbb{C}^{M_\Rt}$. In the following subsections, we present some mathematical transformations based on the signal structure to reveal a \ac{cs}-like structure, thereby facilitating the later development of channel estimators. We first recap the conventional formulation without \ac{mc}, and then propose a new formulation that accurately accounts for \ac{mc}.

\subsection{The \ac{cs} Formulation Based on the Conventional Model}
\label{sec:CSCV}
Based on the conventional model~\eqref{eq:convCHM} that does not account for~\ac{mc}, one can rewrite~\eqref{eq:YmI}, by setting~$\Sm=\mathbf{0}$, as
\begin{equation}
    \label{eq:Ycv_2}
    [\Ym_\convidx]_{:,m_\It} =\Pm\vect{\Hm_\convidx^{m_\It}} + \bar{\omegav}_{m_\It}, 
\end{equation}
where we denote the received signal based on the conventional model as~$\Ym_\mathrm{cv}$ to distinguish it from the exact received signal~$\Ym$ defined in~\eqref{eq:YmI}.  {We can express the term $\vect{\Hm_\convidx^{m_\It}\!}$ in~\eqref{eq:Ycv_2} as 
\begin{equation}
\label{eq:vec_H_BIU_cv}
    \vect{\Hm_\convidx^{m_\It}}\!=\!\Am_{\Tt\Rt}(\varphiv,\varthetav)\!\big(\Sigmam_{\It\Tt}\otimes\Sigmam_{\Rt\It}\big)\!\big(\Am_{\It\It}^\convidx(\phiv,\thetav)\big)^\TT\bm{\gamma}_{m_\It},\!
\end{equation}
where $\Am_{\Tt\Rt}(\varphiv,\varthetav) \triangleq\Am_\Tt(\varphiv)\otimes\Am_\Rt(\varthetav)\in\mathbb{C}^{N_\Tt N_\Rt \times L_\Tt L_\Rt}$ and $\Am_{\It\It}^\convidx(\phiv,\thetav)\triangleq \Am_\It(\phiv)\bullet\Am_\It(\thetav)\in\mathbb{C}^{N_\It \times L_\Tt L_\Rt}$. A step-by-step derivation to obtain \eqref{eq:vec_H_BIU_cv} is given in Appendix \ref{sec:vec_Hcv_mI}. Note that the structure of $\Am_{\It\It}^\convidx(\phiv,\thetav)$ mainly holds because $\Gammam_{m_\It}$ is a diagonal matrix.} This is an essential difference between the conventional and exact models that will be derived later in Section~\ref{sec:CSFEM}. 

Furthermore, we define the \emph{conventional equivalent cascaded channel} and its vectorized version as~\cite{Wang2020Compressed,Pan2022An}
\begin{align}
&\label{eq:Gcv}\Gm_\convidx \!\triangleq\! \Am_{\Tt\Rt}(\varphiv,\varthetav)\!\big(\Sigmam_{\It\Tt}\!\otimes\!\Sigmam_{\Rt\It}\big)\!\big(\!\Am_{\It\It}^\convidx(\phiv,\thetav)\!\big)\!^\TT\!\!\!\in\!\mathbb{C}^{N_\mathrm{U}\!N_\mathrm{B}\!\times\! N_\mathrm{I}}\!,\!\\
&\label{eq:vec_Gcv}\vect{\Gm_\convidx}=\underbrace{\Am_{\It\It}^\convidx\!(\phiv,\!\thetav)\!\otimes\!\Am_{\Tt\Rt}(\varphiv,\!\varthetav)}_{\triangleq \Dm_\convidx}\!\mathrm{vec}\big(\Sigmam_{\It\Tt}\!\otimes\!\Sigmam_{\Rt\It}\big).
\end{align}
By substituting \eqref{eq:vec_H_BIU_cv} and \eqref{eq:Gcv} into \eqref{eq:Ycv_2}, we can write $[\Ym_\convidx]_{:,m_\It}$ as
\begin{equation}
    \label{eq:Ycv_tot}[\Ym_\convidx]_{:,m_\It}=\Pm\Gm_\convidx\bm{\gamma}_{m_\It}+ \bar{\omegav}_{m_\It},
\end{equation}
and we obtain, after collecting all of the measurements, the following observation model~\cite{Pan2022An,Wang2020Compressed,Wei2021Channel}
\begin{equation}
    \label{eq:rxsig_conv_ALLmeas}  \Ym_\convidx=\Pm\Gm_\convidx\Thetam_\convidx+\Omegam,
\end{equation}
where $\Thetam_\convidx\triangleq{\left[{\bm{\gamma}_1,\bm{\gamma}_2,\cdots,\bm{\gamma}_{M_\It}}\right]}\in\mathbb{C}^{N_\It\times M_\It}$ and~$\Omegam\triangleq[\bar{\omegav}_1,\bar{\omegav}_2,\dots,\bar{\omegav}_{M_\It}]\in\mathbb{C}^{M_\Rt\times M_\It}$. By applying a vectorization operation to \eqref{eq:rxsig_conv_ALLmeas} and substituting~\eqref{eq:vec_Gcv}, we obtain
\begin{align}
    \vect{\Ym_\convidx-\Omegam} &= \big(\Thetam_\convidx^\TT\otimes\Pm\big)\vect{\Gm_\convidx},\notag\\
    &= \Psim_\convidx\Dm_\convidx\mathrm{vec}\big(\Sigmam_{\It\Tt}\otimes\Sigmam_{\Rt\It}\big),\label{eq:vecYcv}
\end{align}
where we define~$\Psim_\convidx\triangleq\Thetam_\convidx^\TT\otimes\Pm\in\mathbb{C}^{M_\Rt M_\It\times N_\Rt  N_\Tt N_\It}$ as the known measurement matrix, and~$\Dm_\convidx$ is the basis matrix define in~\eqref{eq:vec_Gcv}. Here, it is clear that~\eqref{eq:vecYcv} is a typical \ac{cs} problem structure~\cite{Baraniuk2007Lecture}. By carefully designing the dictionary~$\Dm_\mathrm{cv}$ and adopting sparse \ac{cs} estimation algorithms, we can recover the equivalent cascaded channel~$\Gm_\mathrm{cv}$ based on the received signal~$\Ym_\mathrm{cv}$, since $\vect{\Gm_\mathrm{cv}}=\Dm_\mathrm{cv}\mathrm{vec}\big(\Sigmam_{\It\Tt}\otimes\Sigmam_{\Rt\It}\big)$.

\subsection{The \ac{cs} Formulation Based on the Exact Model}
\label{sec:CSFEM}
Now, let's consider the case with \ac{mc}. Due to the existence of matrix~$\Sm$, the \ac{ris} response~$(\Gammam_{m_\It}^{-1} - \Sm)^{-1}$ is no longer a diagonal matrix, thus \eqref{eq:vec_H_BIU_cv} does not hold. The following derivation aims to formulate an accurate \ac{cs} structure accounting for \ac{mc} starting from~\eqref{eq:YmI}. Following the same steps in deriving~\eqref{eq:vec_H_BIU_cv} in Section~\ref{sec:CSCV}, we can express the term $\vect{\Hm_\mcidx^{m_\It}}$ in~\eqref{eq:rxsig_mc_onemeas_3} as
\begin{align}
\vect{\Hm_\mcidx^{m_\It}}\!&=\!\vect{\!\Am_\Rt(\varthetav)\Sigmam_{\Rt\It}\Am_\It^\TT(\thetav)\bar{\Gammam}_{m_\It}\Am_\It(\phiv)\Sigmam_{\It\Tt}\Am^\TT_\Tt(\varphiv)\!},\notag\\ &\label{eq:vec_H_BIU_mc}=\Am_{\Tt\Rt}(\varphiv,\varthetav)\big(\Sigmam_{\It\Tt}\otimes\Sigmam_{\Rt\It}\big)\Am_{\It\It}^\TT(\phiv,\thetav)\xiv_{m_\It},
\end{align}
where $\Am_{\It\It}(\phiv,\thetav)\triangleq\Am_\It(\phiv)\otimes\Am_\It(\thetav)\in\mathbb{C}^{N_\It^2 \times L_\Tt L_\Rt}$ and, for notational convenience, we denote ~$\xiv_{m_\It}\triangleq\vect{\bar{\Gammam}_{m_\It}}\in\mathbb{C}^{N^2_\It}$ with ~$\bar{\Gammam}_{m_\It}\triangleq(\Gammam_{m_\It}^{-1} - \Sm)^{-1}$. The difference lies in the last two terms of~\eqref{eq:vec_H_BIU_mc} compared to~\eqref{eq:vec_H_BIU_cv}, which arises from the non-diagonality of the \ac{mc}-aware \ac{ris} response matrix.

We define the \emph{exact equivalent cascaded channel} and its vectorized version
\begin{align}
\label{eq:Gmc}&\Gm_\mcidx \!\triangleq\! \Am_{\Tt\Rt}(\varphiv,\varthetav)\big(\Sigmam_{\It\Tt}\!\otimes\!\Sigmam_{\Rt\It}\big)\Am_{\It\It}^\TT(\phiv,\thetav)\!\in\!\mathbb{C}^{N_\mathrm{U}\!N_\mathrm{B}\!\times\! N_\mathrm{I}^2}\!,\\
\label{eq:vec_Gmc}&\vect{\Gm_\mcidx}\!=\!\underbrace{\Am_{\It\It}(\phiv,\thetav)\!\otimes\!\Am_{\Tt\Rt}(\varphiv,\varthetav)}_{\triangleq\Dm_\mcidx}\mathrm{vec}\big(\Sigmam_{\It\Tt}\otimes\Sigmam_{\Rt\It}\big). 
\end{align}
Similarly, based on~\eqref{eq:YmI}, we have
\begin{equation}\label{eq:Y_ExactModel}
    \Ym =\Pm\Gm_\mcidx\Thetam_\mcidx+\Omegam,
\end{equation}
where we define~$\Thetam_\mcidx\triangleq[\xiv_1, \xiv_2,\cdots,\xiv_{M_\It}]\in\mathbb{C}^{N^2_\It\times M_\It}$. By vectorizing~\eqref{eq:Y_ExactModel} and substituting~\eqref{eq:vec_Gmc}, we have
\begin{equation}
\label{eq:Rxsig_mc_vecChExc}
\vect{\Ym-\Omegam}=\Psim_\mcidx\Dm_\mcidx\mathrm{vec}\big(\Sigmam_{\It\Tt}\otimes\Sigmam_{\Rt\It}\big),
\end{equation}
which follows the typical structure of the \ac{cs} problem~\cite{Baraniuk2007Lecture}. Here, $\Psim_\mcidx\triangleq(\Thetam_\mcidx^\TT\otimes\Pm)\in\mathbb{C}^{M_\Rt M_\It\times N_\Rt N_\Tt N^2_\It}$ is the measurement matrix, and~$\Dm_\mcidx$ is the basis matrix defined in~\eqref{eq:vec_Gmc}. 

\begin{remark}
Comparing~\eqref{eq:vecYcv} and~\eqref{eq:Rxsig_mc_vecChExc} reveals that accounting for \ac{mc} results in a dimensional increase in both the measurement and basis matrices. Specifically, the dimension of~$\Psim_\mathrm{cv}$ is~$M_\Rt M_\It \times N_\Rt N_\Tt N_\It$, while the dimension of~$\Psim_\mathrm{mc}$ expands to~$M_\Rt M_\It \times N_\Rt N_\Tt N_\It^2$. Moreover, the dimension of the atoms in $\mathbf{D}_\mathrm{mc}$ is~$N_\mathrm{I}$ times that of $\mathbf{D}_\mathrm{cv}$. The enhanced accuracy comes at the expense of \emph{higher complexity}. Given the large number of elements in an \ac{ris} (and consequently, a large $N_\mathrm{I}$), such an increase in dimensionality can render \ac{cs} solutions infeasible in practice. We will show in Section~\ref{sec:prop_ch_mc} that applying a \ac{dr} strategy can effectively mitigate this complexity issue.
\end{remark}

\section{Low-Complexity Channel Estimation Strategy} 
\label{sec:CE_MC}
This section first presents the on-grid \ac{cs} solutions for channel estimation based on both the conventional and exact channel models. Next, a two-stage estimation strategy is proposed to strike a balance between computational complexity and model accuracy. 

\subsection{Quantized Dictionaries and On-Grid Channel Estimation}
We approximate the channel based on pre-determined dictionaries and apply on-grid \ac{cs} techniques to solve~\eqref{eq:vecYcv} and~\eqref{eq:Rxsig_mc_vecChExc}. Following angular quantization to a grid of size $G_i=G_i^\mathrm{h} G_i^\mathrm{v}$,\footnote{ {Since our channel model is based on ideal array response matrices, we use overcomplete discrete Fourier transform (DFT) matrices as sparsifying dictionaries. However, in practical scenarios involving non-isotropic radiation patterns, hardware impairments, and calibration errors, DFT dictionaries become suboptimal. In such cases, dictionary learning offers a promising approach to better capture the underlying sparse channel structure. We leave this direction for future work.}} the discrete array response matrices~$\tilde{\Am}_i(\psiv)\in\mathbb{C}^{N_i\times G_i}$, recalling that $i\in \{\Tt,\Rt,\It\}$, can be constructed as
\begin{equation}
    [\tilde{\Am}_i(\psiv)]_{:,g_{i}^\mathrm{h}+(g_{i}^\mathrm{v}\!-\!1)G_{i}^\mathrm{h}} \!\!=\!\frac{1}{\sqrt{N_i}}e^{-j2\pi{\beta}^\mathrm{h}_{g_{i}^\mathrm{h}}\nv(\!N_i^\mathrm{h}\!)}\!\!\!\otimes\!e^{-j2\pi{\beta}^\mathrm{v}_{g_{i}^\mathrm{v}}\nv(\!N_i^\mathrm{v}\!)}\!\!.\notag
\end{equation}
For $\mathrm{\varepsilon}\in \{\mathrm{h},\mathrm{v}\}$, the discrete spatial angles are defined as~\cite{Tarboush2023Compressive}
\begin{equation}      
    \label{eq:grid_spatialangle}
    {\beta}^\mathrm{\varepsilon}_{g_i^\mathrm{\varepsilon}}\!=\!\frac{2d_i}{\lambda G_i^\mathrm{\varepsilon}}\!\Big(g_i^\mathrm{\varepsilon}\!-\!\frac{G_i^\mathrm{\varepsilon}\!+\!1}{2}\Big),\ g_i^\mathrm{\varepsilon}\!=\!1,\!\cdots\!,G_i^\mathrm{\varepsilon};\ G_i^\mathrm{\varepsilon}\!\geq\! N_i^\mathrm{\varepsilon}.
\end{equation}
According to~\eqref{eq:vec_Gcv} and~\eqref{eq:vec_Gmc}, we construct the conventional and exact dictionary matrices as 
\begin{align}
    &\label{eq:Dtilde_cv}\tilde{\Dm}_\convidx \!=\! \tilde{\Am}_{\It\It}^{\convidx}(\phiv,\thetav)\!\otimes\!\tilde{\Am}_{\Tt\Rt}(\varphiv,\varthetav)\!\in\!\mathbb{C}^{N_\Rt N_\Tt N_\It\!\times\! G_\Rt G_\Tt G^2_\It}\!,\\&\label{eq:Dtilde_mc}\tilde{\Dm}_\mcidx \!=\! \tilde{\Am}_{\It\It}(\phiv,\thetav)\!\otimes\!\tilde{\Am}_{\Tt\Rt}(\varphiv,\varthetav)\in\mathbb{C}^{N_\Rt N_\Tt N^2_\It\!\times\! G_\Rt G_\Tt G^2_\It},
\end{align}
where $\tilde{\Am}_{\It\It}^{\convidx}(\phiv,\thetav)\triangleq\tilde{\Am}_\It(\phiv)\bullet\tilde{\Am}_\It(\thetav)\in\mathbb{C}^{N_\It\times G^2_\It}$, $\tilde{\Am}_{\It\It}(\phiv,\thetav)\triangleq\tilde{\Am}_\It(\phiv)\otimes\tilde{\Am}_\It(\thetav)\in\mathbb{C}^{N^2_\It\times G^2_\It}$, and $\tilde{\Am}_{\Tt\Rt}(\varphiv,\varthetav) \triangleq\tilde{\Am}_\Tt(\varphiv)\otimes\tilde{\Am}_\Rt(\varthetav)\in\mathbb{C}^{N_\Rt N_\Tt \times G_\Rt G_\Tt}$. Following~\cite[Proposition~1]{Wang2020Compressed}, the Khatri-Rao-based dictionary $\tilde{\Am}_{\It\It}^{\convidx}(\phiv,\thetav)$ of conventional structure~\eqref{eq:Dtilde_cv} contains significant redundancy.  {Letting $ \bar{G}_{\It\It}$ denote the number of distinct columns in $\tilde{\Am}_{\It\It}^{\convidx}(\phiv,\thetav)$ with $ G_\mathrm{I}\leq\bar{G}_{\It\It}\ll G^2_\It$, we then have $G_\Rt G_\Tt G_\It\!\leq\! \bar{G}_\convidx\!\ll\! G_\Rt G_\Tt G^2_\It$, with $\bar{G}_\convidx\triangleq G_\Rt G_\Tt \bar{G}_{\It\It}$ denoting the number of distinct columns in $\bar{\Dm}_{\mathrm{cv}}$.} However, the Kronecker-based dictionary $\tilde{\Am}_{\It\It}(\phiv,\thetav)$ of exact structure~\eqref{eq:Dtilde_mc} does not pose similar redundancy, i.e. the number of distinct columns in $\tilde{\Dm}_\mathrm{mc}$ is $\bar{G}_\mcidx=G_\Rt G_\Tt G^2_\It$.
Thus, we define~$\bar{\Dm}_\convidx\in\mathbb{C}^{N_\Rt N_\Tt N_\It\times \bar{G}_\convidx}$ and~$\bar{\Dm}_\mcidx\in\mathbb{C}^{N_\Rt N_\Tt N_\It^2\times \bar{G}_\mcidx}$ as the dictionaries that contain only distinct columns. This implies that the exact dictionary~$\bar{\Dm}_\mcidx$ not only has a higher dimensionality of atoms but also a greater number of atoms.

By substituting~$\bar{\Dm}_\convidx$/$\bar{\Dm}_\mcidx$ into~\eqref{eq:vecYcv}/\eqref{eq:Rxsig_mc_vecChExc}, we can get the sparse formulation for the conventional/exact model as
\begin{equation}
    \label{eq:Rxsig_CS}  \yv=\vect{\Ym}=\Xim_\chi\bar{\sigmav}_\chi+\vect{\Omegam},
\end{equation}
where $\Xim_\chi\!=\!\Psim_\chi\bar{\Dm}_\chi\!\in\!\mathbb{C}^{M_\Rt M_\It \times \bar{G}_\chi}$ is the sensing matrix and $\bar{\sigmav}_\chi\!\in\!\mathbb{C}^{\bar{G}_\chi}$ is a sparse vector,~$\chi\in\{\mathrm{cv},\mathrm{mc}\}$. Here,~$\bar{\sigmav}_\chi$ can be estimated using the CS algorithm such as \ac{omp}~\cite{Lee2016Channel}. Then the equivalent cascaded channel~$\Gm_\mathrm{cv}$/$\Gm_\mathrm{mc}$ can be recovered based on~\eqref{eq:vec_Gcv}/\eqref{eq:vec_Gmc}. Nonetheless, the conventional model-based solution suffers from significant model mismatch due to the neglect of MC, while the exact model-based solution involves a dimensional lift (as seen when comparing the dimensions of the dictionaries~$\bar{\Dm}_\mathrm{cv}$ and~$\bar{\Dm}_\mathrm{mc}$). In the following subsection, we propose a two-stage compromise estimation strategy that enjoys the advantages of both methods.

\subsection{Proposed Two-Stage Channel Estimation Strategy}
\label{sec:prop_ch_mc}

\begin{algorithm}[t]
 \caption{\small{Proposed Two-Stage MC--aware Channel Estimation}}
 \label{algo:prop_ch_est}
 \begin{algorithmic}[1]
 \Statex \textbf{Input:} $\yv$, $\Xim_\convidx$, $\Psim_\mathrm{mc}$, $\hat{L}$, and $G_\DicRed$. \quad  \textbf{Output:} $\hat{\sigmav}_{\mathrm{mc}}$
 \LineComment{\textbf{Stage 1:} Coarse estimation}
 \State Estimate using OMP [$\hat{\sigmav}_\mathrm{cv}, \Upsilon^\mathrm{cv}$] = OMP($\yv,\Xim_\convidx,\hat{L}$).
 \LineComment{\textbf{Stage 2:} Refined estimation}
 \State Define $\acute{\Am}_\It=[\tilde{\Am}_{\It\It}^{\convidx}]_{:,\Upsilon^\mathrm{cv}_\It}$ according to the indices $\Upsilon^\mathrm{cv}_\It$ that corresponds to~\ac{ris} from previous estimated support $\Upsilon^\mathrm{cv}$.
 \State Apply proposed DR algorithm $\hat{\Am}_\It^\DicRed$ = DR($\acute{\Am}_\It,\tilde{\Am}_{\It\It},G_\DicRed$).
 \State Compute $\Xim_\DicRed=\Psim_\mathrm{mc}\bar{\Dm}_\DicRed$ with $\bar{\Dm}_\DicRed=\hat{\Am}_\It^\DicRed\otimes\tilde{\Am}_{\Tt\Rt}$.
 \State Estimate using OMP [$\hat{\sigmav}_{\mathrm{mc}}, \Upsilon^{\mathrm{mc}}$] = OMP($\yv,\Xim_\DicRed,\hat{L}$).
  \Statex {------ \textbf{\ac{dr} Algorithm}: ${\Am}_\mathrm{DR}$ = DR(${\Am},{\Am}_\ovs,{G}$) ------}
 \State Compute correlation: $ {\Cm\!\!=\!\!([{\Am}_\ovs]_{1:N_\It,1:\bar{G}_{\It\It}})^\HH {{\Am}}}$.
 \State Select ${G}$ indices from $\sqrt{\diagopr{\Cm \Cm^\HH}}$ with largest correlation value and populate the index set ${\mathcal{G}}$ with these indices. 
 \State Return DR matrix ${\Am}_\mathrm{DR}=[{\Am}_\ovs]_{:,{\mathcal{G}}}$.
 \end{algorithmic} 
\end{algorithm}
\subsubsection{ {Motivation}}
\label{sec:motv}
Before presenting our estimation strategy, we elaborate on the motivation behind it.  {Both the \ac{ris} and the~\ac{bs} typically deploy massive antenna arrays operating at high frequencies, resulting in increased channel estimation complexity and training overhead. When an initial estimate or prior information is available, reducing the search space to a carefully designed subspace, rather than the full parameter space, can significantly reduce the computational complexity.} As discussed in Section~\ref{sec:MCACM}, the exact model simplifies to the conventional model when~$\Sm=\mathbf{0}$. Although this simplification introduces model mismatch, we can leverage its low complexity to obtain a coarse estimate as prior information for the exact model-based method. This, in turn, helps reduce the complexity of the exact CS solution. Building on this idea, we propose a two-stage estimation algorithm outlined in Algorithm~\ref{algo:prop_ch_est}. Specifically, the first stage estimates~$\bar{\sigmav}_\mathrm{cv}$ by applying \ac{omp} to~\eqref{eq:Rxsig_CS} with~$\chi=\mathrm{cv}$. In the second stage, we utilize the prior information $\hat{\sigmav}_\mathrm{cv}$ and support~$\Upsilon^\mathrm{cv}$ to reduce the search space by shrinking the columns of $\tilde{\Am}_{\It\It}(\phiv,\thetav)$ (and consequently $\bar{\Dm}_\mathrm{mc}$).\footnote{Here, we apply the proposed \ac{dr} technique to the \ac{ris} side only because of the typically larger size of the RIS compared to the BS antenna array. However, the same logic can be applied to the~\ac{bs}.} 

\subsubsection{Coarse Estimation}
\label{sec:CScoarse}
In the first stage, we apply the~\ac{omp} algorithm to estimate $\bar{\sigmav}_\mathrm{cv}$ depending on $\Xim_\convidx$ and $\hat{L}$, as presented in step~1 of Algorithm~\ref{algo:prop_ch_est}. The parameter~$\hat{L}$ is the number of paths to be estimated and can be different from the actual channel sparsity level because it is unknown. Details of the \ac{omp} algorithm can be found in~\cite[Algorithm~1]{Lee2016Channel}. The estimated sparse vector~$\hat{\sigmav}_\mathrm{cv}$ is filled based on the detected support~$\Upsilon^\mathrm{cv}$ with a cardinality of $\abs{\Upsilon^\mathrm{cv}}=\hat{L}$. 

\subsubsection{Refined Estimation}
In the second stage, we perform exact CS estimation with prior information of a coarse estimate.\footnote{This prior information can be obtained either through the conventional model-based CS estimation described in Section~\ref{sec:CScoarse} or from the previous exact CS estimation.} 
 {The primary goal of the~\ac{dr} stage is to intelligently select a subset of atoms and construct a compact dictionary. This selection is made from a large and potentially oversampled dictionary. In this paper, the correlation-based atom selection is used to quantify the coherence between atoms associated with the estimated support from the first stage and those in the larger dictionary. For any atom in an overcomplete discrete Fourier transform (DFT) dictionary, selecting the $G$ most correlated atoms results in a refined set that effectively captures the angular region around the original atoms.} 

In step 2 of Algorithm~\ref{algo:prop_ch_est}, we construct $\acute{\Am}_\It\!\in\!\mathbb{C}^{N_\It\!\times\!\hat{L}}$, a matrix consisting of the columns of~$\tilde{\Am}_{\It\It}^{\convidx}$ corresponding to the support $\Upsilon^\mathrm{cv}$ obtained in step~1. This dictionary represents the spatial angles to limit the search space over and is the first input for the proposed~\ac{dr} algorithm. The second input is the dictionary~$\tilde{\Am}_{\It\It}$. We apply the proposed~\ac{dr} algorithm in step 3, where the idea is to extract only a small number of columns from $\tilde{\Am}_{\It\It}$ that were used to get~$\bar{\Dm}_\mcidx$. This extraction process, i.e., the \ac{dr} process, is summarized in steps 6--8 of Algorithm~\ref{algo:prop_ch_est}.  {To reduce the complexity of the \ac{dr} process, we do not select columns directly from $\tilde{\Am}_{\It\It}$. Instead, we first choose a submatrix $\tilde{\Am}_{\It\It}^\mathrm{sub}\in\mathbb{C}^{N_\mathrm{I}\times \bar{G}_\mathrm{II}}$ from $\tilde{\Am}_{\It\It}\in\mathbb{C}^{N_\mathrm{I}^2\times G_\mathrm{I}^2}$. Empirically, $\tilde{\Am}_{\It\It}^\mathrm{sub}$ can be chosen as the first $\bar{G}_\mathrm{II}$ columns of $\tilde{\Am}_{\It\It}$. Next, we extract $G_\DicRed$ columns extracted from $\tilde{\Am}_{\It\It}^\mathrm{sub}$ to construct the new dictionary. The selected $G_\DicRed$ columns have the highest correlation with $\acute{\Am}_\It$ (from step 2) and correspond to the most probable spatial angles where the channel is expected to be. Specifically, we define the \emph{\ac{dr} factor} $\rho_\DicRed\in(0,1]$ as the ratio of the number of selected atoms (columns) in the reduced dictionary to the $\bar{G}_{\It\It}$ candidate atoms from $\tilde{\Am}_{\It\It}^\mathrm{sub}$, and accordingly define $G_\DicRed= \rho_\DicRed \bar{G}_{\It\It}= \rho_\DicRed \bar{G}_\mathrm{cv}/{G_\Rt G_\Tt}\ll\! G^2_\It$. The \emph{\ac{dr} factor} controls how many atoms are kept after the DR algorithm.} In step 4, by using~$\hat{\Am}_\It^\DicRed$ provided by step~3, we compute the~\ac{dr} sensing matrix to be later used in step~5 where we apply the~\ac{omp} algorithm. The output~$\hat{\sigmav}_{\mathrm{mc}}$ is an estimate of~$\bar{\sigmav}_{\mathrm{mc}}$ defined in~\eqref{eq:Rxsig_CS}. Finally, the exact equivalent cascaded channel~$\Gm_\mathrm{mc}$ can be recovered as~$\mathrm{vec}(\hat{\Gm}_\mathrm{mc})=\bar{\Dm}_\mathrm{DR}\hat{\sigmav}_\mathrm{mc}$.

 {\subsubsection{Complexity Analysis}

\begin{table*} [htb]
\vspace{-3mm}
\footnotesize
\centering
 {\caption{Overall Computational complexity of proposed and baseline algorithms}
\begin{tabular} {c c c}
 \hthickline
 Algorithm & Online operation & Offline operation\\ [0.5ex] 
 \hthickline
 \ac{mc}-unaware OMP& $\mathcal{C}_{\convidx}^{\mathrm{OMP}} = \mathcal{O}\big(\hat{L} M_\It M_\Rt \bar{G}_\convidx\big)$ & $\mathcal{C}_{\Xim_\convidx} = \mathcal{O}\big(M_\Rt M_\It N_\Rt N_\Tt N_\It \bar{G}_\convidx\big)$\\ 
 \hline
 \ac{mc}-aware OMP & $\mathcal{C}_{\mcidx}^{\mathrm{OMP}} = \mathcal{O}\big(\hat{L} M_\It M_\Rt \bar{G}_\mcidx\big)$ & $\mathcal{C}_{\Xim_\mcidx} = \mathcal{O}\big(M_\Rt M_\It N_\Rt N_\Tt N^2_\It \bar{G}_\mcidx\big)$ \\
 \hline
 Proposed & $\mathcal{C}_{\mathrm{prop}}\!=\!\mathcal{O}\Big(\!\hat{L} M_\It M_\Rt \bar{G}_\convidx \big(\! 1\!+\!\rho_\DicRed\!\big)\!\Big)$ & $ \mathcal{C}_{\Xim_\convidx,\Xim_\DicRed,\DicRed} \!=\!\mathcal{O}\Big(\!M_\Rt M_\It N_\Rt N_\Tt N_\It \bar{G}_\convidx \big(\!1\!+\!N_\It \rho_\DicRed\!\big) \!+\! \bar{G}_{\It\It} \hat{L} \big(\!\bar{G}_{\It\It} \!+\! N_\It \!\big)\!+\!N^2_\It G_\DicRed \!\Big)$ \\
 \hthickline
\end{tabular}
\vspace{-3mm}
\label{table:complexity_com}}
\end{table*}

Since some steps can be performed before running the channel estimation algorithms, we will distinguish between online and offline operations. The computational complexity of online operation for the baseline and derived exact algorithms, namely \ac{mc}-unaware and \ac{mc}-aware, is mainly the complexity of applying the OMP, which depends on the sensing matrix $\Xim$ size~\cite{Tarboush2024Cross}. The computational complexity of step 2 in Algorithm~\ref{algo:prop_ch_est}, indexing and data-copying operations, is $\mathcal{O}\big(\hat{L} N_\It \big)$, and it is negligible compared to other steps. The online computational complexity of the proposed Algorithm~\ref{algo:prop_ch_est} consists of two main parts: 1) the coarse estimation in step 1 applies \ac{mc}-unaware OMP and has a complexity of $\mathcal{O}\big(\hat{L} M_\It M_\Rt \bar{G}_\convidx\big)$, and 2) step 5 incurs computational cost of $\mathcal{O}\big(\hat{L} M_\It M_\Rt \bar{G}_\convidx \rho_\DicRed\big)$. Moreover, the offline operation is mainly dominated by computing the \ac{dr} algorithm and generating the corresponding sensing matrix $\Xim$. The baseline \ac{mc}-unaware and the derived \ac{mc}-aware methods entail offline complexities of $\mathcal{C}_{\Xim_\convidx} = \mathcal{O}\big(M_\Rt M_\It N_\Rt N_\Tt N_\It \bar{G}_\convidx\big)$ and $\mathcal{C}_{\Xim_\mcidx} = \mathcal{O}\big(M_\Rt M_\It N_\Rt N_\Tt N^2_\It \bar{G}_\mcidx\big)$, respectively. Note that, for a fair comparison with these methods, we have considered steps 3 and 4 to be offline operations. The complexity of step 4 is $\mathcal{O}\big(M_\Rt M_\It N_\Rt N_\Tt N^2_\It \bar{G}_\convidx \rho_\DicRed\big)$ while the total computational complexity of the DR algorithm in step 3 is $\mathcal{O}\big(\bar{G}_{\It\It} N_\It \hat{L} + \bar{G}_{\It\It}^2 \hat{L} + N^2_\It G_\DicRed \big)$, mainly due to steps 6, 7, and 8, while the computational complexity of operations like extracting diagonal elements, computing the square root, and selecting the largest $G_\DicRed$ values, in step 7, are negligible. The computational complexity for different channel estimation methods is summarized in Table \ref{table:complexity_com}. Recalling that $\bar{G}_\convidx\!\ll\! \bar{G}_\mcidx, \rho_\DicRed\in(0,1]$, and $G_\DicRed \ll G^2_\It$, our proposal has lower offline computational complexity than the \ac{mc}-aware solution. For moderate \ac{ris} size and large $\rho_\DicRed$, the proposed solution has a comparable online complexity with the derived exact \ac{mc}-aware model. However, for large \ac{ris} sizes, the proposed solution enjoys reduced online complexity.}
 {\begin{remark}
Note that our proposed channel estimation approach estimates the equivalent cascaded channel~$\Gm_\mathrm{mc}$ defined in~\eqref{eq:vec_Gmc}, which differs from the commonly used individual subchannel matrices $\Hm_\mathrm{IU}$ and $\Hm_\mathrm{BI}$. This formulation enables \ac{cs}-based channel estimation that explicitly accounts for the \ac{ris} \ac{mc} effects. However, since most existing beamforming methods rely on individual channel knowledge $\Hm_\mathrm{IU}$ and $\Hm_\mathrm{BI}$, dedicated beamforming methods tailored to the equivalent channel~$\Gm_\mathrm{mc}$ are required. In the following Section~\ref{sec:JBF}, we propose a novel beamforming algorithm based on~$\Gm_\mathrm{mc}$, demonstrating that this equivalent cascaded channel contains sufficient information to enable effective \ac{mc}-aware beamforming. Together, the proposed channel estimation and beamforming methods constitute a unified framework for configuring active \ac{ris}-assisted \ac{mimo} systems under mutual coupling.
\end{remark}
}

\section{MC-Aware Downlink Beamforming}\label{sec:JBF}

This section analyzes and addresses the downlink beamforming problem given an exact equivalent cascaded channel estimate obtained via the uplink channel estimator described in Section~\ref{sec:CE_MC}. For the sake of clarification, we further denote the estimated unlink channel as ~${{\Gm}}_{\mathrm{mc}}^\mathrm{UL}\in\mathbb{C}^{N_\mathrm{U}N_\mathrm{B}\times N_\mathrm{I}^2}$.

\subsection{Problem Formulation}
 To maintain consistency with the notation in previous sections, we continue to denote the downlink precoder at the \ac{bs} as $\wv \in \mathbb{C}^{N_\mathrm{B}}$ and the combiner at the \ac{ue} as $\fv \in \mathbb{C}^{N_\mathrm{U}}$, both being analog. Hence, the downlink signal model is written as
\begin{equation}
    s_\mathrm{U} = \fv^\HH\Hm_{\mathrm{UI}}(\Gammam^{-1}-\Sm)^{-1}\Hm_\mathrm{IB}\wv s_\mathrm{B} + \omega,
\end{equation}
where $\Hm_{\mathrm{UI}} = \Hm_{\mathrm{IU}}^\TT$, $\Hm_{\mathrm{IB}} = \Hm_{\mathrm{BI}}^\TT$, $s_\mathrm{B}$ and $s_\mathrm{U}$ are the transmitted and received symbols, respectively, and $\omega$ denotes the noise. Specifically, $\omega = \fv^\HH\big(\Hm_{\mathrm{UI}}(\Gammam^{-1} - \Sm)^{-1}\omegav^{\mathrm{I}} + \omegav^{\mathrm{U}}\big)$,  where $\omegav^{\mathrm{I}}\!\sim\!\mathcal{CN}(\mathbf{0},\sigma_\mathrm{I}^2\mathbf{I}_{N_{\It}})$ and $\omegav^{\mathrm{U}}\!\sim\!\mathcal{CN}(\mathbf{0},\sigma_\mathrm{U}^2\mathbf{I}_{N_{\mathrm{U}}})$ represent the thermal noise at the active \ac{ris} and the \ac{ue}, respectively. 

Following the analog beamforming setup, we constrain that the entries in $\wv$ as $|w_{j}|^2=\frac{P_\mathrm{B}}{{N_\Rt}}$, $\forall j=1,2,\dots,N_\Rt$, where $P_\mathrm{B}$ is the transmit power at the \ac{bs} with $\mathbb{E}[s_\mathrm{B}s_\mathrm{B}^*]=1$. As a result, the receive \ac{snr} can be calculated as~\cite{Zhang2022Active} 
\begin{equation}\label{eq:SNR}
\mathrm{SNR} = \frac{|\fv^\HH\Hm_\mathrm{UI}(\Gammam^{-1} -\Sm)^{-1}\Hm_\mathrm{IB}\wv|^2}{\|\fv\|_2^2\sigma_\mathrm{U}^2 + \|\fv^\HH\Hm_\mathrm{UI}({\Gammam}^{-1}-\Sm)^{-1}\|_2^2\sigma_\It^2}.
\end{equation}
Note that the denominator here contains two terms of noise power: (i) the power of noise generated at the \ac{ue} itself, computed as $\|\fv\|_2^2\sigma_\mathrm{U}^2$, and (ii) the power of noise generated at the active \ac{ris}, which is amplified, propagated, and received by the \ac{ue}, computed as $\|\fv^\HH\Hm_\mathrm{UI}({\Gammam}^{-1} - \Sm)^{-1}\|_2^2\sigma_\It^2$. 
\begin{remark}\label{rmk:noisePower}
    In far-field scenarios, the additional noise term $\|\fv^\HH\Hm_\mathrm{UI}({\Gammam}^{-1} - \Sm)^{-1}\|_2^2\sigma_\It^2$ due to active RIS is usually negligible compared to the noise level at the receiver \ac{ue} $\|\fv\|_2^2\sigma_\mathrm{U}^2$.
\end{remark}

\begin{figure}[t]
    \centering

\definecolor{mycolor1}{rgb}{0.3686, 0.5098, 0.7098}  % 柔和蓝灰
\definecolor{mycolor2}{rgb}{0.6824, 0.4353, 0.7922}  % 柔和紫
\definecolor{mycolor3}{rgb}{0.5020, 0.6824, 0.4509}  % 橄榄绿/鼠尾草绿
\begin{tikzpicture}

\begin{axis}[%
width=2.8in,
height=1.9in,
at={(0in,0in)},
scale only axis,
xmode=log,
xmin=1.9,
xmax=67.2,
xtick={ 2,  4,  8, 16, 32, 64},
xticklabels={2,  4,  8, 16, 32, 64},
xminorticks=true,
xlabel style={font=\color{white!15!black},font=\footnotesize},
xticklabel style = {font=\color{white!15!black},font=\footnotesize},
xlabel={Amplification factor $a$},
ymin=-175,
ymax=-80,
ylabel style={font=\color{white!15!black},font=\footnotesize},
yticklabel style = {font=\color{white!15!black},font=\footnotesize},
ylabel={Noise power (dBm)},
axis background/.style={fill=white}
]
\addplot [color=red, line width=1.0pt, forget plot]
  table[row sep=crcr]{%
1	-88.9794000867204\\
4	-88.9794000867204\\
8	-88.9794000867204\\
16	-88.9794000867204\\
32	-88.9794000867204\\
70	-88.9794000867204\\
};
\addplot [color=mycolor1, line width=1.0pt, forget plot]
 plot [error bars/.cd, y dir=both, y explicit, error bar style={line width=1.0pt}, error mark options={line width=1.0pt, mark size=3.0pt, rotate=90}]
 table[row sep=crcr, y error plus index=2, y error minus index=3]{%
2	-126.923028728838	3.34946971824145	5.44152126627777\\
4	-120.901973981868	3.35093392623781	5.43576637530487\\
8	-114.879570420223	3.35390226434605	5.42450325932569\\
16	-108.851781491165	3.35999272394989	5.40288096600989\\
32	-102.802376437634	3.3727575604093	5.36263847274741\\
64	-96.6647557052045	3.40068622838531	5.28979011414512\\
};
\addplot [color=mycolor2, line width=1.0pt, forget plot]
 plot [error bars/.cd, y dir=both, y explicit, error bar style={line width=1.0pt}, error mark options={line width=1.0pt, mark size=3.0pt, rotate=90}]
 table[row sep=crcr, y error plus index=2, y error minus index=3]{%
2	-139.566288546725	3.34946971824144	5.44152126627779\\
4	-133.545233799755	3.35093392623781	5.43576637530484\\
8	-127.52283023811	3.35390226434605	5.42450325932568\\
16	-121.495041309052	3.35999272394992	5.40288096600987\\
32	-115.445636255521	3.37275756040931	5.36263847274741\\
64	-109.308015523092	3.4006862283853	5.28979011414513\\
};
\addplot [color=mycolor3, line width=1.0pt, forget plot]
 plot [error bars/.cd, y dir=both, y explicit, error bar style={line width=1.0pt}, error mark options={line width=1.0pt, mark size=3.0pt, rotate=90}]
 table[row sep=crcr, y error plus index=2, y error minus index=3]{%
2	-152.209548364612	3.34946971824144	5.44152126627779\\
4	-146.188493617642	3.35093392623781	5.43576637530484\\
8	-140.166090055997	3.35390226434603	5.42450325932569\\
16	-134.138301126939	3.35999272394994	5.40288096600986\\
32	-128.088896073408	3.3727575604093	5.36263847274742\\
64	-121.951275340979	3.4006862283853	5.28979011414512\\
};
\addplot [color=black, line width=1.0pt, forget plot]
 plot [error bars/.cd, y dir=both, y explicit, error bar style={line width=1.0pt}, error mark options={line width=1.0pt, mark size=3.0pt, rotate=90}]
 table[row sep=crcr, y error plus index=2, y error minus index=3]{%
2	-164.852808182499	3.34946971824144	5.44152126627779\\
4	-158.83175343553	3.35093392623784	5.43576637530478\\
8	-152.809349873884	3.353902264346	5.42450325932569\\
16	-146.781560944827	3.35999272394994	5.40288096600986\\
32	-140.732155891295	3.37275756040927	5.36263847274739\\
64	-134.594535158866	3.4006862283853	   5.28979011414512\\
};
\end{axis}

\node[right, align=left] at (0,4.1) {\scriptsize{\red{$\|\fv\|_2^2\sigma_\mathrm{U}^2$}}};

\draw (2.2,2.0) ellipse (0.2 and 1.28 );
\draw (1.5,3.4)[-Stealth] -- (2,3.05);

\node[right, align=left] at (0,3.55) {\scriptsize{{$\big\|\fv^\HH\Hm_\mathrm{IU}^\TT\big(\diagopr{\gammav}^{-1}-\Sm\big)^{-1}\big\|_2^2\sigma_\It^2$}}};

\draw (5,3.8)[-Stealth] -- (5,0.8);

\node[left, align=right] at (7.1,0.6) {\scriptsize{RIS-UE distance: $\{1,4,16,64\}\times \text{RD}$}};

\end{tikzpicture}%
    \vspace{-3em}
    \caption{A numerical evaluation of the noise terms received at the \ac{ue} in the far-field. The configuration includes a $10\times 10$ active \ac{ris} with a uniform amplification factor~$a$ across all unit cells and a $2\times 2$ \ac{ue}. Both the \ac{ris} and \ac{ue} arrays are half-wavelength spaced. RD is the Rayleigh distance calculated as $\text{RD}=2(D_\mathrm{I}+D_\mathrm{U})^2/\lambda$, where $D_\mathrm{I}$ and $D_\mathrm{U}$ represent the array apertures of the \ac{ris} and \ac{ue}, respectively, and $\lambda$ denotes the signal wavelength. The noise power is evaluated at a \unit[30]{GHz} mmWave frequency, where $\text{RD}\approx\unit[1]{m}$. In addition, we set $\sigma_\mathrm{U}^2=\sigma_\mathrm{I}^2=\unit[-95]{dBm}$ according to\cite[Fig.~25]{Wang2024Wideband}. }
    \label{fig:noisePower}
    \vspace{-1em}
\end{figure}

Figure~\ref{fig:noisePower} provides a numerical validation of Remark~\ref{rmk:noisePower}. Even with an extremely large RIS amplification factor, e.g., $a=64$, the additional noise power remains several orders of magnitude weaker than the noise power at the receiver. It is worth noting that a typical phase-reconfigurable reflection amplifier can achieve a maximum performance of approximately $a=10$ only~\cite[Fig.~9]{Rao2023An}. Consequently, the additional noise introduced by the active RIS can be safely neglected in far-field scenarios, allowing the receive \ac{snr} to be simplified as
\begin{align}
\mathrm{SNR} &\approx \frac{|\fv^\HH\Hm_\mathrm{UI}(\Gammam^{-1} -\Sm)^{-1}\Hm_\mathrm{IB}\wv|^2}{\|\fv\|_2^2\sigma_\mathrm{U}^2 },\notag\\
&= {\Big|\frac{\fv^\HH}{\|\fv\|_2}\Hm_\mathrm{UI}(\Gammam^{-1} -\Sm)^{-1}\Hm_\mathrm{IB}\wv\Big|^2}/{\sigma_\mathrm{U}^2 }.\label{eq:SNR2}
\end{align}
Hence, the joint beamforming problem can be formulated as
\begin{equation}\label{eq:JOPT_ori}
    \begin{aligned}
    \max_{\fv,\wv,\bm{\gamma}}&\quad {\big|\fv^\HH\overbrace{\Hm_\mathrm{UI}\big(\diagopr{\bm{\gamma}}^{-1} -\Sm\big)^{-1}\Hm_\mathrm{IB}}^{\Hm_\mathrm{UIB}}\wv\big|^2},\\
    \text{s.t.}&\quad |f_i|^2 = \frac{1}{N_\mathrm{U}},\ |w_j|^2 = \frac{P_\Rt}{N_\mathrm{B}},\ \forall i,j,\  \|\bm{\gamma}\|_2^2 \leq A. 
\end{aligned}
\end{equation}
Here, we use a value $A>N_\mathrm{I}$ to constrain the amplification vector $\bm{\gamma}$ of the active \ac{ris}, which corresponds to the maximum radiated power of the first-order radiation when a unit-power far-field signal impinges on the active RIS. For notational convenience, we define $\Hm_\mathrm{UI}\big(\diagopr{\bm{\gamma}}^{-1} -\Sm\big)^{-1}\Hm_\mathrm{IB} \triangleq \Hm_\mathrm{UIB}$. The problem is that $\Hm_\mathrm{UIB}$ is unavailable. Instead, we only have the uplink exact equivalent cascaded channel $\Gm_\mathrm{mc}^\mathrm{UL}$. The following steps establish connections between them.

 According to~\eqref{eq:Gmc}, we have $\Gm_\mathrm{mc}^\mathrm{UL}=\big(\Am_\mathrm{U}(\varphiv)\Sigmam_\mathrm{IU}\Am_\mathrm{I}^\TT(\phiv)\big) \otimes \big(\Am_\mathrm{B}(\varthetav)\Sigmam_\mathrm{BI}\Am_\mathrm{I}^\TT(\thetav)\big)$. Now, we define the downlink exact equivalent cascaded channel as $\Gm_\mathrm{mc}^\mathrm{DL} \triangleq \big(\Am_\mathrm{B}(\varthetav)\Sigmam_\mathrm{BI}\Am_\mathrm{I}^\TT(\thetav)\big) \otimes \big(\Am_\mathrm{U}(\varphiv)\Sigmam_\mathrm{IU}\Am_\mathrm{I}^\TT(\phiv)\big)$, which can be obtained from~$\Gm_\mathrm{mc}^\mathrm{UL}$ by rearranging the layout of its entries. Then, we have 
\begin{equation}\label{eq:H_UIB}
\mathrm{vec}\big(\Hm_\mathrm{UIB}\big) = \Gm_\mathrm{mc}^\mathrm{DL}\mathrm{vec}\big((\diagopr{\bm{\gamma}}^{-1} -\Sm)^{-1}\big).
\end{equation}
A proof is given in Appendix \ref{sec:UL_DL_reciprocity}. This relationship means that we can obtain the cascaded channel~$\Hm_\mathrm{UIB}$ from the previously estimated uplink equivalent channel~$\Gm_\mathrm{mc}^\mathrm{UL}$, and thus the joint beamforming problem~\eqref{eq:JOPT_ori} can be solved without the knowledge about the individual subchannels~$\Hm_\mathrm{UI}$ and $\Hm_\mathrm{IB}$.

In the following, we solve~\eqref{eq:JOPT_ori} using an alternating optimization approach. At each iteration, we first optimize $\{\fv, \wv\}$ while keeping $\bm{\gamma}$ fixed, and then optimize $\bm{\gamma}$ while keeping $\{\fv, \wv\}$ fixed. These two steps are repeated iteratively until the objective function converges.

\subsection{\ac{bs} Precoding and \ac{ue} Combining: A Closed-Form Solution}\label{sec:PreCom}
Given a feasible RIS configuration~$\bm{\gamma}$, we can obtain~$\Hm_\mathrm{UIB}$ according to~\eqref{eq:H_UIB}. Then, problem~\eqref{eq:JOPT_ori} is reduced to
\begin{equation}\label{eq:OPT1}
    \max_{\fv,\wv}\ |\fv^\HH\Hm_\mathrm{UIB}\wv|^2,\ 
     \text{s.t.}\ |f_i|^2 = \frac{1}{N_\mathrm{U}},\ |w_j|^2 = \frac{P_\Rt}{N_\mathrm{B}},\ \forall i,j. 
\end{equation}

Considering the analog beamforming constraint, this subproblem can be addressed by alternately projecting \(\fv^{(t+1)} = \frac{1}{\sqrt{N_\mathrm{U}}}e^{j\cdot\mathrm{arg}(\Hm_\mathrm{UIB}\wv^{(t)})}\) and \(\wv^{(t+1)} = \sqrt{\frac{P_\mathrm{B}}{N_\mathrm{B}}}e^{j\cdot\mathrm{arg}(\Hm_\mathrm{UIB}^\HH\fv^{(t+1)})}\), where \(t\) denotes the iteration index. Furthermore, when the digital beamforming is available, the constraints in~\eqref{eq:OPT1} become $\|\fv\|_2^2=1$ and $\|\wv\|_2^2=P_\mathrm{B}$. In this case, this subproblem can be solved by simply applying \ac{svd} to $\Hm_\mathrm{UIB}$. Let $\uv_\mathrm{max}$ and $\vv_\mathrm{max}$ denote the left- and right-singular vectors corresponding to the maximum singular value of $\Hm_\mathrm{UIB}$, respectively. The solution of is given by $\fv_\star=\uv_\mathrm{max}$ and $\wv_\star=\sqrt{P_\mathrm{B}}\vv_\mathrm{max}$.

\subsection{RIS Beamforming: Approximation and Reformulation}\label{sec:RISbf}
Based on~\eqref{eq:H_UIB}, we can rewrite the objective function of~\eqref{eq:JOPT_ori} as $\big|(\wv^\TT\otimes\fv^\HH)\Gm_\mathrm{mc}^\mathrm{DL}\mathrm{vec}\big((\diagopr{\bm{\gamma}}^{-1} -\Sm)^{-1}\big)\big|^2$.
Given a pair of feasible~$\{\fv,\wv\}$ and defining~$\tv \triangleq (\Gm_\mathrm{mc}^\mathrm{DL})^\TT(\wv\otimes\fv^*)\in\mathbb{C}^{N_\mathrm{I}^2\times 1}$, the joint problem~\eqref{eq:JOPT_ori} is reduced into
\begin{equation}\label{eq:OPT2}
    \begin{aligned}
    \max_{\bm{\gamma}}\ \big|\tv^\TT\mathrm{vec}\big((\diagopr{\bm{\gamma}}^{-1}-\Sm)^{-1}\big)\big|^2,\quad
    \text{s.t.}\ \|\bm{\gamma}\|_2^2 \leq A. 
\end{aligned}
\end{equation}
The matrix inversion in the objective function poses significant challenges for optimization. To circumvent this intractable process, we leverage the Neumann series expansion~\cite{Ortega2013Matrix}
\begin{equation}\label{eq:Neumann}
    (\Gammam^{-1}\!-\!\Sm)^{-1}\!=\!\Big(\sum_{n=0}^\infty \big(\Gammam\Sm\big)^n\Big)\Gammam
    \!=\!\Gammam + \Gammam\Sm\Gammam + \Gammam\Sm\Gammam\Sm\Gammam + \cdots.
\end{equation}
This series expansion has been widely utilized in the analysis of the mutual coupling in array signal processing~\cite{Wijekoon2024Phase,Qian2021Mutual,Abrardo2021MIMO}.
Note that this expansion holds only when all the eigenvalues of~$\Gammam\Sm$ are within the unit circle, i.e.,~$|\lambda_i(\Gammam\Sm)|<1$, $\forall\ i=1,\dots,N_\mathrm{I}$. Fortunately, such a condition can be satisfied and the expression is valid for a typical~\ac{mc} strength.  {Therefore, we can truncate the Neumann series in~\eqref{eq:Neumann} to obtain a linear approximation of $(\Gammam^{-1} - \Sm)^{-1}$. A typical evaluation of the truncation error for different orders of expansion can be found in, e.g.,~\cite[Fig.~10]{Zheng2024Mutual}.}
\begin{remark}
     {When only the first term in~\eqref{eq:Neumann} is retained, the model reduces to the conventional formulation without accounting for \ac{ris} mutual coupling. This can be verified by substituting $(\Gammam^{-1} - \Sm)^{-1} \approx \Gammam$ into~\eqref{eq:MC_ChM}, which then simplifies to the conventional model~\eqref{eq:convCHM}.}
\end{remark}

By retaining the first two terms of~\eqref{eq:Neumann}, $(\Gammam^{-1}-\Sm)^{-1} \approx \Gammam + \Gammam\Sm\Gammam$, and we can rewrite~\eqref{eq:OPT2} as
\begin{equation}\label{eq:OPT2_1}
    \begin{aligned}
    \min_{\bm{\gamma}}&\quad -|\tv^\TT\vect{\diagopr{\bm{\gamma}}+\diagopr{\bm{\gamma}}\Sm\diagopr{\bm{\gamma}}}|^2,\\
    \text{s.t.}&\quad \|\bm{\gamma}\|_2^2 \leq A.
\end{aligned}
\end{equation}
Notice the following relationship:
\begin{align}
    \tv^\TT\vect{\diagopr{\bm{\gamma}}} &\!=\! \big(\invdiag{\invvec{\tv}}\big)^\TT\bm{\gamma},\\
    \tv^\TT\vect{\diagopr{\bm{\gamma}}\Sm\diagopr{\bm{\gamma}}} &\!=\! \bm{\gamma}^\TT\big(\Sm\odot\invvec{\tv}\big)\bm{\gamma},
\end{align}
where $\invvec{\cdot}$ represents the operation of reshaping an $N^2 \times 1$ vector into an $N \times N$ square matrix while preserving the column-wise order, and $\invdiag{\cdot}$ refers to the extraction of the diagonal entries of an $N \times N$ square matrix to form a $N \times 1$ column vector. 
Defining~$\qv\triangleq\invdiag{\invvec{\tv^*}}$, $\Bm\triangleq\Sm\odot\invvec{\tv}$, we can rewrite~\eqref{eq:OPT2_1} as
\begin{equation}
\begin{aligned}\label{eq:OPT2_2}
    \min_{\bm{\gamma}}&\ -(\qv^\HH\bm{\gamma}+\bm{\gamma}^\TT\Bm\bm{\gamma})^\HH(\qv^\HH\bm{\gamma}+\bm{\gamma}^\TT\Bm\bm{\gamma}),\\
    \text{s.t.}&\quad \|\bm{\gamma}\|_2^2 \leq A. 
\end{aligned}
\end{equation}

Although simplified, problem~\eqref{eq:OPT2_2} is still challenging due to its non-convexity. In Section~\ref{sec:RISCO}, we will adopt the \ac{sca} framework to solve this problem.

\section{RIS Configuration Optimization via SCA}\label{sec:RISCO}

This section presents an iterative solution to the non-convex optimization problem~\eqref{eq:OPT2_2} utilizing the \ac{sca} framework. A brief review of the \ac{sca} is offered in Appendix~\ref{sec:SCA}.

\subsection{Successive Convex Approximation of~\eqref{eq:OPT2_2}}\label{sec:RISCOSCA}
Based on the \ac{sca} principle presented in Appendix~\ref{sec:SCA}, the following Proposition~\ref{prop:SCAsuro} proposes a surrogate function for the objective function in~\eqref{eq:OPT2_2}.
\begin{proposition}\label{prop:SCAsuro}
  For objective function $f(\bm{\gamma}) = -(\qv^\HH\bm{\gamma}+\bm{\gamma}^\TT\Bm\bm{\gamma})^\HH(\qv^\HH\bm{\gamma}+\bm{\gamma}^\TT\Bm\bm{\gamma})$ in~\eqref{eq:OPT2_2}, given an arbitrary feasible point~$\bm{\gamma}_i$, an associated surrogate function that satisfies Condition~\ref{cond:1} and Condition~\ref{cond:2} in Appendix~\ref{sec:SCA} is
\begin{multline}\label{eq:g}
    g(\bm{\gamma}|\bm{\gamma}_i) = -\bm{\gamma}^\HH\qv\qv^\HH\bm{\gamma} - \bm{\gamma}^\HH\qv\bm{\gamma}_i^\TT\Bm\bm{\gamma} - \bm{\gamma}^\HH\Bm^\HH\bm{\gamma}_i^*\qv^\HH\bm{\gamma} \\ - \bm{\gamma}^\HH\Bm^\HH\bm{\gamma}_i^*\bm{\gamma}_i^\TT\Bm\bm{\gamma} + K\|\bm{\gamma}-\bm{\gamma}_i\|_2^2 - \bm{\gamma}_i^\HH\qv\bm{\gamma}^\TT\Bm\bm{\gamma}_i \\ \! -\! \bm{\gamma}_i^\HH\Bm^\HH\bm{\gamma}^*\qv^\HH\bm{\gamma}_i\! -\! \bm{\gamma}_i^\HH\Bm^\HH\bm{\gamma}^*\bm{\gamma}_i^\TT\Bm\bm{\gamma}_i\! -\! \bm{\gamma}_i^\HH\Bm^\HH\bm{\gamma}_i^*\bm{\gamma}^\TT\Bm\bm{\gamma}_i.
\end{multline}
Here,~$K\in\mathbb{R}^+$ is a parameter that must be chosen to ensure
\begin{equation}\label{eq:KBi}
K \geq \lambda_\mathrm{max}(\Qm_i),
\end{equation}
where 
\begin{equation}\label{eq:Bi}
    \Qm_i = \qv\qv^\HH + \qv\bm{\gamma}_i^\TT\Bm + \Bm^\HH\bm{\gamma}_i^*\qv^\HH + \Bm^\HH\bm{\gamma}_i^*\bm{\gamma}_i^\TT\Bm.
\end{equation}
Because~$\Qm_i$ is a Hermitian matrix, all its eigenvalues are real, resulting in a real~$K$ in~\eqref{eq:KBi}.
\end{proposition}
\begin{proof}
    See Appendix~\ref{app:SCAsuro}.
\end{proof}

\subsection{Solving the Surrogate Optimization Problem}\label{sec:solSurOpt}

For the second step of \ac{sca}, we now consider optimizing the surrogate function~$g(\bm{\gamma}|\bm{\gamma}_i)$. Based on~\eqref{eq:g}, this surrogate optimization problem can be formulated as
\begin{equation}\label{eq:suropt}
    \min_{\bm{\gamma}}\ \bm{\gamma}^\HH(K\mathbf{I}-\Qm_i)\bm{\gamma}-\bv_i^\TT\bm{\gamma}-\bv_i^\HH\bm{\gamma}^*,\quad \text{s.t.}\ \|\bm{\gamma}\|_2^2 \leq A,  
\end{equation}
where $\bv_i = K\bm{\gamma}_i^* + (\bm{\gamma}_i^\HH\qv + \bm{\gamma}_i^\HH\Bm^\HH\bm{\gamma}_i^*)\Bm\bm{\gamma}_i$. Note that we omit the constant term in the objective function of~\eqref{eq:suropt}.

Since the objective function~$g(\bm{\gamma}|\bm{\gamma}_i)$ and the inequality constraint~$\|\bm{\gamma}\|_2^2-A$ in~\eqref{eq:suropt} are convex, and there exists strictly feasible points satisfying~$\|\bm{\gamma}\|_2^2 < A$ (i.e., the Slater's condition is satisfied~\cite[Sec. 5.2.3]{Boyd2004Convex}), the strong duality holds for~\eqref{eq:suropt}. Therefore, the optimal solution of~\eqref{eq:suropt} can be found based on the \ac{kkt} conditions.

To begin  with, the Lagrangian of~\eqref{eq:suropt} is given by
\begin{equation}
    L(\bm{\gamma},\nu) = \bm{\gamma}^\HH(K\mathbf{I}-\Qm_i)\bm{\gamma}-\bv_i^\TT\bm{\gamma}-\bv_i^\HH\bm{\gamma}^* + \nu(\bm{\gamma}^\HH\bm{\gamma}-A),
\end{equation}
where~$\nu\in\mathbb{R}_{\geq 0}$ is the Lagrange multiplier (or dual variable) associated with the inequality constraint. 
Then, the \ac{kkt} conditions indicate that the optimal primal and dual points, denoted as~$(\bm{\gamma}_\mathrm{o},\nu_\mathrm{o})$, must satisfy the following conditions:
\begin{subequations}
    \begin{align}
\nabla_{\scriptsize{\bm{\gamma}}}L(\bm{\gamma}_\mathrm{o},\nu_\mathrm{o} )&=\mathbf{0},\label{eq:grad0}\\
    \nu_\mathrm{o}(\bm{\gamma}_\mathrm{o}^\HH\bm{\gamma}_\mathrm{o}-A) &= 0,\label{eq:slack}\\
    \nu_\mathrm{o} &\geq 0,\label{eq:nu0}\\
    \bm{\gamma}_\mathrm{o}^\HH\bm{\gamma}_\mathrm{o}-A &\leq 0.\label{eq:fi}
\end{align}
\end{subequations}

Based on~\eqref{eq:grad0}, we have
\begin{equation}
\frac{\partial L(\bm{\gamma}_\mathrm{o},\nu_\mathrm{o} )}{\partial \bm{\gamma}_\mathrm{o}^*}\! =\! \Big( \frac{\partial L(\bm{\gamma}_\mathrm{o},\nu_\mathrm{o} )}{\partial \bm{\gamma}_\mathrm{o}} \Big)^*\!\! =\! (K\mathbf{I}-\Qm_i)\bm{\gamma}_\mathrm{o}-\bv_i^*+\nu_\mathrm{o}\bm{\gamma}_\mathrm{o}=\mathbf{0}.\notag
\end{equation}
This immediately yields
\begin{equation}\label{eq:gamma0}
    \bm{\gamma}_\mathrm{o}=\big((K+\nu_\mathrm{o})\mathbf{I}-\Qm_i\big)^{-1}\bv_i^*. 
\end{equation}
The conditions in~\eqref{eq:KBi} and~\eqref{eq:nu0} guarantee that all eigenvalues of~$\big((K+\nu_\mathrm{o})\mathbf{I}-\Qm_i\big)$ lie in~$\mathbb{R}_{\geq 0}$; thus, this matrix is always invertible. Subsequently, $\bm{\gamma}_\mathrm{o}$ can be determined through a bisection search by combining~\eqref{eq:slack}--\eqref{eq:gamma0}.

\subsection{Solving~\eqref{eq:OPT2_2} Using SCA}

Based on the convex approximation in Section~\ref{sec:RISCOSCA} and the corresponding solution in Section~\ref{sec:solSurOpt}, we summarize the complete \ac{sca} procedure for solving~\eqref{eq:OPT2_2} in Algorithm~\ref{algo:SCA}, which follows the \ac{sca} principle in Appendix~\ref{sec:SCA}. As mentioned, the step size~$\eta_i$ of the iteration can be chosen according to different rules. As an example, here we present a line search rule~\cite{Nedic2018Parallel} for determining~$\eta_i$. Again, let~$f(\bm{\gamma})$ denote the objective function in~\eqref{eq:OPT2_2} and choose~$\delta_1,\delta_2\in(0,1)$. The line search rule chooses~$\eta_i=\delta_1^{n_i}$, where~$n_i$ is the smallest natural number satisfying
\begin{equation}\label{eq:stepSize}
    f(\bm{\gamma}_i+\eta_i\Delta\bm{\gamma}_i) \leq f(\bm{\gamma}_i) + \delta_2\eta_i\Big(\big(\nabla_{\bm{\gamma}^*} f(\bm{\gamma}_i)\big)^\HH\Delta\bm{\gamma}_i\Big),
\end{equation}
where~$\Delta\bm{\gamma}_i\triangleq\tilde{\bm{\gamma}}_{i+1}-\bm{\gamma}_i$ and~$\tilde{\bm{\gamma}}_{i+1}$ denotes the solution of~\eqref{eq:suropt}. Here,~$\nabla_{\bm{\gamma}^*} f(\bm{\gamma}_i)=\partial f(\bm{\gamma}_i)/\partial \bm{\gamma}^*$~\cite[Theorem 3.4]{Hjorungnes2011Complex}.

\begin{algorithm}[t]
 \caption{Solving~\eqref{eq:OPT2_2} Using SCA}
 \label{algo:SCA}
 \begin{algorithmic}[1]
 \State \textbf{Input:} $\qv,\Bm,A$.\qquad \textbf{Output:} $\bm{\gamma}_\star$.
 \State Initialize~$i=0$ and randomly select~$\bm{\gamma}_0$,~$\|\bm{\gamma}_0\|^2\leq A$.
 \For{$i\in\mathbb{N}$}
 \State Compute $\Qm_i$ using~\eqref{eq:Bi} and select~$K \geq \lambda_\mathrm{max}(\Qm_i)$.
 \State  Compute $\bv_i\! =\! K\bm{\gamma}_i^*\! +\! (\bm{\gamma}_i^\HH\qv\! +\! \bm{\gamma}_i^\HH\Bm^\HH\bm{\gamma}_i^*)\Bm\bm{\gamma}_i$.
 \State Apply bisection to solve~\eqref{eq:gamma0} and obtain~$\bm{\gamma}_\mathrm{o}$.
 \State Let~$\tilde{\bm{\gamma}}_{i+1}=\bm{\gamma}_\mathrm{o}$ and choose a step size~$\eta_i$ through~\eqref{eq:stepSize}.
 \State Update $\bm{\gamma}_{i+1} = \bm{\gamma}_{i} + \eta_i (\tilde{\bm{\gamma}}_{i+1}-\bm{\gamma}_i)$.
 \State Stop if~$\|\bm{\gamma}_{i+1}-\bm{\gamma}_{i}\|_2$ is lower than a preset threshold.
 \EndFor
 \Return $\bm{\gamma}_\star=\bm{\gamma}_{i+1}$.
 \end{algorithmic} 
\end{algorithm}

 {
\subsection{Algorithm Summary and Complexity Analysis}

The proposed beamforming algorithm solves the optimization problem~\eqref{eq:JOPT_ori} by decomposing it into two subproblems,~\eqref{eq:OPT1} and~\eqref{eq:OPT2}, which are solved alternately.

\begin{itemize}
    \item \textbf{Solving~\eqref{eq:OPT1}:}  
    According to Section~\ref{sec:PreCom}, this step involves an alternating projection, whose computational complexity is given by $\mathcal{O}(I_\mathrm{PRJ}N_\mathrm{U}N_\mathrm{B})$ with $I_\mathrm{PRJ}$ denoting the total number of iterations. 
    
    \item \textbf{Solving~\eqref{eq:OPT2}:}  
    According to Section~\ref{sec:RISbf}, we first compute the quantities $\{\tv, \qv, \Bm\}$, which incurs a complexity of $\mathcal{O}(N_\mathrm{I}^2 N_\mathrm{U} N_\mathrm{B})$ in total. Then, we perform the \ac{sca} procedure as outlined in Algorithm~\ref{algo:SCA}. Specifically, steps 4, 5, 6, and 8 in Algorithm~\ref{algo:SCA} incur computational costs of $\mathcal{O}(N_\mathrm{I}^3)$, $\mathcal{O}(N_\mathrm{I}^2)$, $\mathcal{O}(I_\mathrm{BIS} N_\mathrm{I}^3)$, and $\mathcal{O}(N_\mathrm{I})$, respectively, where $I_\mathrm{BIS}$ denotes the number of bisection search required to determine $\gammav_\mathrm{o}$. We omit the complexity of step size selection in Step~7, as it varies with the choice of step size selection strategy. Let $I_\mathrm{SCA}$ denote the number of SCA iterations. The overall complexity of Algorithm~\ref{algo:SCA} is thus $\mathcal{O}(I_\mathrm{SCA} I_\mathrm{BIS} N_\mathrm{I}^3)$, leading to a total complexity for solving~\eqref{eq:OPT2} of
    $\mathcal{O}(N_\mathrm{I}^2 N_\mathrm{U} N_\mathrm{B} + I_\mathrm{SCA} I_\mathrm{BIS} N_\mathrm{I}^3)$.
\end{itemize}
Assuming the alternating optimization between~\eqref{eq:OPT1} and~\eqref{eq:OPT2} is repeated $I_\mathrm{ALT}$ times, the overall computational complexity of the proposed algorithm is then given by $\mathcal{O}\big(I_\mathrm{ALT}(I_\mathrm{PRJ}N_\mathrm{U}N_\mathrm{B} + N_\mathrm{I}^2 N_\mathrm{U} N_\mathrm{B} + I_\mathrm{SCA} I_\mathrm{BIS} N_\mathrm{I}^3)\big)$.
}

\section{Simulation results and discussion}
\label{sec:Sim_Res_Disc}

\subsection{Simulation Setup}

\subsubsection{System Configuration} \label{sec:sys_confg}
Throughout the simulations, we position the \ac{ue} at~$\unit[2.6]{m} \times [\sin{\frac{\pi}{6}}, \cos{\frac{\pi}{6}}, 0]^\TT$, the \ac{ris} at~$[0, 0, 0]^\TT$ facing the positive Y-axis, and the \ac{bs} at~$\unit[2.2]{m} \times [\sin{\frac{\pi}{3}}, \cos{\frac{\pi}{3}}, 0]^\TT$ facing the negative Y-axis, consistent with the far-field measurement setup in~\cite{Wang2024Wideband}. We set $f_c=\unit[30]{GHz}$. The~\ac{ris} is composed of $N_\It=16\times8$ elements with $d_\It=\lambda/20$ while the~\ac{bs} and~\ac{ue} deploy $N_\Rt=4\times2$ and $N_\Tt=2\times1$ with $d_\Rt\!=\!d_\Tt\!=\!\lambda/2$. We use  $L_{\Tt}=2$, the \ac{los} path gain in~\eqref{eq:ch_TxRIS_spatial_domain} is $\alpha_1=(\lambda/(4\pi d_{\It\Tt}))^{\gamma_{\It\Tt}/2}$, where $d_{\It\Tt}$ is the distance between the~\ac{ue} and~\ac{ris}, and the path loss exponent $\gamma_{\It\Tt}=2.1$. For \ac{nlos} paths $\alpha_\ell\sim\mathcal{CN}(0,\alpha_1^2)$. The same logic applies to $\rho_\ell$ in~\eqref{eq:ch_RISRx_spatial_domain} with $L_\Rt=2$. Unless stated otherwise, we set~$A\!=\!896$, i.e., the RIS has an average amplification factor~$\bar{a}\!=\!\sqrt{{A}/{N_\It}}\!=\!7$.

\subsubsection{Mutual Coupling Generation} 
 {The scattering matrix~$\Sm$ is obtained by first determining the mutual impedances ($\Zm$ matrix) between \ac{ris} unit cells and then converting it to scattering parameters ($\Sm$ matrix) by applying~$\Sm = (\Zm + Z_0\mathbf{I})^{-1}(\Zm - Z_0\mathbf{I})$~\cite[Eq.~(24)]{Abrardo2024Design}, where $Z_0 = \unit[50]{\Omega}$ is the characteristic impedance. In general, the~$\Zm$ matrix can be obtained through full-wave electromagnetic simulations using tools such as Ansys HFSS. In this study, we use a simple analytical approach by approximating all \ac{ris} unit cells as parallel cylindrical thin wires of perfectly conducting material. Under this assumption, the mutual impedance between the~$i^\text{th}$ and $j^\text{th}$ unit cells, i.e., the entry $Z_{i,j}$ (and $Z_{j,i}$) can be analytically calculated using~\cite[Eq.~(2)]{Zheng2024On}. This analytical mutual impedance model is governed by the intrinsic impedance of free space~$\eta_0$, the signal wavenumber $k_0$, the length and radius of the equivalent cylindrical thin wire ($h$ and $r$), and the separation distance between the two wires. In the following simulations, we fix $\eta_0 = \unit[377]{\Omega}$, $k_0 = 2\pi / \lambda$, $h = \lambda / 32$, and $r = \lambda / 500$, consistent with the setup in~\cite{Gradoni2021End}, and vary only the separation distance between \ac{ris} unit cells. Thus, $\Sm$ is a function of \ac{ris} spacing, with denser element integration typically resulting in stronger mutual coupling\cite[Fig.~2]{Zheng2024On}.} 

\subsubsection{Compressed Sensing Parameters}
The training uses random phase shifts, uniformly distributed as $\mathcal{U}(0,2\pi)$. We perform $100$ trials to average each simulation point in the results (e.g., transmit power, \ac{ris} inter-element spacing, etc.), where at each trial different random beamformer, combiner, and \ac{ris} configurations are used, following the training procedure in Section \ref{sec:rx_sig}. We set $M_\Rt={3 N_\Rt}/{4}$ $M_\It={3 N_\It}/{4}$, $\hat{L}=5$, $G_\It=N_\It$, and $G_\Rt=N_\Rt$. It is important to clarify that our simulations consider arbitrary \ac{aoa} and \ac{aod}, as this provides a more practical and realistic scenario. Consequently, we anticipate a reduction in channel estimation accuracy due to the off-grid power leakage effect associated with the used on-grid \ac{omp}. However, this performance loss can be mitigated by selecting larger dictionaries (e.g., $G_\It=2 N_\It, G_\Rt=2 N_\Rt$) or by leveraging more advanced off-grid estimation methods.
 
\subsubsection{Noise Setup}
 {Although our methodologies are developed under the assumption that the additional noise introduced by the active \ac{ris} is negligible (see Remark~\ref{rmk:noisePower}), all subsequent simulations are conducted under realistic conditions where noise is present at both the active \ac{ris} and the \ac{ue} consistently, as described in~\eqref{eq:totalNoise}. Specifically, we set $\sigma_\Rt^2 = \sigma_\It^2 = \unit[-95]{dBm}$, following the measured result in~\cite[Fig.~25]{Wang2024Wideband}.
}

\subsection{Performance Metrics}
\subsubsection{Channel Estimation}
Based on Section~\ref{sec:rx_sig_mc} and Section~\ref{sec:CE_MC}, we can estimate the conventional and exact equivalent cascaded channel~$\hat{\Gm}_\mathrm{cv}\in\mathbb{C}^{N_\Rt\times N_\It}$ and~$\hat{\Gm}_\mathrm{mc}\in\mathbb{C}^{N_\Rt\times N_\It^2}$ based on the conventional \ac{mc}-unaware model and the exact \ac{mc}-awre model, respectively. Since~$\hat{\Gm}_\mathrm{cv}$ and~$\hat{\Gm}_\mathrm{mc}$ have different dimensions, for a fair comparison, we evaluate the estimation accuracy by computing the \ac{nmse} of the reconstructed received signal as
\begin{align}
    \text{NMSE}_\mathrm{cv} &= \mathbb{E}\|\Pm\hat{\Gm}_\mathrm{cv}\Thetam_\convidx-\bar{\Ym}\|_\mathrm{F}^2/\|\bar{\Ym}\|_\mathrm{F}^2,\\
    \text{NMSE}_\mathrm{mc} &= \mathbb{E}\|\Pm\hat{\Gm}_\mathrm{mc}\Thetam_\mathrm{mc}-\bar{\Ym}\|_\mathrm{F}^2/\|\bar{\Ym}\|_\mathrm{F}^2,
\end{align}
where~$\bar{\Ym}$ is the noise-free version of the received signal.

\subsubsection{Beamforming}
For the beamforming performance, we evaluate the spectral efficiency of given~$\{\fv,\wv,\bm{\gamma}\}$ as
$\text{SE}=\log_2 (1 + \mathrm{SNR})$ using the accurate definition of SNR in~\eqref{eq:SNR}.

\subsubsection{Computational Complexity}
 {For a numerical evaluation of the computational complexity of Algorithm~\ref{algo:prop_ch_est}, the wall-clock running time has been simulated.} 

\subsection{Channel Estimation Performance Evaluation}

\begin{figure}[t]
    \centering
\definecolor{mycolor1}{rgb}{1.00000,0.00000,1.00000}%
\definecolor{mycolor2}{rgb}{0.00000,1.00000,1.00000}%
\definecolor{ForestGreen}{rgb}{0.1333    0.8451    0.1333}
\begin{tikzpicture}

\begin{axis}[%
width=2.8in,
height=1.6in,
at={(0in,0in)},
scale only axis,
xmin=-8.4,
xmax=12.4,
xlabel style={font=\color{white!15!black},font=\footnotesize},
xticklabel style = {font=\color{white!15!black},font=\footnotesize},
xlabel={UE transmit power %$P_\Tt^\mathrm{UL}$ 
$P_\Tt$ (dBm)},
ymin=-12,
ymax=10,
ylabel style={font=\color{white!15!black},font=\footnotesize},
yticklabel style = {font=\color{white!15!black},font=\footnotesize},
ylabel={Estimation NMSE (dB)},
ytick = {-12,-10, -8,-6,-4, -2, 0,2,4,6,8, 10},
axis background/.style={fill=white},
xmajorgrids,
ymajorgrids,
grid style={dashed},
legend style={at={(1,1)}, anchor=north east, legend cell align=left, align=left, draw=white!15!black, font=\tiny, fill opacity=0.85}
]
\addplot [color=ForestGreen, mark=triangle, mark options={solid, rotate=90, ForestGreen}, line width=0.7pt, mark size=2.5pt]
  table[row sep=crcr]{%
-8	8.91493791639805\\
-6	7.31145299275716\\
-4	4.86803406526645\\
-2	3.45948322912057\\
0	2.02116206480035\\
2	1.38258584787448\\
4	0.656718323007226\\
6	0.252501189506923\\
8	0.0979069166195889\\
10	-0.231107672614356\\
12	-0.291082631129151\\
};
\addlegendentry{MC-unaware OMP}

%\addplot [color=blue, mark=triangle, mark options={solid, rotate=270, blue}, line width=0.7pt, mark size=2.5pt]
%  table[row sep=crcr]{%
%0	11.1542846679687\\
%5	3.95701597332954\\
%10	-0.688697168976068\\
%15	-2.88864563703537\\
%20	-3.7805641412735\\
%25	-4.04078487157822\\
%};
%\addlegendentry{Dic 2: Improved Conventional}

%\addplot [color=mycolor1, mark=triangle, mark options={solid, mycolor1}, line width=0.7pt, mark size=2.5pt]
%  table[row sep=crcr]{%
%-6	-0.218935156799853\\
%-2	-8.07203350067139\\
%2	-14.8109446048737\\
%6	-20.2021180152893\\
%10	-22.9818064689636\\
%14	-25.978032207489\\
%};
%\addlegendentry{Proposed $(\rho_\DicRed=0.04)$}

%\addplot [color=mycolor2, mark=diamond, mark options={solid, mycolor2}, line width=0.7pt, mark size=2.7pt]
%  table[row sep=crcr]{%
%-6	-0.0984697423875332\\
%-2	-7.7058730840683\\
%2	-14.6432003974915\\
%6	-20.3586381912231\\
%10	-23.8991782188416\\
%14	-26.141183757782\\
%};
%\addlegendentry{Proposed $(\rho_\DicRed=0.07)$}

\addplot [color=blue, mark=o, mark options={solid, rotate=270, blue}, line width=0.7pt, mark size=2.3pt]
  table[row sep=crcr]{%
-8	3.62947482566039\\
-6	2.28565771778425\\
-4	-0.94691002368927\\
-2	-3.17317625035842\\
0	-6.21538640807072\\
2	-6.53053322806954\\
4	-7.79871680736542\\
6	-9.02228455940882\\
8	-9.13859035571416\\
10	-9.78759246667226\\
12	-10.0753294030825\\
};
\addlegendentry{Proposed $(\rho_{\scalebox{0.7}{$\mathrm{DR}$}}=0.1)$}

\addplot [color=red, mark=triangle, mark options={solid, rotate=180, red}, line width=0.7pt, mark size=2.5pt]
  table[row sep=crcr]{%
-8	4.65935019105673\\
-6	2.74567199349403\\
-4	-0.691081937526663\\
-2	-3.30108859737714\\
0	-6.31810194465021\\
2	-6.9645220408837\\
4	-8.71545544366042\\
6	-10.1455333928267\\
8	-10.3970056692759\\
10	-10.8433075586955\\
12	-11.509513648351\\
};
\addlegendentry{Proposed $(\rho_{\scalebox{0.7}{$\mathrm{DR}$}}=1)$}

\addplot [color=black, mark=square, mark options={solid, black}, line width=0.7pt, mark size=2pt]
  table[row sep=crcr]{%
-8	4.86541893593967\\
-6	2.80965921580791\\
-4	-0.454972085418801\\
-2	-3.02313821191589\\
0	-6.23800949553649\\
2	-6.9009876092275\\
4	-8.7336779375871\\
6	-10.0505542516708\\
8	-10.2826570232709\\
10	-10.9251841465632\\
12	-11.4283880710602\\
};
\addlegendentry{Proposed (no dict. reduction)}

\end{axis}

\end{tikzpicture}%
    \vspace{-3em}
    \caption{ {Performance evaluation of channel estimation accuracy assessed by NMSE versus transmit power.}}
    \label{fig:chest_perf}
\end{figure}

\begin{figure}[t]
    \centering
\definecolor{ForestGreen}{rgb}{0.1333    0.8451    0.1333}
\begin{tikzpicture}

\begin{axis}[%
width=2.8in,
height=1.6in,
at={(0in,0in)},
scale only axis,
xmin=0.92,
xmax=10.08,
xminorticks=true,
xlabel style={font=\color{white!15!black},font=\footnotesize},
xticklabel style = {font=\color{white!15!black},font=\footnotesize},
xlabel={RIS average amplification factor~$\bar{a}$},
ymin=-15,
ymax=30,
ylabel style={font=\color{white!15!black},font=\footnotesize},
yticklabel style = {font=\color{white!15!black},font=\footnotesize},
ytick = {-15,-10,-5,0,5,10,15,20,25,30},
ylabel={Estimation NMSE (dB)},
axis background/.style={fill=white},
xmajorgrids,
xminorgrids,
ymajorgrids,
grid style={dashed},
legend style={at={(1,1)}, anchor=north east, font=\tiny, legend cell align=left, align=left, draw=white!15!black, fill opacity=0.85}
]
\addplot [color=ForestGreen, mark=triangle, mark options={solid, rotate=90, ForestGreen}, line width=0.7pt, mark size=2.5pt]
  table[row sep=crcr]{%
1	31.0573861757914\\
2	18.3494810740153\\
3	9.96512612104416\\
4	3.58978089249382\\
5	1.55596935003996\\
6	0.659902322826868\\
7	0.771998855564743\\
8	0.770412800212701\\
9	0.702386105634893\\
10	0.665423926338553\\
};
\addlegendentry{MC-unaware OMP}

\addplot [color=blue, mark=o, mark options={solid, rotate=270, blue}, line width=0.7pt, mark size=2.3pt]
  table[row sep=crcr]{%
1	29.6470118204753\\
2	16.9415745099386\\
3	7.94420635700226\\
4	-1.94575649052858\\
5	-6.03246349642674\\
6	-8.82099313785633\\
7	-7.94359578688939\\
8	-8.44650224049886\\
9	-8.51650192141533\\
10	-8.72549512585004\\
};
\addlegendentry{Proposed $(\rho_{\scalebox{0.7}{$\mathrm{DR}$}}=0.1)$}

\addplot [color=red, mark=triangle, mark options={solid, rotate=180, red}, line width=0.7pt, mark size=2.5pt]
  table[row sep=crcr]{%
1	30.5938053766886\\
2	18.0109385649363\\
3	8.8261683344841\\
4	-1.20160914571024\\
5	-6.18522026172529\\
6	-9.87347900966803\\
7	-8.56952850818634\\
8	-8.7691141863664\\
9	-9.04659142494202\\
10	-9.44289788405101\\
};
\addlegendentry{Proposed $(\rho_{\scalebox{0.7}{$\mathrm{DR}$}}=1)$}

\addplot [color=black, mark=square, mark options={solid, black}, line width=0.7pt, mark size=2pt]
  table[row sep=crcr]{%
1	30.6991898854574\\
2	18.3068394978841\\
3	8.99083162546158\\
4	-0.959313549567014\\
5	-6.11403736670812\\
6	-9.9188027381897\\
7	-8.50609925190608\\
8	-8.84719509780407\\
9	-9.09682326714198\\
10	-9.42476406097412\\
};
\addlegendentry{Proposed (no dict. reduction)}

\end{axis}

\end{tikzpicture}%
    \vspace{-3em}
    \caption{Evaluation of the channel estimation NMSE versus~\ac{ris} average amplification factor~$\bar{a}=\sqrt{A/N_\It}$. As highlighted in Remark~\ref{rmk:RISamp}, higher amplification intensifies the impact of MC.}
    \label{fig:nmse_vs_rispower}
\end{figure}

\begin{figure}[t]
    \centering
\definecolor{ForestGreen}{rgb}{0.1333    0.8451    0.1333}
\begin{tikzpicture}

\begin{axis}[%
width=2.8in,
height=1.6in,
at={(0in,0in)},
scale only axis,
xmode=log,
xmin=0.019,
xmax=0.525,
xminorticks=true,
xlabel style={font=\color{white!15!black},font=\footnotesize},
xticklabel style = {font=\color{white!15!black},font=\footnotesize},
xtick = {1/50,1/32,1/16,1/8,1/4,1/2},
xticklabels = {$\frac{\lambda}{50}$,$\frac{\lambda}{32}$,$\frac{\lambda}{16}$,$\frac{\lambda}{8}$,$\frac{\lambda}{4}$,$\frac{\lambda}{2}$},
xlabel={RIS inter-element spacing $d_\It$},
ymin=-26.5,
ymax=1.5,
ylabel style={font=\color{white!15!black},font=\footnotesize},
yticklabel style = {font=\color{white!15!black},font=\footnotesize},
ytick = {0,-5,-10,-15,-20,-25},
ylabel={Estimation NMSE (dB)},
axis background/.style={fill=white},
xmajorgrids,
xminorgrids,
ymajorgrids,
grid style={dashed},
legend style={at={(1,0)}, anchor=south east, font=\tiny, legend cell align=left, align=left, draw=white!15!black, fill opacity=0.85}
]
\addplot [color=ForestGreen, mark=triangle, mark options={solid, rotate=90, ForestGreen}, line width=0.7pt, mark size=2.5pt]
  table[row sep=crcr]{%
0.02	0.585715701027463\\
0.03125	0.72443260626557\\
0.0364540324867536	0.704770257680987\\
0.0425246875054493	0.621582712264111\\
0.0496062828740062	0.315362680594747\\
0.0578671695179556	-0.0746293331185977\\
0.0675037336807691	-0.525016495088736\\
0.0787450656184295	-1.23595455723504\\
0.0918584057672249	-1.64262316580862\\
0.107155497856634	-1.17472109906375\\
0.125	-0.574179835617542\\
0.125	-0.574179835617542\\
0.198425131496025	0.139289692789316\\
0.314980262473718	0.102071940992028\\
0.5	-0.245243303105235\\
};
\addlegendentry{MC-unaware OMP}

\addplot [color=blue, mark=o, mark options={solid, rotate=270, blue}, line width=0.7pt, mark size=2.3pt]
  table[row sep=crcr]{%
0.02	-24.8739119052887\\
0.03125	-16.8302205403646\\
0.0364540324867536	-13.8943096876144\\
0.0425246875054493	-11.0279143214226\\
0.0496062828740062	-9.14528695742289\\
0.0578671695179556	-8.52010598182678\\
0.0675037336807691	-4.02682088704314\\
0.0787450656184295	-3.19644176103175\\
0.0918584057672249	-2.13309362779061\\
0.107155497856634	-1.52907127117117\\
0.125	-1.22129179984331\\
0.125	-1.22129179984331\\
0.198425131496025	-0.313799666333944\\
0.314980262473718	0.530698454628388\\
0.5	-0.173811126997073\\
};
\addlegendentry{Proposed $(\rho_{\scalebox{0.7}{$\mathrm{DR}$}}=0.1)$}

\addplot [color=red, mark=triangle, mark options={solid, rotate=180, red}, line width=0.7pt, mark size=2.5pt]
  table[row sep=crcr]{%
0.02	-24.7631267547607\\
0.03125	-17.7626951376597\\
0.0364540324867536	-14.6164686203003\\
0.0425246875054493	-11.8089305202166\\
0.0496062828740062	-9.8973685224851\\
0.0578671695179556	-9.6983817845583\\
0.0675037336807691	-4.12320452829202\\
0.0787450656184295	-2.84899874658634\\
0.0918584057672249	-2.25258214579274\\
0.107155497856634	-1.31983764742812\\
0.125	-0.632620182385047\\
0.125	-0.632620182385047\\
0.198425131496025	0.0497777496774991\\
0.314980262473718	0.092536375268052\\
0.5	-0.166971524929007\\
};
\addlegendentry{Proposed $(\rho_{\scalebox{0.7}{$\mathrm{DR}$}}=1)$}

\addplot [color=black, mark=square, mark options={solid, black}, line width=0.7pt, mark size=2pt]
  table[row sep=crcr]{%
0.02	-24.6560253302256\\
0.03125	-17.8060419400533\\
0.0364540324867536	-14.4985753854116\\
0.0425246875054493	-11.6159342249235\\
0.0496062828740062	-9.70185931324959\\
0.0578671695179556	-9.56391948089004\\
0.0675037336807691	-3.94360655620694\\
0.0787450656184295	-2.80790382967098\\
0.0918584057672249	-2.02728079647447\\
0.107155497856634	-1.21576800122857\\
0.125	-0.487823507810632\\
0.125	-0.487823507810632\\
0.198425131496025	0.09709485967954\\
0.314980262473718	0.25713968804727\\
0.5	-0.160588485312959\\
};
\addlegendentry{Proposed (no dict. reduction)}

\end{axis}

\end{tikzpicture}%
    \vspace{-3em}
    \caption{Evaluation of the channel estimation NMSE versus~\ac{ris} inter-element spacing. A shorter spacing typically causes stronger MC effect\cite[Fig.~2]{Zheng2024On}.}
    \label{fig:nmse_vs_spacing}
\end{figure}

\begin{table} [t]
    \footnotesize
    \centering
    \caption{ {Performance evaluation of channel estimation computational complexity assessed by the running time (seconds) for different \ac{ris} configurations with $N_\Rt=4\times2$ and $ N_\Tt=2\times1$.}  
    \label{table:comp_vs_size}}
    \vspace{-2em}
    \usetikzlibrary{matrix}

\begin{tikzpicture}
    \matrix(dict)[matrix of nodes,
        nodes={align=center,text width=1cm, font=\scriptsize},
        row 1/.style={anchor=south,font=\scriptsize},
        column 1/.style={nodes={text width=0.9cm,align=right}}]
        {
        Config. & Method & Conv. & \scriptsize{$\rho_{\scalebox{0.7}{$\mathrm{DR}$}}\!\!=\!\!0.1$} &\scriptsize{$\rho_{\scalebox{0.7}{$\mathrm{DR}$}}\!\!=\!\!0.5$} &\scriptsize{$\rho_{\scalebox{0.7}{$\mathrm{DR}$}}\!\!=\!\!1$} & MC-aware \\
        \tiny{$N_\It\!\!=\!\!8\!\!\times\!\!8\!$} & Off. & 0.0173 & 0.2378 & 0.5495 & 0.8918 & 1.3495\\ 
        \tiny{$N_\It\!\!=\!\!8\!\!\times\!\!8$} & On. & 0.0257 & 0.0282 & 0.0336 & 0.0502 & 0.0540\\ 
        \tiny{$N_\It\!\!=\!\!16\!\!\times\!\!8$} & Off. & 0.1996 & 4.0011 & 10.6755 & 18.5376 & 32.3983\\ 
        \tiny{$N_\It\!\!=\!\!16\!\!\times\!\!8$} & On. & 0.1770 & 0.1901 & 0.2676 & 0.3552 & 0.3485\\   
    };
    %\draw(dict-1-1.south west)--(dict-1-4.south east);
    %\draw(dict-1-1.north east)--(dict-2-1.south east);
    \draw (dict-1-1.south west)--(dict-1-7.south east);
    \draw (dict-2-1.south west)--(dict-2-7.south east);
    \draw (dict-3-1.south west)--(dict-3-7.south east);
    \draw (dict-4-1.south west)--(dict-4-7.south east);
    \draw (dict-1-1.north east)--(dict-5-1.south east);
    \draw (dict-1-2.north east)--(dict-5-2.south east);
\end{tikzpicture}
    \vspace{-5em}
\end{table}

\begin{figure}[t]
    \centering
\definecolor{mycolor1}{rgb}{1.00000,0.00000,1.00000}%
\definecolor{mycolor2}{rgb}{0.00000,1.00000,1.00000}%
\definecolor{ForestGreen}{rgb}{0.1333    0.8451    0.1333}
\begin{tikzpicture}

\begin{axis}[%
width=2.8in,
height=2in,
at={(0in,0in)},
scale only axis,
xmin=-8.4,
xmax=12.4,
xlabel style={font=\color{white!15!black},font=\footnotesize},
xticklabel style = {font=\color{white!15!black},font=\footnotesize},
xlabel={UE transmit power %$P_\Tt^\mathrm{UL}$ 
$P_\Tt$ (dBm)},
ymin=-12,
ymax=28,
ylabel style={font=\color{white!15!black},font=\footnotesize},
yticklabel style = {font=\color{white!15!black},font=\footnotesize},
ylabel={Estimation NMSE (dB)},
ytick = {-12,-7, -2, 3, 8, 13, 18, 23, 28},
axis background/.style={fill=white},
xmajorgrids,
ymajorgrids,
grid style={dashed},
legend columns=2,
legend style={
    at={(0.5,1.03)},
    anchor=south,
    font=\tiny,
    legend cell align=left,
    align=left,
    draw=white!15!black,
    fill opacity=0.85
  }
]

\addplot [color=blue, dashed, line width=0.7pt, mark size=2.3pt, mark=o, mark options={solid, rotate=270, blue}]
  table[row sep=crcr]{%
-8	25.6096281671524\\
-6  23.2765334312341\\  
-4	18.3616048002243\\
-2  15.5701529123803\\
0	10.9111025667191\\
2   7.09258874513990\\
4	2.29598801612854\\
6   0.27441712394581\\
8	-2.38390565536916\\
10  -4.54309281234546\\
12	-6.52338094830513\\
};
\addlegendentry{Proposed $(\rho_{\scalebox{0.7}{$\mathrm{DR}$}}\!=\!0.1\!,\! \sigma_{\scalebox{0.7}{$\err$}}^2\!=\!10 \sigma_{\scalebox{0.7}{$\It$}}^2)$}

\addplot [color=red, dashed, mark=triangle, mark options={solid, rotate=180, red}, line width=0.7pt, mark size=2.5pt]
  table[row sep=crcr]{%
-8	26.948813829422\\
-6  21.985611246712\\
-4	19.6037086129189\\
-2  14.6362934599821\\
0	11.5369949531555\\
2   6.60355092882358\\
4	2.65955745182931\\
6   -0.0441209812769\\
8	-2.11073016613722\\
10  -4.45983090876541\\
12	-6.97782143086195\\
};
\addlegendentry{Proposed $(\rho_{\scalebox{0.7}{$\mathrm{DR}$}}\!=\!1\!,\! \sigma_{\scalebox{0.7}{$\err$}}^2\!=\!10 \sigma_{\scalebox{0.7}{$\It$}}^2)$}

\addplot [color=blue, dotted, line width=0.7pt, mark size=2.3pt, mark=o, mark options={solid, rotate=270, blue}]
  table[row sep=crcr]{%
-8	17.0503397464752\\
-6  14.91135675721398\\
-4	10.3534946644725\\
-2  7.84158287273191\\
0	3.34734362900257\\
2   0.26885852340123\\
4	-3.48869785696268\\
6   -5.21175567883234\\
8	-6.44899149954319\\
10  -7.58295411233792\\
12	-8.71694522380829\\
};
\addlegendentry{Proposed $(\rho_{\scalebox{0.7}{$\mathrm{DR}$}}\!=\!0.1\!,\! \sigma_{\scalebox{0.7}{$\err$}}^2\!=\!3 \sigma_{\scalebox{0.7}{$\It$}}^2)$}

\addplot [color=red, dotted, mark=triangle, mark options={solid, rotate=180, red}, line width=0.7pt, mark size=2.5pt]
  table[row sep=crcr]{%
-8	18.2113927841187\\
-6  13.7019334312341\\  
-4	11.6113122451305\\
-2  6.85041129431235\\
0	4.07181375831366\\
2   -0.0707123123803\\
4	-3.53409506082535\\
6   -4.96887634512301\\
8	-6.88906991928816\\
10  -8.13912534873011\\
12	-9.38917553901672\\
};
\addlegendentry{Proposed $(\rho_{\scalebox{0.7}{$\mathrm{DR}$}}\!=\!1\!,\!\sigma_{\scalebox{0.7}{$\err$}}^2\!=\!3 \sigma_{\scalebox{0.7}{$\It$}}^2)$}

\addplot [color=blue, dashdotted, mark=o, mark options={solid, rotate=270, blue}, line width=0.7pt, mark size=2.3pt]
  table[row sep=crcr]{%
-8	11.4312585926056\\
-6	8.52149283665401\\
-4	5.61139210045338\\
-2	2.23501298781327\\
0	-1.14123940467834\\
2   -3.82871223865567\\
4	-6.51629140973091\\
6   -7.28509234756192\\
8	-8.05376019120216\\
10  -8.80050123532387\\ 
12	-9.54705720424652\\
};
\addlegendentry{Proposed $(\rho_{\scalebox{0.7}{$\mathrm{DR}$}}\!=\!0.1\!,\!\sigma_{\scalebox{0.7}{$\err$}}^2\!=\!\sigma_{\scalebox{0.7}{$\It$}}^2)$}

\addplot [color=red, dashdotted, mark=triangle, mark options={solid, rotate=180, red}, line width=0.7pt, mark size=2.5pt]
  table[row sep=crcr]{%
-8	12.4088907718658\\
-6	9.35688943827123\\
-4	6.30414094924927\\
-2  2.75341235098345\\
0	-0.797411224655807\\
2   -3.83609287690132\\
4	-6.87452405691147\\
6   -7.98280459123398\\
8	-9.08959671020508\\
10  -9.76491667743282\\
12	-10.4402094936371\\
};
\addlegendentry{Proposed $(\rho_{\scalebox{0.7}{$\mathrm{DR}$}}\!=\!1\!,\!\sigma_{\scalebox{0.7}{$\err$}}^2\!=\!\sigma_{\scalebox{0.7}{$\It$}}^2)$}

\addplot [color=ForestGreen, mark=triangle, mark options={solid, rotate=90, ForestGreen}, line width=0.7pt, mark size=2.5pt]
  table[row sep=crcr]{%
-8	8.91493791639805\\
-6	7.31145299275716\\
-4	4.86803406526645\\
-2	3.45948322912057\\
0	2.02116206480035\\
2	1.38258584787448\\
4	0.656718323007226\\
6	0.252501189506923\\
8	0.0979069166195889\\
10	-0.231107672614356\\
12	-0.291082631129151\\
};
\addlegendentry{MC-unaware OMP}

\addplot [color=black, mark=square, mark options={solid, black}, line width=0.7pt, mark size=2pt]
  table[row sep=crcr]{%
-8	4.86541893593967\\
-6	2.80965921580791\\
-4	-0.454972085418801\\
-2	-3.02313821191589\\
0	-6.23800949553649\\
2	-6.9009876092275\\
4	-8.7336779375871\\
6	-10.0505542516708\\
8	-10.2826570232709\\
10	-10.9251841465632\\
12	-11.4283880710602\\
};
\addlegendentry{Proposed (no dict. reduction)}

\addplot [color=blue, mark=o, mark options={solid, rotate=270, blue}, line width=0.7pt, mark size=2.3pt]
  table[row sep=crcr]{%
-8	3.62947482566039\\
-6	2.28565771778425\\
-4	-0.94691002368927\\
-2	-3.17317625035842\\
0	-6.21538640807072\\
2	-6.53053322806954\\
4	-7.79871680736542\\
6	-9.02228455940882\\
8	-9.13859035571416\\
10	-9.78759246667226\\
12	-10.0753294030825\\
};
\addlegendentry{Proposed $(\rho_{\scalebox{0.7}{$\mathrm{DR}$}}\!=\!0.1\!,\!\sigma_{\scalebox{0.7}{$\err$}}^2\!=\!0)$}

\addplot [color=red, mark=triangle, mark options={solid, rotate=180, red}, line width=0.7pt, mark size=2.5pt]
  table[row sep=crcr]{%
-8	4.65935019105673\\
-6	2.74567199349403\\
-4	-0.691081937526663\\
-2	-3.30108859737714\\
0	-6.31810194465021\\
2	-6.9645220408837\\
4	-8.71545544366042\\
6	-10.1455333928267\\
8	-10.3970056692759\\
10	-10.8433075586955\\
12	-11.509513648351\\
};
\addlegendentry{Proposed $(\rho_{\scalebox{0.7}{$\mathrm{DR}$}}\!=\!1\!,\!\sigma_{\scalebox{0.7}{$\err$}}^2\!=\!0)$}

\end{axis}

\end{tikzpicture}%
    \vspace{-3em}
    \caption{ {Evaluation of the channel estimation NMSE versus transmit power in the presence of propagation error with different error variance $\sigma_\err^2$ levels.}}
    \label{fig:errana_vs_sigma}
\end{figure}

Figure~\ref{fig:chest_perf} presents the channel estimation \ac{nmse} versus transmit power in the uplink. The proposed two-stage algorithm consistently outperforms the conventional MC-unaware OMP in the presence of \ac{mc}. Even for small values of \ac{dr} factor~$\rho_\DicRed$, our method achieves several dB improvements in accuracy compared to the conventional approach. Furthermore, the estimation accuracy increases with greater $\rho_\mathrm{DR}$. For $\rho_\DicRed \geq 0.1$, the proposed algorithm performs on par with the direct MC-aware OMP (without dictionary reduction).\footnote{ {The \ac{mc} effect in \eqref{eq:MC_ChM} is influenced by both the structure of the scattering matrix $\Sm$ and the magnitude of its elements, as well as the power of the active RIS. In practice, these parameters are highly dependent on the specific RIS implementation. Given this dependency, a fixed value for $\rho_\DicRed$ may not be optimal across varying \ac{mc} intensities. Therefore, a promising direction for future research is to investigate an adaptive strategy where $\rho_\DicRed$ is dynamically adjusted to offer the best trade-off between channel estimation performance and computational complexity.}} 

Figure~\ref{fig:nmse_vs_rispower} examines the effect of the RIS amplification factor on channel estimation accuracy. We observe that under a low amplification factor ($\bar{a} \leq 3$), the conventional and proposed methods perform similarly, and both with accuracy improving as the RIS amplification increases, indicating a negligible impact of MC. However, as the RIS amplification factor continues to increase, the performance of the conventional MC-unaware OMP deteriorates and diverges from the proposed method due to the stronger impact of MC, as analyzed in Remark~\ref{rmk:RISamp}. Notably, when $\bar{a} > 5$, conventional channel estimation performance remains constant with further increases in RIS amplification. This suggests that simply increasing RIS power is not always beneficial unless MC is well-managed. In contrast, our proposed method effectively accounts for the MC, allowing it to consistently benefit from increased RIS amplification, even with a small dictionary reduction factor. 

In Figure \ref{fig:nmse_vs_spacing}, the effect of different~\ac{mc} levels is investigated by varying the~\ac{ris} inter-element spacing. When the RIS spacing is large ($d_\It > \lambda/10$), the MC effect is weak, and thus the performance of the MC-unaware OMP and the proposed method with appropriate dictionary reduction factors ($\rho_\mathrm{DR} \geq 1$ or no reduction) is similar. Nevertheless, with shorter RIS spacing ($d_\It < \lambda/10$), the MC-unaware OMP suffers significant performance degradation, which is exacerbated as the MC effect intensifies. The proposed method, however, maintains a performance gain of over $\unit[20]{dB}$ under strong-MC conditions.

 {Table~\ref{table:comp_vs_size} provides an evaluation of the wall-clock running time for two different \ac{ris} sizes, where these results were obtained by averaging the running time for 100 trials. As expected, the derived \ac{mc}-aware solution has the highest complexity and incurs the longest running time, while the \ac{mc}-unaware approach enjoys the lowest running time at the price of poor channel estimation performance, as depicted in previous results. Moreover, the proposed algorithm has a much lower overall running time, especially with small~$\rho_\mathrm{DR}$ values, highlighting a trade-off between performance and complexity. The running time of the online part for large $\rho_\mathrm{DR}$ values, e.g., $\rho_\mathrm{DR}=1$ is almost similar to the exact solution. However, the proposed Algorithm \ref{algo:prop_ch_est} still results in much lower running time when accounting for the overall complexities. As shown in our results, the offline complexity of the \ac{mc}-aware solution is the main bottleneck since computing the sensing matrix requires large measurement and dictionary matrices, even with a moderate \ac{ris} size of $N_\It=128$ unit cells. Note that even our proposal is not scalable for large \ac{ris} sizes, the authors believe that this paper establishes a first step for physically consistent \ac{ris} channel estimation and beamforming in the presence of \ac{mc}; however, further research in this direction is needed. The main drawback of the aforementioned sparse estimation algorithm, whether using \ac{mc}-aware or the proposed Algorithm \ref{algo:prop_ch_est}, lies in the structure of its dictionary matrix. This matrix is constructed as a Kronecker product of several individual dictionaries, so its overall size scales with the product of the individual dictionary sizes. Moreover, the performance of the on-grid CS approach heavily relies on the discretization grid used to quantize the dictionary, which is again equal to the dictionary size multiplied by an additional oversampling factor to enhance estimation accuracy. Thus, such a combination, along with the dictionary matrix's structure, significantly increases the sparse recovery complexity.

To assess the impact of error propagation on the performance of the proposed channel estimation, we test the cases with different noise levels. We simulate a scenario in which the input signal in step 5 of Algorithm \ref{algo:prop_ch_est} operates with additional additive error. The estimation error $\boldsymbol{\hbar} \!\sim\!\mathcal{CN}(\mathbf{0},\sigma^2_\err)$ is modeled as a zero-mean complex Gaussian perturbation with variance of $\sigma^2_\err$. Figure~\ref{fig:errana_vs_sigma} evaluates the impact of this additive error on the channel estimation performance. As the error variance increases, the performance of the proposed algorithm degrades, resulting in a higher NMSE, especially in the low SNR region, which means that the proposed method is sensitive to additive error propagation. However, at high SNR regions, even with large error variance ($\sigma_\err^2=10 \sigma_\It^2$), the performance still outperforms the \ac{mc}-unaware OMP.}

\subsection{Beamforming Performance Evaluation}

\begin{figure}[t]
    \centering
\definecolor{ForestGreen}{rgb}{0.1333    0.8451    0.1333}
\begin{tikzpicture}

\begin{axis}[%
width=2.8in,
height=2in,
at={(0in,0in)},
scale only axis,
xmin=-10.4,
xmax=20.4,
xlabel style={font=\color{white!15!black},font=\footnotesize},
xticklabel style = {font=\color{white!15!black},font=\footnotesize},
xlabel={BS transmit power %$P_\Rt^{\mathrm{DL}}$ 
$P_\Rt$ (dBm)},
ymin=-0.1,
ymax=15,
ylabel style={font=\color{white!15!black},font=\footnotesize},
yticklabel style = {font=\color{white!15!black},font=\footnotesize},
ytick = {0,5,10,15},
ylabel={Spectral efficiency (bits/s/Hz)},
axis background/.style={fill=white},
xmajorgrids,
ymajorgrids,
grid style={dashed},
legend style={at={(0,1)}, anchor=north west, font=\tiny, legend cell align=left, align=left, draw=white!15!black, fill opacity=0.85}
]
\addplot [color=red, dashed, mark=+, mark options={solid, red}, line width=0.8pt, mark size=2.5pt]
  table[row sep=crcr]{%
-10	4.35186742156375\\
-5	5.96369294201565\\
0	7.60876369563791\\
5	9.26466462023765\\
10	10.924024427572\\
15	12.5844802331112\\
20	14.2452840883004\\
};
\addlegendentry{$\mathbf{G}_\mathrm{mc}$ (Ground truth) + SCA BF}

\addplot [color=red, mark=diamond, mark options={solid, red}, line width=0.7pt, mark size=3pt]
  table[row sep=crcr]{%
-10	2.8802316403389\\
-5	4.39073646306992\\
0	5.99819068908691\\
5	7.64129183769226\\
10	9.29648900985718\\
15	10.9556161880493\\
20	12.6160000801086\\
};
\addlegendentry{$\hat{\mathbf{G}}_\mathrm{mc}$ (Proposed, $\rho_{\scalebox{0.7}{$\mathrm{DR}$}}=1$) + SCA BF}

\addplot [color=blue, dashed, mark=x, mark options={solid, blue}, line width=0.8pt, mark size=2.5pt]
  table[row sep=crcr]{%
-10	1.71844968044838\\
-5	3.04338429967806\\
0	4.57943714035548\\
5	6.19853882517831\\
10	7.84600797838145\\
15	9.50267816085092\\
20	11.1622817050093\\
};
\addlegendentry{$\mathbf{G}_\mathrm{cv}$ (Ground truth) + SVD BF}

\addplot [color=blue, mark=square, mark options={solid, blue}, line width=0.7pt, mark size=2pt]
  table[row sep=crcr]{%
-10	0.815984952747822\\
-5	1.76233124315739\\
0	3.09075803995132\\
5	4.61086959838867\\
10	6.22799573898315\\
15	7.87466161727905\\
20	9.53105628967285\\
};
\addlegendentry{$\hat{\mathbf{G}}_\mathrm{cv}$ (MC-unaware OMP) + SVD BF}

\addplot [color=ForestGreen, dashed, mark=star, mark options={solid, ForestGreen}, line width=0.8pt, mark size=2.5pt]
  table[row sep=crcr]{%
-10	0.995984775205635\\
-5	2.05126768028444\\
0	3.45212287759895\\
5	5.02001318907504\\
10	6.65024966748906\\
15	8.301359005832\\
20	9.95919265914318\\
};
\addlegendentry{$\mathbf{G}_\mathrm{mc}$ (Ground truth) + GD BF}

\addplot [color=ForestGreen, mark=o, mark options={solid, ForestGreen}, line width=0.7pt, mark size=2.3pt]
  table[row sep=crcr]{%
-10	0.52780412517488\\
-5	1.25705689311028\\
0	2.42259053051472\\
5	3.8757165145874\\
10	5.4603120136261\\
15	7.09524301528931\\
20	8.74771506309509\\
};
\addlegendentry{$\hat{\mathbf{G}}_\mathrm{mc}$ (Proposed, $\rho_{\scalebox{0.7}{$\mathrm{DR}$}}=1$) + GD BF}
\end{axis}
\end{tikzpicture}%
    \vspace{-3em}
    \caption{Performance evaluation of spectral efficiency versus base station transmit power for different estimation and beamforming methods.}
    \label{fig:se_vs_snr}
\end{figure}
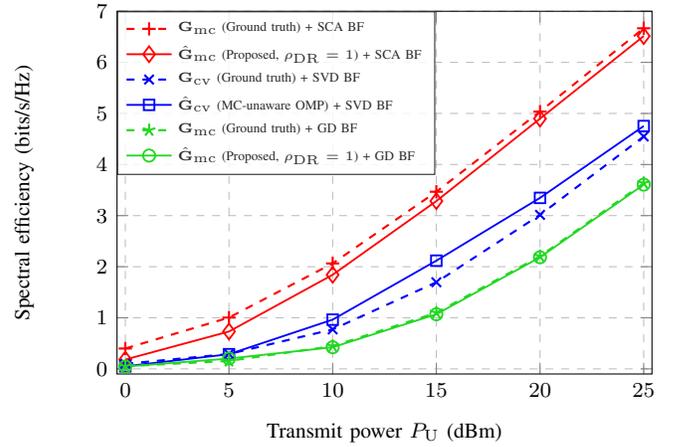
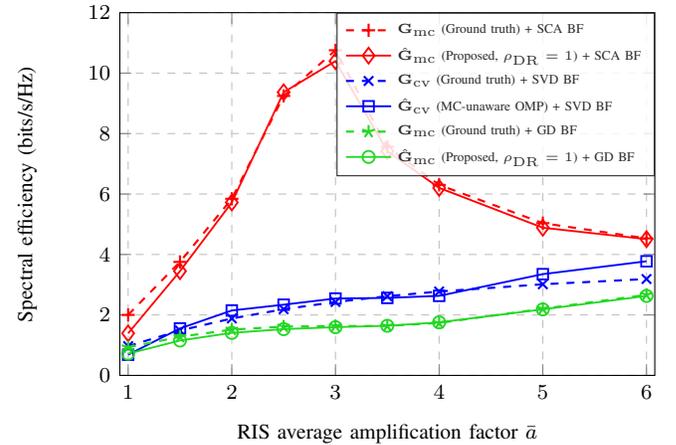
\begin{figure}[t]
    \centering
\definecolor{ForestGreen}{rgb}{0.1333    0.8451    0.1333}
\begin{tikzpicture}

\begin{axis}[%
width=2.8in,
height=2in,
at={(0in,0in)},
scale only axis,
xmin=0.92,
xmax=10.08,
xminorticks=true,
xlabel style={font=\color{white!15!black},font=\footnotesize},
xticklabel style = {font=\color{white!15!black},font=\footnotesize},
xtick = {1,2,3,4,5,6,7,8,9,10},
xlabel={RIS average amplification factor~$\bar{a}$},
ymin=0,
ymax=16,
ylabel style={font=\color{white!15!black},font=\footnotesize},
yticklabel style = {font=\color{white!15!black},font=\footnotesize},
ylabel={Spectral efficiency (bits/s/Hz)},
ytick = {0,2,4,6,8,10,12,14,16},
axis background/.style={fill=white},
xmajorgrids,
xminorgrids,
ymajorgrids,
grid style={dashed},
legend style={at={(0,1)}, anchor=north west, font=\tiny, legend cell align=left, align=left, draw=white!15!black, fill opacity=0.85}
]
\addplot [color=red, dashed, mark=+, mark options={solid, red}, line width=0.8pt, mark size=2.5pt]
  table[row sep=crcr]{%
1	3.41258317723052\\
2	5.57179751442176\\
3	6.95722926949852\\
4	8.03555209812062\\
5	8.95614761536438\\
6	9.78500001660686\\
7	10.924024427572\\
8	12.5347603231758\\
9	13.3160753343701\\
10	12.7904015174146\\
};
%\addlegendentry{Ground truth $\mathbf{G}_\mathrm{mc}$ + SCA BF}
\addlegendentry{$\mathbf{G}_\mathrm{mc}$ (Ground truth) + SCA BF}

\addplot [color=red, mark=diamond, mark options={solid, red}, line width=0.7pt, mark size=3pt]
  table[row sep=crcr]{%
1	0.0336000848213462\\
2	0.213261841590079\\
3	0.586328469019807\\
4	1.73714194225587\\
5	3.07513882942514\\
6	7.35101833343506\\
7	8.91454286575317\\
8	10.0126180648804\\
9	11.7559144973755\\
10	12.2459935188293\\
};
%\addlegendentry{Exact RD-OMP estimated $\hat{\mathbf{G}}_\mathrm{mc}$ + SCA BF}
\addlegendentry{$\hat{\mathbf{G}}_\mathrm{mc}$ (Proposed, $\rho_{\scalebox{0.7}{$\mathrm{DR}$}}=1$) + SCA BF}

\addplot [color=blue, dashed, mark=x, mark options={solid, blue}, line width=0.8pt, mark size=2.5pt]
  table[row sep=crcr]{%
1	3.23687593804208\\
2	4.98586457694995\\
3	5.9944765831612\\
4	6.67549889083028\\
5	7.17215973021879\\
6	7.55037030631886\\
7	7.84600797838145\\
8	8.08091085568755\\
9	8.26947878906762\\
10	8.42178615759294\\
};
%\addlegendentry{Ground truth $\mathbf{G}_\mathrm{cv}$ + SVD BF}
\addlegendentry{$\mathbf{G}_\mathrm{cv}$ (Ground truth) + SVD BF}

\addplot [color=blue, mark=square, mark options={solid, blue}, line width=0.7pt, mark size=2pt]
  table[row sep=crcr]{%
1	0.0308841200365259\\
2	0.14484144133709\\
3	0.352706546269156\\
4	1.19616813690945\\
5	2.05907486117649\\
6	4.36999378204346\\
7	6.69851484298706\\
8	6.68614845275879\\
9	6.87568960189819\\
10	6.97628803253174\\
};
%\addlegendentry{Conv. OMP estimated $\hat{\mathbf{G}}_\mathrm{cv}$ + SVD BF}
\addlegendentry{$\hat{\mathbf{G}}_\mathrm{cv}$ (MC-unaware OMP) + SVD BF}

\addplot [color=ForestGreen, dashed, mark=star, mark options={solid, ForestGreen}, line width=0.8pt, mark size=2.5pt]
  table[row sep=crcr]{%
1	3.21739202956824\\
2	4.85153143743169\\
3	5.62853934777478\\
4	6.05479646702552\\
5	6.32215635100632\\
6	6.50893151303876\\
7	6.65024966748907\\
8	6.76343819365181\\
9	6.85780072444544\\
10	6.93872031927556\\
};
%\addlegendentry{Ground truth $\mathbf{G}_\mathrm{mc}$ + GD BF}
\addlegendentry{$\mathbf{G}_\mathrm{mc}$ (Ground truth) + GD BF}

\addplot [color=ForestGreen, mark=o, mark options={solid, ForestGreen}, line width=0.7pt, mark size=2.3pt]
  table[row sep=crcr]{%
1	0.0304193536921269\\
2	0.141984240884208\\
3	0.360798844496967\\
4	1.05423212874848\\
5	1.84255906515638\\
6	4.38620352745056\\
7	5.16689405441284\\
8	5.11421298980713\\
9	5.2056884765625\\
10	5.25716438293457\\
};
%\addlegendentry{Exact RD-OMP estimated $\hat{\mathbf{G}}_\mathrm{mc}$ + GD BF}
\addlegendentry{$\hat{\mathbf{G}}_\mathrm{mc}$ (Proposed, $\rho_{\scalebox{0.7}{$\mathrm{DR}$}}=1$) + GD BF}

\end{axis}

\end{tikzpicture}%
    \vspace{-3em}
    \caption{Evaluation of spectral efficiency versus~\ac{ris} average amplification factor~$\bar{a}=\sqrt{A/N_\It}$. As highlighted in Remark~\ref{rmk:RISamp}, higher amplification intensifies the impact of MC.}
    \label{fig:se_vs_rispower}
\end{figure}
\begin{figure}[t]
    \centering
\definecolor{ForestGreen}{rgb}{0.1333    0.8451    0.1333}
\begin{tikzpicture}

\begin{axis}[%
width=2.8in,
height=2in,
at={(0in,0in)},
scale only axis,
xmode=log,
xmin=0.019,
xmax=0.525,
xminorticks=true,
xlabel style={font=\color{white!15!black},font=\footnotesize},
xticklabel style = {font=\color{white!15!black},font=\footnotesize},
xtick = {1/50,1/32,1/16,1/8,1/4,1/2},
xticklabels = {$\frac{\lambda}{50}$,$\frac{\lambda}{32}$,$\frac{\lambda}{16}$,$\frac{\lambda}{8}$,$\frac{\lambda}{4}$,$\frac{\lambda}{2}$},
xlabel={RIS inter-element spacing $d_\It$},
ymin=0,
ymax=14.1,
ylabel style={font=\color{white!15!black},font=\footnotesize},
yticklabel style = {font=\color{white!15!black},font=\footnotesize},
ytick = {0,2,4,6,8,10,12,14},
ylabel={Spectral efficiency (bits/s/Hz)},
axis background/.style={fill=white},
xmajorgrids,
xminorgrids,
ymajorgrids,
grid style={dashed},
legend columns=2,
legend style={
    at={(0.5,1.03)},
    anchor=south,
    font=\tiny,
    legend cell align=left,
    align=left,
    draw=white!15!black,
    fill opacity=0.85
  }
]
\addplot [color=red, dashed, mark=+, mark options={solid, red}, line width=0.8pt, mark size=2.5pt]
  table[row sep=crcr]{%
0.02	4.15473232850204\\
0.03125	3.99111242675829\\
0.0364540324867536	4.3393623714567\\
0.0425246875054493	4.94776334187817\\
0.0496062828740062	5.63049213702672\\
0.0578671695179556	7.30478163369533\\
0.0675037336807691	13.2428691131643\\
0.0787450656184295	9.80308329796026\\
0.0918584057672249	12.8716033382705\\
0.107155497856634	10.1818082093672\\
0.125	9.62251952046107\\
0.125	9.62251952046107\\
0.198425131496025	8.92682328699948\\
0.314980262473718	8.81782205190781\\
0.5	8.82792405614085\\
};
%\addlegendentry{Ground truth $\mathbf{G}_\mathrm{mc}$ + SCA BF}
\addlegendentry{$\mathbf{G}_\mathrm{mc}$ (Ground truth) + SCA BF}

\addplot [color=red, mark=diamond, mark options={solid, red}, line width=0.7pt, mark size=3pt]
  table[row sep=crcr]{%
0.02	2.14503438636661\\
0.03125	2.26788985744119\\
0.0364540324867536	2.89145854026079\\
0.0425246875054493	3.53193435132503\\
0.0496062828740062	4.66000070810318\\
0.0578671695179556	5.5025871515274\\
0.0675037336807691	5.83176773011684\\
0.0787450656184295	6.07461835349528\\
0.0918584057672249	7.52583995580673\\
0.107155497856634	6.10535707481205\\
0.125	6.49869286591013\\
0.125	6.49869286591013\\
0.198425131496025	5.42977514248603\\
0.314980262473718	3.83146353520781\\
0.5	3.50490528184239\\
};
%\addlegendentry{Exact RD-OMP estimated $\hat{\mathbf{G}}_\mathrm{mc}$ + SCA BF}
\addlegendentry{$\hat{\mathbf{G}}_\mathrm{mc}$ (Proposed, $\rho_{\scalebox{0.7}{$\mathrm{DR}$}}=1$) + SCA BF}

\addplot [color=blue, dashed, mark=x, mark options={solid, blue}, line width=0.8pt, mark size=2.5pt]
  table[row sep=crcr]{%
0.02	3.54681221348555\\
0.03125	3.61481683037372\\
0.0364540324867536	3.91816197972567\\
0.0425246875054493	4.42720685890809\\
0.0496062828740062	5.05256074826463\\
0.0578671695179556	5.71884826948441\\
0.0675037336807691	6.35968340735097\\
0.0787450656184295	6.96995204537731\\
0.0918584057672249	7.54337197110514\\
0.107155497856634	8.0782249084702\\
0.125	8.47760578537612\\
0.125	8.47760578537612\\
0.198425131496025	8.85991289594693\\
0.314980262473718	8.80776932249289\\
0.5	8.82749892517745\\
};
%\addlegendentry{Ground truth $\mathbf{G}_\mathrm{cv}$ + SVD BF}
\addlegendentry{$\mathbf{G}_\mathrm{cv}$ (Ground truth) + SVD BF}

\addplot [color=blue, mark=square, mark options={solid, blue}, line width=0.7pt, mark size=2pt]
  table[row sep=crcr]{%
0.02	0.565966561936866\\
0.03125	0.618528550217865\\
0.0364540324867536	0.947993456671111\\
0.0425246875054493	1.30698051988948\\
0.0496062828740062	1.22903006077191\\
0.0578671695179556	2.06846957692876\\
0.0675037336807691	3.77232299089432\\
0.0787450656184295	3.82907490801066\\
0.0918584057672249	4.16049845484877\\
0.107155497856634	4.91003696653992\\
0.125	5.53308026480154\\
0.125	5.53308026480154\\
0.198425131496025	4.92573241205483\\
0.314980262473718	3.70201859303212\\
0.5	3.49970987546805\\
};
%\addlegendentry{Conv. OMP estimated $\hat{\mathbf{G}}_\mathrm{cv}$ + SVD BF}
\addlegendentry{$\hat{\mathbf{G}}_\mathrm{cv}$ (MC-unaware OMP) + SVD BF}

\addplot [color=ForestGreen, dashed, mark=star, mark options={solid, ForestGreen}, line width=0.8pt, mark size=2.5pt]
  table[row sep=crcr]{%
0.02	3.37567921912513\\
0.03125	3.57300840916087\\
0.0364540324867536	3.7140868998064\\
0.0425246875054493	3.88174818421475\\
0.0496062828740062	5.73746287747382\\
0.0578671695179556	4.78337751651459\\
0.0675037336807691	4.6190853261479\\
0.0787450656184295	5.78985401432681\\
0.0918584057672249	6.39585197613096\\
0.107155497856634	6.81247827105891\\
0.125	7.05465044711497\\
0.125	7.05465044711497\\
0.198425131496025	7.34761682782063\\
0.314980262473718	7.13975055149605\\
0.5	6.93968192453189\\
};
%\addlegendentry{Ground truth $\mathbf{G}_\mathrm{mc}$ + GD BF}
\addlegendentry{$\mathbf{G}_\mathrm{mc}$ (Ground truth) + GD BF}

\addplot [color=ForestGreen, mark=o, mark options={solid, ForestGreen}, line width=0.7pt, mark size=2.3pt]
  table[row sep=crcr]{%
0.02	1.7421061526984\\
0.03125	1.89071507081389\\
0.0364540324867536	2.26065897122025\\
0.0425246875054493	2.4424239321053\\
0.0496062828740062	3.71862864017487\\
0.0578671695179556	3.39064029037952\\
0.0675037336807691	3.6741531952153\\
0.0787450656184295	3.0457937028259\\
0.0918584057672249	3.08110265053809\\
0.107155497856634	3.46675797447562\\
0.125	4.24716798452651\\
0.125	4.24716798452651\\
0.198425131496025	4.23652964010835\\
0.314980262473718	2.94816424066899\\
0.5	2.37830096228281\\
};
%\addlegendentry{Exact RD-OMP estimated $\hat{\mathbf{G}}_\mathrm{mc}$ + GD BF}
\addlegendentry{$\hat{\mathbf{G}}_\mathrm{mc}$ (Proposed, $\rho_{\scalebox{0.7}{$\mathrm{DR}$}}=1$) + GD BF}

\end{axis}

\end{tikzpicture}%
    \vspace{-3em}
    \caption{Evaluation of the spectral efficiency versus~\ac{ris} inter-element spacing. A shorter spacing typically causes stronger MC effect\cite[Fig.~2]{Zheng2024On}.}
    \label{fig:se_vs_spacing}
\end{figure}
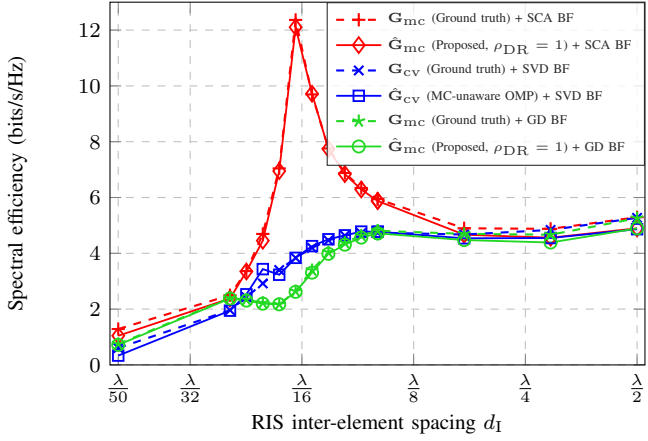

 {Now, we focus on the beamforming gains of the proposed strategy while incorporating the estimated channel. We present both scenarios; the first assumes perfect \ac{csi}, and the second utilizes the output of the proposed channel estimation algorithm (i.e., partial \ac{csi}).} This is different from the majority of the existing literature, where authors try to simulate partial \ac{csi} by simply adding some random variables (typically following normal or uniform distributions) to the ground truth channel to change either channel gains (most common) or the \ac{aoa}/\ac{aod}. We select $L_\Tt\!=\!L_\Rt\!=\!1$ to minimize the power leakage effect since we use regular dictionaries (without oversampling) for the \ac{omp} and choose arbitrary \ac{aoa}/\ac{aod}. Other parameters are set as $d_\It=\lambda/10, \bar{a}=7, P_\Tt=\unit[5]{mW},P_\Rt=\unit[10]{mW}$. 

Figure~\ref{fig:se_vs_snr} evaluates the spectral efficiency of different channel estimation and beamforming methods by changing the \ac{bs} transmit power $P_\Rt$ in the downlink. In addition to the MC-aware beamforming approaches based on~$\Gm_\mathrm{mc}$ (using either \ac{sca} or \ac{gd}), we also assess the conventional MC-unaware beamforming based on~$\Gm_\mathrm{cv}$, which is addressed using \ac{svd}\cite{ElAyach2014Spatially}. It is evident that the proposed \ac{sca}-based MC-aware solution achieves the highest spectral efficiency in both perfect and partial \ac{csi} cases. Moreover, the spectral efficiency achieved with the estimated~$\hat{\Gm}_\mathrm{mc}$ from the proposed channel estimation method is only slightly lower than that with the ground truth channel, which again confirms the effectiveness of the proposed channel estimator. A similar trend is observed with the \ac{gd}-based MC-aware beamforming method. 

Analogous to Fig.~\ref{fig:nmse_vs_rispower} and Fig.~\ref{fig:nmse_vs_spacing}, Fig.~\ref{fig:se_vs_rispower} and Fig.~\ref{fig:se_vs_spacing} assess the spectral efficiency with varying RIS amplification factors and inter-element spacings, respectively. Both figures demonstrate that the proposed \ac{sca}-based MC-aware method consistently achieves the highest spectral efficiency. An interesting observation is that, as the impact of MC increases, the performance gain of the proposed \ac{sca}-based solution initially improves but eventually declines. This occurs because the Neumann series approximation~\eqref{eq:Neumann} becomes less accurate when MC is excessively strong. In addition, in Fig.~\ref{fig:se_vs_spacing}, we note a significant high spectral efficiency between $d_\It = \lambda/16$ and $d_\It = \lambda/8$. Since this setup does not increase RIS power but only shifts the inter-element spacing, it suggests that with sophisticated signal processing solutions, additional gains can be obtained from~\ac{mc}. This observation indicates that the prevailing view of MC as solely a negative effect may be biased. 

\section{Conclusion}
\label{sec:Conc}
 {
This paper investigates channel estimation and beamforming for active \ac{ris}-assisted \ac{mimo} systems in the presence of \ac{mc}. We show that, despite its complexity, the MC-aware channel estimation problem can still be efficiently formulated as a \ac{cs} task by exploiting inherent signal structures. Employing the \ac{cs} and \ac{dr} techniques, we develop a robust, low-complexity estimator that works effectively even under strong MC. Beyond estimation, we demonstrate that accurate beamforming remains equally crucial and propose an MC-aware beamforming algorithm based on the Neumann series expansion and SCA optimization framework.

A central finding of this work is the dual role of MC. While typically viewed as a performance-limiting impairment, MC can in fact be leveraged to enhance performance when properly accounted for. Our simulation results reveal that the spectral efficiency is highly sensitive to MC-related RIS design parameters such as the active RIS amplification factor and element spacing. Notably, when the RIS inter-element spacing falls within a moderate range (e.g., $\lambda/16$ to $\lambda/8$ as shown in Fig.~\ref{fig:se_vs_spacing}), MC can contribute positively, enabling significant performance gains. These results suggest that an optimal RIS density and amplification level may exist for MC-aware designs, balancing MC-induced signal distortion with its constructive effects. Overall, this work highlights that MC, when carefully modeled and exploited, is not merely a source of degradation but also a design opportunity for next-generation RIS-assisted communication systems.
}

\appendices

\counterwithin*{equation}{section}
\renewcommand\theequation{\thesection.\arabic{equation}}

 {\section{A Proof of~\eqref{eq:vec_H_BIU_cv}}
\label{sec:vec_Hcv_mI}
This appendix section provides a step-by-step derivation to obtain~\eqref{eq:vec_H_BIU_cv}. By using the matrix equality $\vect{\Am\Bm\Cm}=(\Cm^\TT\otimes\Am)\vect{\Bm}$ and the beamspace representation of~\eqref{eq:ch_TxRIS_angle_domain} and~\eqref{eq:ch_RISRx_angle_domain}, we write $\vect{\Hm_\convidx^{m_\It}}$ as
\begin{align}
    &\vect{\Hm_\convidx^{m_\It}}\!=\!\vect{\!\Am_\Rt(\varthetav)\Sigmam_{\Rt\It}\Am_\It^\TT(\thetav)\Gammam_{m_\It}\Am_\It(\phiv)\Sigmam_{\It\Tt}\Am^\TT_\Tt(\varphiv)\!},\notag\\&=\big(\Am_\Tt(\varphiv)\otimes\Am_\Rt(\varthetav)\big)\vect{\Sigmam_{\Rt\It}\Am_\It^\TT(\thetav)\Gammam_{m_\It}\Am_\It(\phiv)\Sigmam_{\It\Tt}}\notag,\\&\notag=\Am_{\Tt\Rt}(\varphiv,\varthetav)\big(\Sigmam^\TT_{\It\Tt}\otimes\Sigmam_{\Rt\It}\big)\vect{\Am_\It^\TT(\thetav)\Gammam_{m_\It}\Am_\It(\phiv)},\\&=\Am_{\Tt\Rt}(\varphiv,\varthetav)\big(\Sigmam_{\It\Tt}\otimes\Sigmam_{\Rt\It}\big)\big(\Am_\It^\TT(\phiv)\otimes\Am_\It^\TT(\thetav)\big)\vect{\Gammam_{m_\It}}\notag,\\&\label{eq:vec_H_cv_gamma}=\Am_{\Tt\Rt}(\varphiv,\varthetav)\big(\Sigmam_{\It\Tt}\otimes\Sigmam_{\Rt\It}\big)\big(\Am_\It(\phiv)\bullet\Am_\It(\thetav)\big)^\TT\bm{\gamma}_{m_\It}.
\end{align}
where $\Sigmam_{\It\Tt}=\Sigmam^\TT_{\It\Tt}$ since it is a diagonal matrix. Therefore, by substituting  $\Am_{\It\It}^\convidx(\phiv,\thetav)= \Am_\It(\phiv)\bullet\Am_\It(\thetav)$ in \eqref{eq:vec_H_cv_gamma}, we can directly obtain~\eqref{eq:vec_H_BIU_cv}.}
\section{A Proof of~\eqref{eq:H_UIB}}
\label{sec:UL_DL_reciprocity}
This appendix section derives the channel reciprocity, i.e., the relationship between the uplink channel and the downlink channel. First, we can express these two cascaded channels as %\red{since I changed some notation, should I also unify with this or what you prefer?}
\begin{align}
    &\Hm_{\Rt\It} \bar{\Gammam} \Hm_{\It\Tt}
    = \Am_\Rt(\varthetav)\Sigmam_{\Rt\It}\Am_\It^\TT(\thetav) \bar{\Gammam} \Am_\It(\phiv)\Sigmam_{\It\Tt}\Am^\TT_\Tt(\varphiv),\\
    &\Hm_{\Tt\It} \bar{\Gammam} \Hm_{\It\Rt} = \Am_\Tt(\varphiv)\Sigmam_{\It\Tt}\Am^\TT_\It(\phiv) \bar{\Gammam} \Am_\It(\thetav) \Sigmam_{\Rt\It} \Am^\TT_\Rt(\varthetav),
\end{align}
where $\bar{\Gammam} = (\Gammam^{-1} - \Sm)^{-1}$. According to~\eqref{eq:vec_H_BIU_mc}, we have
\begin{align}
&\vect{\Hm_{\Rt\It} \bar{\Gammam} \Hm_{\It\Tt}}\notag\\
&=\big(\Am_\mathrm{U}(\varphiv)\otimes\Am_\mathrm{B}(\varthetav)\big)\big(\Sigmam_{\It\Tt}\otimes\Sigmam_{\Rt\It}\big)\big(\Am_\mathrm{I}^\TT(\phiv)\otimes\Am_\mathrm{I}^\TT(\thetav)\big)\xiv,\notag\\
&=\Big(\underbrace{\big(\Am_\mathrm{U}(\varphiv)\Sigmam_{\It\Tt}\Am_\mathrm{I}^\TT(\phiv)\big)\otimes\big(\Am_\mathrm{B}(\varthetav)\Sigmam_{\Rt\It}\Am_\mathrm{I}^\TT(\thetav)\big)}_{{{\Gm}}_{\mathrm{mc}}^\mathrm{UL}\in\mathbb{C}^{N_\mathrm{U}N_\mathrm{B}\times N_\mathrm{I}^2}}\Big)\xiv, \label{eq:HBIUexp}
\end{align}
where $\xiv=\vect{\bar{\Gammam}}$.
Similarly, we can express the downlink channel as
\begin{align}
&\vect{\Hm_{\Tt\It} \bar{\Gammam} \Hm_{\It\Rt}}\notag\\
&=\big(\Am_\mathrm{B}(\varthetav)\otimes\Am_\mathrm{U}(\varphiv)\big)\big(\Sigmam_{\Rt\It}\otimes\Sigmam_{\It\Tt}\big)\big(\Am_\mathrm{I}^\TT(\thetav)\otimes\Am_\mathrm{I}^\TT(\phiv)\big)\xiv,\notag\\
&=\Big(\underbrace{\big(\Am_\mathrm{B}(\vartheta)\Sigmam_{\Rt\It}\Am_\mathrm{I}^\TT(\thetav)\big) \otimes \big(\Am_\mathrm{U}(\varphiv)\Sigmam_{\It\Tt}\Am_\mathrm{I}^\TT(\phiv)\big) }_{{{\Gm}}_{\mathrm{mc}}^\mathrm{DL}\in\mathbb{C}^{N_\mathrm{U}N_\mathrm{B}\times N_\mathrm{I}^2}}\Big)\xiv. \label{eq:HUIBexp}
\end{align}
By comparing~\eqref{eq:HBIUexp} and~\eqref{eq:HUIBexp}, we observe that the uplink and downlink equivalent channels, ${{\Gm}}_{\mathrm{mc}}^\mathrm{UL}$ and ${{\Gm}}_{\mathrm{mc}}^\mathrm{DL}$, contain the same entries but are arranged in different layouts. Therefore, ${{\Gm}}_{\mathrm{mc}}^\mathrm{DL}$ can be derived directly from ${{\Gm}}_{\mathrm{mc}}^\mathrm{UL}$, and~\eqref{eq:HUIBexp} yields~\eqref{eq:H_UIB}.

\section{Background of Successive Convex Approximation} \label{sec:SCA}
Consider the following general optimization problem:
\begin{equation}
    \min_{\xv}\ f(\xv),\quad\text{s.t.}\ \xv\in\mathcal{X}.  
\end{equation}
The \ac{sca} framework generates a sequence of descending feasible points~$\{\xv_i\}_{i\in\mathbb{N}}$ through a cyclic application of two steps: (i) formulating a simpler surrogate function (which is convex) at the current feasible point, and (ii) optimizing the surrogate function to acquire an improved feasible point.

\subsubsection{Formulating Surrogate Functions}
At each feaible point~$\xv_i$, the \ac{sca} method formulates a surrogate function~$g(\xv|\xv_i)$ satisfying the following two conditions:
\begin{condition}\label{cond:1}
    $g(\xv|\xv_i)$ is convex on~$\mathcal{X}$, $\forall~\xv_i\in\mathcal{X}$.
\end{condition}
\begin{condition}\label{cond:2}
    $g(\xv|\xv_i)$ is differentiable on~$\mathcal{X}$ and its gradient equals to the gradient of~$f(\xv)$ at~$\xv_i$, i.e., $\nabla g(\xv|\xv_i)|_{\xv=\xv_i}=\nabla f(\xv)|_{\xv=\xv_i}$.
\end{condition}

\subsubsection{Optimizing Surrogate Functions}
Once a surrogate function~$g(\xv|\xv_i)$ that satisfies Condition~\ref{cond:1} and Condition~\ref{cond:2} is obtained, we can obtain a descending direction of~$f(\xv)$ by addressing the following convex optimization problem:
\begin{equation}
   \hat{\xv}_{i+1}=\arg\min_{\xv}\  g(\xv|\xv_i),\quad
    \text{s.t.}\ \xv\in\mathcal{X}.  
\end{equation}
Then, we update the feasible point as $\xv_{i+1} = \xv_{i} + \eta_i (\hat{\xv}_{i+1}-\xv_i)$, where~$\eta_i\in\mathbb{R}^+$ is the step size. Various rules for selecting the step size are outlined in~\cite[Assumption 3.6]{Nedic2018Parallel}, and the convergence of this iterative \ac{sca} procedure has been established by~\cite[Theorem~3.7]{Nedic2018Parallel}.

\section{Proof of Proposition~\ref{prop:SCAsuro}}\label{app:SCAsuro}
Both functions,~$f(\bm{\gamma})$ and~$g(\bm{\gamma}|\bm{\gamma}_i)$, are a mapping:~$\mathbb{C}^{N_\It}\rightarrow\mathbb{R}$. Following the complex-value differentiation theory in~\cite{Hjorungnes2007Complex,Hjorungnes2011Complex}, we rewrite the two functions~$f(\bm{\gamma})$ and~$g(\bm{\gamma}|\bm{\gamma}_i)$ as~$f(\bm{\gamma},\bm{\gamma}^*)$ and~$g(\bm{\gamma},\bm{\gamma}^*|\bm{\gamma}_i)$, respectively. Here,~$\bm{\gamma}$ and~$\bm{\gamma}^*$ are treated as two linearly independent variables.

According to~\eqref{eq:g} and~\eqref{eq:Bi}, we have
\vspace{-0.3em}
\begin{align}
    &g(\bm{\gamma},\bm{\gamma}^*|\bm{\gamma}_i) = -\bm{\gamma}^\HH\Qm_i\bm{\gamma} + K(\bm{\gamma}\!-\!\bm{\gamma}_i)^\HH(\bm{\gamma}\!-\!\bm{\gamma}_i) - \bm{\gamma}_i^\HH\qv\bm{\gamma}^\TT\Bm\bm{\gamma}_i \notag\\ &\!-\! \bm{\gamma}_i^\HH\Bm^\HH\bm{\gamma}^*\qv^\HH\bm{\gamma}_i \!-\! \bm{\gamma}_i^\HH\Bm^\HH\bm{\gamma}^*\!\bm{\gamma}_i^\TT\Bm\bm{\gamma}_i \!-\! \bm{\gamma}_i^\HH\Bm^\HH\bm{\gamma}_i^*\!\bm{\gamma}^\TT\Bm\bm{\gamma}_i.
\end{align}
We can then observe that 
\begin{align}
    \frac{\partial^2 g(\bm{\gamma},\bm{\gamma}^*|\bm{\gamma}_i)}{\partial \bm{\gamma}\partial \bm{\gamma}} &= \frac{\partial^2 g(\bm{\gamma},\bm{\gamma}^*|\bm{\gamma}_i)}{\partial \bm{\gamma}^*\partial \bm{\gamma}^*} = \mathbf{0},\\
    \frac{\partial^2 g(\bm{\gamma},\bm{\gamma}^*|\bm{\gamma}_i)}{\partial \bm{\gamma} \partial \bm{\gamma}^*} &= \Big(\frac{\partial^2 g(\bm{\gamma},\bm{\gamma}^*|\bm{\gamma}_i)}{\partial \bm{\gamma}^*\partial \bm{\gamma}}\Big)^\TT = K\mathbf{I}-\Qm_i.
\end{align}
Thus, the full complex Hessian matrix of~$g(\bm{\gamma},\bm{\gamma}^*|\bm{\gamma}_i)$ is~\cite{Zhang2017Matrix}
\begin{equation}
    \Hm_{\bm{\gamma},\bm{\gamma}^*}^i 
    = \begin{bmatrix}
    (K\mathbf{I}-\Qm_i)^\TT & \mathbf{0} \\
    \mathbf{0} & K\mathbf{I}-\Qm_i
    \end{bmatrix}\succeq \mathbf{0}.
\end{equation}
Here, the positive semi-definiteness of~$\Hm_{\bm{\gamma},\bm{\gamma}^*}^i$ is ensured by~\eqref{eq:KBi}. Then, according to the second-order condition of convex functions~\cite[Sec. 3.1.4]{Boyd2004Convex},~$g(\bm{\gamma},\bm{\gamma}^*|\bm{\gamma}_i)$ is convex. Hence, Condition~\ref{cond:1} in Appendix~\ref{sec:SCA} is satisfied.

Next, we derive the first-order derivatives of~$f(\bm{\gamma},\bm{\gamma}^*)$  and~$g(\bm{\gamma},\bm{\gamma}^*|\bm{\gamma}_i)$  to assess whether they satisfy Condition~\ref{cond:2}. Based on~\eqref{eq:g}, we can derive:%~\eqref{eq:fgammav} and
\begin{multline}\label{eq:dfdgamma}
    \frac{\partial f(\bm{\gamma},\bm{\gamma}^*)}{\partial\bm{\gamma}} = -(\bm{\gamma}^\HH\qv)\qv^*-(\bm{\gamma}^\HH\qv)(\Bm+\Bm^\TT)\bm{\gamma} - (\bm{\gamma}^\HH\Bm^\HH\bm{\gamma}^*)\qv^* \\- (\bm{\gamma}^\HH\Bm^\HH\bm{\gamma}^*)(\Bm+\Bm^\TT)\bm{\gamma},
\end{multline}
\vspace{-2em}
\begin{multline}\label{eq:dfdgammac}
    \frac{\partial f(\bm{\gamma},\bm{\gamma}^*)}{\partial\bm{\gamma}^*} = -(\bm{\gamma}^\TT\qv^*)\qv-(\bm{\gamma}^\TT\qv^*)(\Bm^\HH+\Bm^*)\bm{\gamma}^* - (\bm{\gamma}^\TT\Bm\bm{\gamma})\qv \\- (\bm{\gamma}^\TT\Bm\bm{\gamma})(\Bm^\HH+\Bm^*)\bm{\gamma}^*,
\end{multline}
\vspace{-2em}
\begin{align}\label{eq:dgdgamma}
    &\frac{\partial g(\bm{\gamma},\!\bm{\gamma}^*|\bm{\gamma}_i)}{\partial\bm{\gamma}}\! =\! -(\bm{\gamma}^\HH\!\qv)\qv^*\!\!\!-\!(\bm{\gamma}^\HH\!\qv)\Bm^\TT\!\bm{\gamma}_i\!\! -\! (\bm{\gamma}^\HH\Bm^\HH\!\bm{\gamma}_i^*)\qv^* \\ &- \!(\bm{\gamma}^\HH\Bm^\HH\!\bm{\gamma}_i^*)\Bm^{\!\TT}\!\!\bm{\gamma}_i\! +\! K(\bm{\gamma}^*\!\!\!-\!\bm{\gamma}_i^*)\! -\! (\bm{\gamma}_i^\HH\qv)\Bm\bm{\gamma}_i\! -\! (\bm{\gamma}_i^\HH\Bm^\HH\!\bm{\gamma}_i^*)\Bm\bm{\gamma}_i,\notag
\end{align}
\vspace{-2em}
\begin{align}\label{eq:dgdgammac}
     &\frac{\partial g(\bm{\gamma},\!\bm{\gamma}^*|\bm{\gamma}_i)}{\partial\bm{\gamma}^*}\! =\! -(\bm{\gamma}^\TT\!\qv^*)\qv\!-\!(\bm{\gamma}^\TT\!\qv^*)\Bm^\HH\bm{\gamma}_i\! -\! (\bm{\gamma}_i^\TT\Bm\bm{\gamma})\qv \\ & - \!(\bm{\gamma}_i^\TT\Bm\bm{\gamma})\Bm^\HH\!\bm{\gamma}_i^*\! +\! K(\bm{\gamma}\!-\!\bm{\gamma}_i)\! -\! (\bm{\gamma}_i^\TT\qv^*)\Bm^*\bm{\gamma}_i^*\! -\! (\bm{\gamma}_i^\TT\Bm\bm{\gamma}_i)\Bm^*\bm{\gamma}_i^*.\notag
\end{align}
By comparing~\eqref{eq:dfdgamma} and~\eqref{eq:dgdgamma}, we can see that~$
    \frac{\partial f(\bm{\gamma},\bm{\gamma}^*)}{\partial\bm{\gamma}}|_{\bm{\gamma}=\bm{\gamma}_i} = \frac{\partial g(\bm{\gamma},\bm{\gamma}^*|\bm{\gamma}_i)}{\partial\bm{\gamma}}|_{\bm{\gamma}=\bm{\gamma}_i}.
$ 
By comparing~\eqref{eq:dfdgammac} and~\eqref{eq:dgdgammac}, we can see that
$
    \frac{\partial f(\bm{\gamma},\bm{\gamma}^*)}{\partial\bm{\gamma}^*}|_{\bm{\gamma}=\bm{\gamma}_i} = \frac{\partial g(\bm{\gamma},\bm{\gamma}^*|\bm{\gamma}_i)}{\partial\bm{\gamma}^*}|_{\bm{\gamma}=\bm{\gamma}_i}.
$ 
Thus Condition~\ref{cond:2} in Appendix~\ref{sec:SCA} is satisfied, which concludes the proof.

\bibliography{abbrev,references}
\bibliographystyle{IEEEtran}

\end{document}